\documentclass{article}

\usepackage{amsmath,amsfonts,amssymb,amsthm,hyperref,enumerate,multicol,bm}
\usepackage{calligra,indentfirst,epsfig,lettrine}
\usepackage{caption}
\usepackage{hyperref}
\usepackage{relsize}
\usepackage[format=hang]{caption}
\usepackage{tabls}
\usepackage{array}
\usepackage{longtable}
\calligra

\usepackage{comment}
\newcommand{\Mod}[1]{\ (\mathrm{mod}\ #1)}

\newcommand{\etal}{\textit{et al.}}
\newcommand{\qbin}[3]{\genfrac{[}{]}{0pt}{}{#1}{#2}_{#3}}
\makeatletter
\newcommand*{\rom}[1]{\expandafter\@slowromancap\romannumeral #1@}
\makeatother
\usepackage{mathtools}
\usepackage{graphicx}
\usepackage{subfigure}
\usepackage{xcolor}
\usepackage{tensor}
\usepackage{color, soul}
\usepackage{BOONDOX-uprscr}

\DeclareMathAlphabet{\mathpzc}{OT1}{pzc}{m}{it}

\usepackage{upgreek}

\allowdisplaybreaks
\numberwithin{equation}{section}

\newtheorem{thm}{Theorem}[section]

\newtheorem{prop}{Proposition}[section]
\newtheorem{cor}{Corollary}[section]
\newtheorem{lemma}{Lemma}[section]
\newtheorem{example}{Example}[section]
\newtheorem{remark}{Remark}[section]

\begin{document}

\title{On Galois LCD codes and LCPs of codes
\\ over mixed alphabets}

\author{Leijo Jose{\footnote{Email address:~\urlstyle{same}\href{mailto:leijoj@iiitd.ac.in}{leijoj@iiitd.ac.in}}} ~ and ~Anuradha Sharma{\footnote{Corresponding Author, Email address:~\urlstyle{same}\href{mailto:anuradha@iiitd.ac.in}{anuradha@iiitd.ac.in}} }  \\
 {Department of  Mathematics, IIIT-Delhi}\\{New Delhi 110020, India}}
\date{}
\maketitle
\date{}

	\maketitle
\begin{abstract}\label{Abstract}
Let $\mathtt{R}$ be a finite commutative chain ring with the maximal ideal $\gamma\mathtt{R}$ of nilpotency index $e\geq 2,$ and let $\check{\mathtt{R}}=\mathtt{R}/\gamma^{s}\mathtt{R}$ for some positive integer $ s< e.$ In this paper, we study and characterize Galois $\mathtt{R}\check{\mathtt{R}}$-LCD codes of an arbitrary block-length.  We  show that each weakly-free $\mathtt{R}\check{\mathtt{R}}$-linear code is monomially equivalent to a Galois $\mathtt{R}\check{\mathtt{R}}$-LCD code when $|\mathtt{R}/\gamma\mathtt{R}|>4,$ while it is monomially equivalent to a   Euclidean $\mathtt{R}\check{\mathtt{R}}$-LCD code when  $|\mathtt{R}/\gamma\mathtt{R}|>3.$  We also  obtain enumeration formulae for all Euclidean and Hermitian $\mathtt{R}\check{\mathtt{R}}$-LCD codes of an arbitrary block-length.   With the help of these enumeration formulae, we classify all Euclidean $\mathbb{Z}_4 \mathbb{Z}_{2}$-LCD codes   and $\mathbb{Z}_9 \mathbb{Z}_{3}$-LCD codes of block-lengths $(1,1),$ $(1,2),$ $(2,1),$ $(2,2),$ $(3,1)$ and $(3,2)$  and all Hermitian $\frac{\mathbb{F}_{4}[u]}{\langle u^2\rangle} \;\mathbb{F}_{4}$-LCD codes of block-lengths $(1,1),$ $(1,2),$ $(2,1)$ and $(2,2)$ up to monomial equivalence. Apart from this, we  study and characterize LCPs of 
$\mathtt{R}\check{\mathtt{R}}$-linear codes.  We  further study a direct sum masking scheme constructed using  LCPs of $\mathtt{R}\check{\mathtt{R}}$-linear codes and obtain its security threshold against fault injection and side-channel attacks.  We also discuss another application of LCPs of $\mathtt{R}\check{\mathtt{R}}$-linear codes in coding for the noiseless two-user adder channel.
\end{abstract}

\section{Introduction}
Linear codes with complementary duals (or LCD codes) over finite fields are  introduced  by Massey \cite{Massey}. He provided an algebraic characterization of LCD codes and showed that asymptotically good LCD codes exist. He also showed that binary LCD codes provide an optimum linear coding solution for the two-user binary adder channel. Later, Bringer \etal~\cite{Bringer} showed that LCD codes play a significant role in counter-measures to passive and active side-channel analyses on embedded cryptosystems. Since then,  the study of LCD codes has attracted much attention \cite{Harada,Carlet1}. Recently, Liu~\etal~\cite{Liu2018} studied Galois LCD codes over finite fields. They characterized the generator polynomials and defining sets of Galois LCD constacyclic codes and derived a necessary and sufficient condition for a constacyclic code to be Galois LCD. They also identified two classes of Galois LCD   maximum distance separable (MDS) codes and a class of Hermitian LCD MDS codes within the family of constacyclic codes. 
Liu and Liu  \cite{Liu}  studied Euclidean  LCD codes over finite commutative chain rings.   The problem of determining   an enumeration formula for  LCD codes over finite commutative chain rings has recently gathered significant interest, as it is useful in classifying these codes up to monomial equivalence \cite{Harada,Yadav1}. 
 
 In another related direction, Ngo \etal~\cite{Ngo} introduced linear complementary pairs (LCPs) of codes over finite fields. They also studied a direct sum masking scheme using LCPs of codes over finite fields  and obtained its security threshold against fault injection and side-channel attacks.  Later, Carlet~\etal~\cite{Carlet2018} studied and characterized LCPs of constacyclic and quasi-cyclic codes. They proved that for an LCP $(C,D)$ of constacyclic codes,  the code $C$ and the Euclidean dual code of $D$ are monomially equivalent.  They also observed that this result holds for a special class of quasi-cyclic codes, \textit{viz.} 2D cyclic codes, and not for quasi-cyclic codes in general. They also obtained a linear programming bound for binary LCPs of codes.  
 In a recent work, Hu and Liu \cite{Hu2021} studied LCP of codes over finite commutative rings. They also characterized LCPs of  constacyclic and quasi-cyclic codes over finite chain rings. In the same work, they gave a characterization of LCPs of codes over  finite principal ideal rings. Under suitable conditions, they  established a judging criterion for pairs of cyclic codes over the principal ideal ring $\mathbb{Z}_k$ to form an LCP and obtained an MDS LCP of cyclic codes over $\mathbb{Z}_{15}.$

Recently, Benbelkacem \etal~\cite{Benbel} generalized the notion of LCD codes  and studied Euclidean complementary dual  codes in the ambient space with mixed binary and quaternary alphabets, which we will refer to as Euclidean $\mathbb{Z}_4\mathbb{Z}_2$-linear  codes with complementary duals (or Euclidean $\mathbb{Z}_4\mathbb{Z}_2$-LCD codes). They provided methods to construct binary LCD codes by studying Gray images of  Euclidean $\mathbb{Z}_4\mathbb{Z}_2$-LCD codes. They also obtained several distance-optimal binary LCD codes as Gray images of Euclidean $\mathbb{Z}_4\mathbb{Z}_2$-LCD codes.  Rif\`{a} and Pujol \cite{Rifa} first introduced and studied $\mathbb{Z}_4\mathbb{Z}_2$-linear codes  as abelian translation-invariant propelinear codes. Later, Borges \etal~\cite{Borges1} thoroughly investigated $\mathbb{Z}_4\mathbb{Z}_2$-linear codes. They obtained the fundamental parameters of these codes and standard forms of their generator and parity-check matrices. They also studied  binary Gray images of $\mathbb{Z}_4\mathbb{Z}_2$-linear codes.  Since then, there has been an enormous interest in the study of $\mathbb{Z}_4\mathbb{Z}_2$-linear codes and their generalizations \cite{Siap3,Siap4,Borges}. Recently, Liu and Hu \cite{Liu2024} studied LCPs of $\mathbb{Z}_4\mathbb{Z}_2$-linear codes and  derived a necessary and sufficient condition under which two $\mathbb{Z}_4\mathbb{Z}_2$-linear codes form an LCP. They also studied their binary Gray images and  constructed binary LCPs of codes in certain special cases. They further  discussed an application of LCPs of $\mathbb{Z}_4\mathbb{Z}_2$-linear codes to the noiseless two-user adder channel.

In a related work \cite{Siapz2}, Aydogdu and Siap  generalized $\mathbb{Z}_4\mathbb{Z}_2$-linear codes to $\mathbb{Z}_{2^s}\mathbb{Z}_{2}$-linear codes and obtained  standard forms of their generator and parity-check matrices, where $s \geq 2$ is an integer. They also obtained two upper bounds on  Lee distances of $\mathbb{Z}_{2^s}\mathbb{Z}_{2}$-linear codes. Later, Aydogdu and Siap \cite{Siap1}  generalized $\mathbb{Z}_{2^s}\mathbb{Z}_{2}$-linear codes to $\mathbb{Z}_{p^r}\mathbb{Z}_{p^s}$-linear codes and obtained standard forms of their generator and parity-check matrices, where  $p$ is a prime number and $r,s$ are positive integers with $r > s.$  They also obtained two upper bounds on their Lee distances. 

Now, let $\mathtt{R}$ be a finite commutative chain ring with the maximal ideal $\gamma\mathtt{R}$ of nilpotency index $e\geq 2,$ and let $\check{\mathtt{R}}=\mathtt{R}/\gamma^{s}\mathtt{R}$ be the quotient ring, where $s $ is a positive integer satisfying $s<e.$   For positive integers $\mathpzc{a}$ and $\mathpzc{b},$ the direct sum $\mathtt{R}^{\mathpzc{a}}\oplus \check{\mathtt{R}}^{\mathpzc{b}}$ can be viewed as an $\mathtt{R}$-module. We will refer to $\mathtt{R}$-submodules of $\mathtt{R}^{\mathpzc{a}}\oplus \check{\mathtt{R}}^{\mathpzc{b}}$ as $\mathtt{R}\check{\mathtt{R}}$-linear codes of block-length $(\mathpzc{a},\mathpzc{b}).$ Borges \etal~\cite{Borges} studied $\mathtt{R}\check{\mathtt{R}}$-linear and cyclic  $\mathtt{R}\check{\mathtt{R}}$-linear codes and their Euclidean dual codes.   In another related work, Bajalan \etal~\cite{Bajalan1} determined a standard form of the parity-check matrix for  $\mathtt{R}\check{\mathtt{R}}$-linear codes and 
described the structure of Euclidean self-dual $\mathtt{R}\check{\mathtt{R}}$-linear codes. By defining a distance-preserving Gray map from the ambient space with mixed alphabets $\mathtt{R}$ and $\check{\mathtt{R}}$ into a Hamming space over a finite
field, they obtained two upper bounds on homogeneous distances of $\mathtt{R}\check{\mathtt{R}}$-linear codes. 
Recently, Bajalan \etal~\cite{Bajalan}  characterized  Galois $\mathtt{R}\check{\mathtt{R}}$-LCD codes  under a certain constraint (see \cite[Th. 1]{Bajalan}) and Galois-invariant $\mathtt{R}\check{\mathtt{R}}$-linear codes. They also derived a generalized Delsarte's Theorem for $\mathtt{R}\check{\mathtt{R}}$-linear codes.  Further, by placing a Gray map from an ambient space with mixed alphabets of $\mathtt{R}$ and $\check{\mathtt{R}}$ into a Hamming space over a finite field \cite{Jitman}, they studied  Gray images of  $\mathtt{R}\check{\mathtt{R}}$-linear  codes  and established conditions under which the Gray image of a Galois $\mathtt{R}\check{\mathtt{R}}$-LCD code is a Galois LCD code in certain special cases.  In a recent work, Bajalan \etal~\cite{Bajalan2023} studied LCPs of $\mathtt{R}\check{\mathtt{R}}$-linear codes. They also  provided a characterization of LCPs of $\mathtt{R}\check{\mathtt{R}}$-linear codes, which extends Theorem 3.19 of Liu and Hu \cite{Liu2024}. Motivated by the works mentioned above, we will study and characterize Galois $\mathtt{R}\check{\mathtt{R}}$-LCD codes  and  obtain explicit enumeration formulae for  Euclidean and Hermitian $\mathtt{R}\check{\mathtt{R}}$-LCD codes. With the help of these enumeration formulae and using Magma, we will classify these codes up to monomial equivalence in certain specific cases. Besides this, we will study LCPs of 
$\mathtt{R}\check{\mathtt{R}}$-linear codes and provide a characterization of LCPs of $\mathtt{R}\check{\mathtt{R}}$-linear codes, which is different from the one derived in Theorem 1 of Bajalan \etal~\cite{Bajalan2023} and Theorem 3.19 of Liu and Hu \cite{Liu2024}. We will also study a   direct sum masking scheme using LCPs of $\mathtt{R}\check{\mathtt{R}}$-linear codes and obtain its security threshold against fault injection and side-channel attacks. We will  discuss another application of LCPs of $\mathtt{R}\check{\mathtt{R}}$-linear codes to the noiseless two-user adder channel.

 This paper is organized as follows: In Section \ref{Preliminaries}, we state  some preliminaries needed to prove our main results. 
In Section \ref{Section3}, we first observe that the constraint imposed  in \cite[Th. 1]{Bajalan} is not necessary for the existence of a Galois $\mathtt{R}\check{\mathtt{R}}$-LCD code (see Example \ref{Example4.2}). Without assuming the constraint imposed in \cite[Th. 1]{Bajalan}, we further show that Theorem  1 of Bajalan \textit{et al.} \cite{Bajalan}  holds (Theorems \ref{Theorem4.1}-\ref{Theorem4.2} and Remarks \ref{Rk1}-\ref{Rk3}).  We also show that every weakly-free $\mathtt{R}\check{\mathtt{R}}$-linear code  of block-length $(\mathpzc{a},\mathpzc{b})$ is monomially equivalent to a Galois $\mathtt{R}\check{\mathtt{R}}$-LCD code  when $|\mathtt{R}/\gamma \mathtt{R}|>4$ (Corollary \ref{corH}). We further show that every weakly-free $\mathtt{R}\check{\mathtt{R}}$-linear code  of block-length $(\mathpzc{a},\mathpzc{b})$ is monomially equivalent to a Euclidean $\mathtt{R}\check{\mathtt{R}}$-LCD code  when $|\mathtt{R}/\gamma \mathtt{R}|>3$ (Corollary \ref{Eucequivalence}). In Section \ref{SectionE}, we obtain an enumeration formula for Euclidean $\mathtt{R}\check{\mathtt{R}}$-LCD codes of block-length $(\mathpzc{a},\mathpzc{b})$ (Theorem \ref{Theorem4.4}). In Section \ref{SectionH}, we obtain an  enumeration formula for  Hermitian $\mathtt{R}\check{\mathtt{R}}$-LCD codes of block-length $(\mathpzc{a},\mathpzc{b})$ (Theorem \ref{Theorem4.8}).  With the help of these enumeration formulae and using Magma \cite{Magma}, we classify all Euclidean $\mathbb{Z}_4 \mathbb{Z}_{2}$-LCD codes  and $\mathbb{Z}_9 \mathbb{Z}_{3}$-LCD codes of block-lengths $(1,1),$ $(1,2),$ $(2,1),$ $(2,2),$ $(3,1)$ and $(3,2),$ and all Hermitian $\frac{\mathbb{F}_{4}[u]}{\langle u^2\rangle} \;\mathbb{F}_{4}$-LCD codes of block-lengths $(1,1),$ $(1,2),$ $(2,1)$ and $(2,2)$ up to monomial equivalence (Tables \ref{table1}-\ref{table3}).   Generator matrices and Lee distances of all monomially inequivalent codes belonging to above-mentioned classes of $\mathtt{R}\check{\mathtt{R}}$-linear codes are provided  in Appendices A-C. In Section \ref{LCP}, we study linear complementary pairs (LCPs) of $\mathtt{R}\check{\mathtt{R}}$-linear codes and derive a necessary and sufficient condition under which a pair of $\mathtt{R}\check{\mathtt{R}}$-linear codes  form an LCP (Theorem \ref{Thm6.3}). It is worth noting that Theorem \ref{Thm6.3} provides a characterization of LCPs of $\mathtt{R}\check{\mathtt{R}}$-linear codes, which is different from the one derived in Theorem 1 of Bajalan \etal~\cite{Bajalan2023} and Theorem 3.19 of Liu and Hu \cite{Liu2024}. In Section \ref{DSM}, we study a direct sum masking scheme based on LCPs of $\mathtt{R}\check{\mathtt{R}}$-linear codes and obtain its security threshold against fault injection and side-channel attacks (Theorem \ref{Thm6.04}). We also provide an example, where the direct sum masking scheme constructed using an LCP of non-separable $\mathtt{R}\check{\mathtt{R}}$-linear codes outperforms its counterparts constructed using the corresponding LCPs of separable $\mathtt{R}\check{\mathtt{R}}$-linear codes and linear codes over $\mathtt{R}$ in terms of the security threshold (see Example \ref{Eg.6.1}).  In Section \ref{tuser}, we discuss an application of LCPs of $\mathtt{R}\check{\mathtt{R}}$-linear codes to the noiseless two-user adder channel, which extends Section 6 of Liu and Hu \cite{Liu2024}.   In Section \ref{Conclusion}, we mention a brief conclusion and discuss some future research directions. 

\section{Some preliminaries}\label{Preliminaries} In this section, we will state some preliminaries needed to derive our main results. For this, we assume, throughout this paper, that $\mathtt{S}$ and $\mathtt{R}$ are finite commutative chain rings such that $\mathtt{R}$ is the Galois extension of $\mathtt{S}$ of degree $w\geq 2,$ \textit{i.e.,} there exists a monic basic irreducible polynomial $f(x)$  of degree $w$ over $\mathtt{S}$  such that $\mathtt{R}\simeq\mathtt{S}[x]/\langle f(x)\rangle$ \cite{MacDonald}. We further assume that the chain ring $\mathtt{R}$ has maximal ideal $\gamma \mathtt{R}$ of nilpotency index $e \geq 2,$  \textit{i.e.,} $e \geq 2$ is the smallest integer satisfying $\gamma^e=0.$

Now, let $\mathscr{M}$ be a finite $\mathtt{R}$-module. A non-empty subset $\{v_1,v_2,\ldots,v_n\}$ of $\mathscr{M}$ is said to be independent over $\mathtt{R}$ if there exist $r_1, r_2,\ldots,r_n\in\mathtt{R}$ satisfying  $r_1v_1+r_2v_2+\cdots+r_nv_n=0,$ then we must have $r_iv_i=0$ for $1\leq i\leq n.$ A non-empty  subset $\{v_1,v_2,\ldots,v_n\}$ of $\mathscr{M}$ is called a basis of $\mathscr{M}$ if the set $\{v_1,v_2,\ldots,v_n\}$ spans $\mathscr{M}$ and is independent over $\mathtt{R}.$   Further, we say that a non-zero vector $v\in\mathscr{M}$ has period $\gamma^i$ if it satisfies $\gamma^iv=0$ and $\gamma^{i-1}v\neq0,$ where $1\leq i\leq e.$ 
\begin{lemma}\label{Lemma4.5}
Let $\mathscr{M}$ be an $\mathtt{R}$-module, and let $v_1,v_2\in\mathscr{M}$ be such that the period of $v_1$ is $\gamma^i$ and the period of $v_2$ is $\gamma^j,$ where $1\leq j\leq i\leq e.$ The following hold.
\begin{enumerate}
    \item[(a)] The set $\{v_1,v_2\}$ is independent over $\mathtt{R}$ if and only if the set $\{v_1-r_1v_2,v_2\}$ is independent over $\mathtt{R}$ for all $r_1\in\mathtt{R}.$
    \item[(b)] The set $\{v_1,v_2\}$ is independent over $\mathtt{R}$ if and only if the set $\{v_1,v_2-r_2v_1\}$ is independent over $\mathtt{R}$ for all $r_2\in\gamma^{i-j}\mathtt{R}.$
    \item[(c)] If $i=j$ and the set $\{v_1,v_2\}$ is  independent over $\mathtt{R},$ then  for all $w_1,w_2\in\mathscr{M},$ the set $\{v_1+\gamma^{e-i+1}w_1,v_2+\gamma^{e-i+1}w_2\}$ is also independent over $\mathtt{R}$ and  both  the elements $v_1+\gamma^{e-i+1}w_1$ and $v_2+\gamma^{e-i+1}w_2$ have the same period $\gamma^{i}.$  The converse also holds.
\end{enumerate}
\end{lemma}
\begin{proof} Its proof is a straightforward exercise.
\end{proof}
From now on, let $\check{\mathtt{R}}$ denote the quotient ring $\mathtt{R}/\gamma^{s}\mathtt{R}$ for some positive integer $s < e.$  Clearly, $\check{\mathtt{R}}$ is also a chain ring with the maximal ideal generated by the element $\gamma+\gamma^{s}\mathtt{R}\in \check{\mathtt{R}}$ of nilpotency index $s.$ One can easily see that the  ring  $\check{\mathtt{R}}$ can be embedded into $\mathtt{R}.$ So for our convenience,   we will identify each element $a+\gamma^{s}\mathtt{R} \in \check{\mathtt{R}}$ with its representative $a \in \mathtt{R},$ and we will perform addition and multiplication in $\check{\mathtt{R}}$ modulo $\gamma^{s}.$   Under this identification, we assume that the  Teichm$\ddot{u}$ller sets of $\mathtt{R}$ and $\check{\mathtt{R}}$ are equal, which we will denote by $\mathcal{T}.$  Similarly, we also assume that the residue fields of $\mathtt{R}$ and $\check{\mathtt{R}}$ are equal, which we will denote by $\overline{\mathtt{R}}.$
Further, let  $\mathpzc{a}$ and $\mathpzc{b}$ be fixed positive integers satisfying $\mathpzc{a}+\mathpzc{b}=\mathpzc{n},$ and  let us define the ring $\mathbf{M}=\mathtt{R}^\mathpzc{a}\oplus \check{\mathtt{{R}}}^\mathpzc{b}.$ Note that each element $m \in \mathbf{M}$ can be uniquely expressed as an $\mathpzc{n}$-tuple $m=(\mathtt{c}_1,\mathtt{c}_2,\ldots,\mathtt{c}_\mathpzc{a}|\mathtt{d}_{1},\mathtt{d}_{2},\ldots,\mathtt{d}_\mathpzc{b})\in\mathbf{M},$ where $\mathtt{c}_i\in\mathtt{R}$ for $1\leq i\leq\mathpzc{a}$ and $\mathtt{d}_j\in\check{\mathtt{R}}$ for $1\leq j\leq\mathpzc{b}.$ The $\mathpzc{a}$-tuple $(\mathtt{c}_1,\mathtt{c}_2,\ldots,\mathtt{c}_\mathpzc{a})\in \mathtt{R}^{\mathpzc{a}}$ is called the first block of $m$ and the $\mathpzc{b}$-tuple $(\mathtt{d}_{1},\mathtt{d}_{2},\ldots,\mathtt{d}_\mathpzc{b})\in \check{\mathtt{R}}^{\mathpzc{b}}$ is called the second block of $m.$ One can easily observe that the ring $\mathbf{M}$ can be viewed as an $\mathtt{R}$-module under the component-wise addition and the component-wise scalar multiplication, defined as follows:
\vspace{-2mm}\begin{eqnarray*} m_1+m_2 &=&(\mathtt{c}_1+\mathtt{c}'_1,\mathtt{c}_2+\mathtt{c}'_2,\ldots,\mathtt{c}_\mathpzc{a}+\mathtt{c}'_\mathpzc{a}|
\mathtt{d}_{1}+\mathtt{d}'_{1},\mathtt{d}_{2}+\mathtt{d}'_{2},\ldots,\mathtt{d}_\mathpzc{b}+\mathtt{d}'_\mathpzc{b})\text{ ~and}\\ \mathtt{r}m_1 & = & (\mathtt{r}\mathtt{c}_1,\mathtt{r}\mathtt{c}_2,\ldots,\mathtt{r}\mathtt{c}_\mathpzc{a}|\mathtt{r}\mathtt{d}_{1},\mathtt{r}\mathtt{d}_{2},\ldots,\mathtt{r}\mathtt{d}_\mathpzc{b})\vspace{-2mm}\end{eqnarray*}
for all $m_1=(\mathtt{c}_1,\mathtt{c}_2,\ldots,\mathtt{c}_\mathpzc{a}|\mathtt{d}_{1},\mathtt{d}_{2},\ldots,\mathtt{d}_\mathpzc{b})$ and $m_2=(\mathtt{c}'_1,\mathtt{c}'_2,\ldots,\mathtt{c}'_\mathpzc{a}|\mathtt{d}'_{1},\mathtt{d}'_{2},\ldots,\mathtt{d}'_\mathpzc{b})$ in $\mathbf{M}$ and  $\mathtt{r}\in\mathtt{R},$ where $\mathtt{c}_{i}+\mathtt{c}'_{i}$ and $\mathtt{r}\mathtt{c}_{i}$ are to be computed modulo $\gamma^{e}$ for $1 \leq i \leq \mathpzc{a},$ and $\mathtt{d}_{j}+\mathtt{d}'_{j}$ and $\mathtt{r}\mathtt{d}_{j}$ are to be computed modulo $\gamma^{s}$ for $1 \leq j \leq \mathpzc{b}.$

Now, a non-empty  subset $\mathtt{C}$ of $\mathbf{M}$ is called an $\mathtt{R}\check{\mathtt{R}}$-linear code of block-length $(\mathpzc{a},\mathpzc{b})$ if it is an $\mathtt{R}$-submodule of  $\mathbf{M}.$ Further, an $\mathtt{R}\check{\mathtt{R}}$-linear code $\mathtt{C}$ of block-length $(\mathpzc{a},\mathpzc{b})$ is said to be weakly-free  (see Remark 1 of Bajalan \etal~\cite{Bajalan2023}) if it has a generator matrix of the form 
\vspace{-1mm}\begin{equation}\label{weaklyfree}
\left[\begin{array}{c|c}
    \mathcal{A}&\mathcal{B} \\ 
    \gamma^{e-s}\mathcal{E}& \mathcal{F} 
\end{array}\right],\vspace{-1mm}
\end{equation}
where  $k_0$ and $\ell_0$ are non-negative integers, rows of the  matrix $\mathcal{A}\in M_{k_0\times \mathpzc{a}}(\mathtt{R})$ are independent over $\mathtt{R}$ and have period $\gamma^e,$ rows of the matrix $\mathcal{F}\in M_{\ell_0\times \mathpzc{b}}(\check{\mathtt{R}})$ are independent over $\check{\mathtt{R}}$ and have period $\gamma^s,$ and the matrices $ \mathcal{B}\in M_{k_0\times \mathpzc{b}}(\check{\mathtt{R}})$ and  $\gamma^{e-s}\mathcal{E}\in M_{\ell_0\times \mathpzc{a}}(\mathtt{R}),$ (throughout this paper,  $M_{m\times n}(Y)$ denotes the set of all $m \times n$ matrices over $Y$). Further, a weakly-free $\mathtt{R}\check{\mathtt{R}}$-linear code of block-length $(\mathpzc{a},\mathpzc{b})$ is said to be of the type $\{k_0,\underbrace{0,0,\ldots, 0}_{e-1 \text{ times}}; \ell_0,\underbrace{0,0, \ldots, 0}_{s-1 \text{ times}}\}$ if it has a generator matrix of the form \eqref{weaklyfree}. One can easily observe that the $\mathtt{R}$-module $\mathbf{M}$ is a weakly-free $\mathtt{R}\check{\mathtt{R}}$-linear code of the type 
$\{\mathpzc{a},\underbrace{0,0,\ldots, 0}_{e-1 \text{ times}}; \mathpzc{b},\underbrace{0,0, \ldots, 0}_{s-1 \text{ times}}\}.$

Two $\mathtt{R}\check{\mathtt{R}}$-linear codes $\mathtt{C}$ and $\mathtt{D}$ of block-length $(\mathpzc{a},\mathpzc{b})$ are said to be monomially equivalent if one code can be obtained from the other by
a combination of operations of the following four types:
\begin{itemize}
\vspace{-2mm}\item[A.]  Permutation of the $\mathpzc{a}$ coordinate positions in the first block of the code.
\vspace{-2mm}\item[B.]  Permutation of the $\mathpzc{b}$ coordinate positions in the second block of the code.
\vspace{-2mm}\item[C.] Multiplication of the code symbols appearing in a given coordinate position of the first block by a unit in  $\mathtt{R}$. 
\vspace{-2mm}\item[D.] Multiplication of the code symbols appearing in a given coordinate position of the second block by a unit in $\check{\mathtt{R}}$.
\vspace{-2mm}\end{itemize}
Otherwise, the codes $\mathtt{C}$ and $\mathtt{D}$ are said to be monomially inequivalent. Two $\mathtt{R}\check{\mathtt{R}}$-linear codes of block-length $(\mathpzc{a},\mathpzc{b})$ are said to be permutation equivalent if one code can be obtained from the other by a combination of  operations of the Types A and B as defined above.  For non-negative integers $k_0$ and $\ell_0,$  we see, by Proposition 3.2 of Borges \etal~\cite{Borges}, that a weakly-free $\mathtt{R}\check{\mathtt{R}}$-linear code  of block-length $(\mathpzc{a},\mathpzc{b})$ and of the type $\{k_0,\underbrace{0,0,\ldots, 0}_{e-1 \text{ times}}; \ell_0,\underbrace{0,0, \ldots, 0}_{s-1 \text{ times}}\}$ is permutation equivalent to a code with a generator matrix in the standard form 
\vspace{-2mm}\begin{equation*}
\left[\begin{array}{cc|cc}
    I_{k_0} & A & 0 & B \\ 
    0 & \gamma^{e-s} C & I_{\ell_0} & D 
\end{array}\right],
\vspace{-2mm}\end{equation*}
where $A\in M_{k_0\times (\mathpzc{a}-k_0)}(\mathtt{R}),~\gamma^{e-s} C\in M_{\ell_0\times (\mathpzc{a}-k_0)}(\mathtt{R}),~B\in M_{k_0\times (\mathpzc{b}-\ell_0)}(\check{\mathtt{R}})$ and $D\in M_{\ell_0\times (\mathpzc{b}-\ell_0)}(\check{\mathtt{R}}).$

Now, to study the dual codes of $\mathtt{R}\check{\mathtt{R}}$-linear codes of block-length $(\mathpzc{a},\mathpzc{b}),$ let $Aut_\mathtt{S}(\mathtt{R})$ be the set consisting of all $\mathtt{S}$-automorphisms of $\mathtt{R}.$   By Theorem XV.2 of \cite{MacDonald}, we see that $Aut_\mathtt{S}(\mathtt{R})$ is a cyclic group of order $w$ under the composition operation, which is called the automorphism group of $\mathtt{R}$ over $\mathtt{S}.$   A generator $\sigma$ of the automorphism group $Aut_\mathtt{S}(\mathtt{R})$ is called the Frobenius $\mathtt{S}$-automorphism of $\mathtt{R}.$ Note that  $Aut_\mathtt{S}(\mathtt{R})=\{\sigma^h:0\leq h\leq w-1\}.$ Further, one can easily see that  the chain ring $\check{\mathtt{R}}$ is the Galois extension of the chain ring $\check{\mathtt{S}}=\mathtt{S}/\gamma^s \mathtt{S}$ of degree $w \geq 2.$ Working as above, we see that the automorphism group $Aut_{\check{\mathtt{S}}}(\check{\mathtt{R}})$ of $\check{\mathtt{R}}$ over $\check{\mathtt{S}}$ is also a cyclic group of order $w$ under the composition operation with the restriction map $\check{\sigma}=\sigma\restriction_{\check{\mathtt{R}}}$ as the Frobenius automorphism of $\check{\mathtt{R}}$ over $\check{\mathtt{S}}.$  Throughout this paper, we assume, without any loss of generality, that $\sigma=\check{\sigma},$ which implies that $Aut_\mathtt{S}(\mathtt{R})=Aut_{\check{\mathtt{S}}}(\check{\mathtt{R}}).$ Now, for an integer $h$ satisfying  $0\leq h\leq w-1,$  let us define a  map $\langle \cdot,\cdot\rangle_h:\mathbf{M}\times\mathbf{M}\rightarrow\mathtt{R}$ as
\vspace{-2mm}$$\langle m_1,m_2\rangle_h=\sum\limits_{i=1}^\mathpzc{a}\mathtt{c}_i\sigma^h(\mathtt{c}'_i)+\gamma^{e-s}\sum\limits_{j=1}^{\mathpzc{b}}\mathtt{d}_j\sigma^h(\mathtt{d}'_j)
\vspace{-2mm}$$ for all $m_1=(\mathtt{c}_1,\mathtt{c}_2,\ldots,\mathtt{c}_\mathpzc{a}|\mathtt{d}_{1},\mathtt{d}_2,\ldots,\mathtt{d}_\mathpzc{b}),~m_2=(\mathtt{c}'_1,\mathtt{c}'_2,\ldots,\mathtt{c}'_\mathpzc{a}|\mathtt{d}'_{1},\mathtt{d}'_2,\ldots,\mathtt{d}'_\mathpzc{b})\in\mathbf{M}.$
One can easily verify that the map $\langle \cdot,\cdot\rangle_h$ is a non-degenerate $\sigma^h$-sesquilinear form on $\mathbf{M}$ and is called the $h$-Galois inner product on $\mathbf{M}$ \cite{Bajalan}. In particular, when $h=0,$ the form $\langle\cdot,\cdot\rangle_0$ is symmetric and   is called the Euclidean inner product on $\mathbf{M}.$ On the other hand, when $w$ is even and $h=\frac{w}{2},$ the form $\langle\cdot,\cdot\rangle_{\frac{w}{2}}$ is  Hermitian and reflexive, and is called the Hermitian inner product on $\mathbf{M}.$ 
Further, the $h$-Galois dual code $\mathtt{C}^{\perp_h}$ of an $\mathtt{R}\check{\mathtt{R}}$-linear code $\mathtt{C}(\subseteq\mathbf{M})$ is defined as \vspace{-1mm}$$\mathtt{C}^{\perp_h}=\{m_1\in\mathbf{M}:\langle m_1,m_2\rangle_h=0 \text{ for all }m_2\in\mathtt{C}\}.\vspace{-1mm}$$ Clearly, $\mathtt{C}^{\perp_h}$ is also an $\mathtt{R}\check{\mathtt{R}}$-linear code.  The code $\mathtt{C}$ is said to be an $h$-Galois $\mathtt{R}\check{\mathtt{R}}$-linear code with complementary dual (or an $h$-Galois $\mathtt{R}\check{\mathtt{R}}$-LCD code) if it satisfies $\mathtt{C}\cap \mathtt{C}^{\perp_h}=\{0\}.$ In particular,  $h$-Galois $\mathtt{R}\check{\mathtt{R}}$-LCD codes coincide with   Euclidean $\mathtt{R}\check{\mathtt{R}}$-LCD codes when $h=0,$ while $h$-Galois $\mathtt{R}\check{\mathtt{R}}$-LCD codes match with  Hermitian $\mathtt{R}\check{\mathtt{R}}$-LCD codes when $w$ is even and $h=\frac{w}{2}.$ 
Bajalan \etal  ~\cite{Bajalan} observed that each $h$-Galois $\mathtt{R}\check{\mathtt{R}}$-LCD code $\mathtt{C}$ of block-length $(\mathpzc{a}, \mathpzc{b})$ is weakly-free. The following lemma states this result.
\begin{lemma}\cite{Bajalan}\label{Lemma4.4}
An $h$-Galois $\mathtt{R}\check{\mathtt{R}}$-LCD code $\mathtt{C}$ of block-length $(\mathpzc{a}, \mathpzc{b})$ is weakly-free. As a consequence, the code $\mathtt{C}$ is permutation equivalent to a code with a generator matrix in the standard form 
\begin{align}\label{Equation4.3}
\mathtt{G}=\left[\begin{array}{cc|cc}
    I_{k_0} & A & 0 & B \\ 
    0 & \gamma^{e-s} C & I_{\ell_0} & D 
\end{array}\right]
\end{align}
for some non-negative integers $k_0$ and $\ell_0,$ where $A\in M_{k_0\times (\mathpzc{a}-k_0)}(\mathtt{R}),~\gamma^{e-s} C\in M_{\ell_0\times (\mathpzc{a}-k_0)}(\mathtt{R}),~B\in M_{k_0\times (\mathpzc{b}-\ell_0)}(\check{\mathtt{R}})$ and $D\in M_{\ell_0\times (\mathpzc{b}-\ell_0)}(\check{\mathtt{R}}).$ 
\end{lemma}

Two $\mathtt{R}\check{\mathtt{R}}$-linear codes $\mathtt{C}$ and $\mathtt{D}$ of block-length $(\mathpzc{a},\mathpzc{b})$ form a linear complementary pair of codes (or an LCP of codes) if their direct sum $\mathtt{C}\oplus \mathtt{D}=\mathbf{M}.$
If an $\mathtt{R}\check{\mathtt{R}}$-linear code $\mathtt{C}$ of block-length $(\mathpzc{a},\mathpzc{b})$ is $h$-Galois LCD, then the codes $\mathtt{C}$ and $\mathtt{C}^{\perp_h}$ form an LCP. In this paper, we will introduce and study a direct sum masking scheme based on  LCPs of $\mathtt{R}\check{\mathtt{R}}$-linear codes and obtain its security threshold against fault injection and side-channel attacks. We will also discuss another application of LCPs of  $\mathtt{R}\check{\mathtt{R}}$-linear codes in coding for the two-user adder channel.

From now on, we will follow the same notations as introduced in Section \ref{Preliminaries}. In the following section, we will study $h$-Galois $\mathtt{R}\check{\mathtt{R}}$-LCD codes of block-length $(\mathpzc{a}, \mathpzc{b}).$

\section{$h$-Galois $\mathtt{R}\check{\mathtt{R}}$-LCD codes of block-length $(\mathpzc{a}, \mathpzc{b})$}\label{Section3}
Here, we first recall, by Lemma \ref{Lemma4.4}, that every $h$-Galois $\mathtt{R}\check{\mathtt{R}}$-LCD code of block-length $(\mathpzc{a}, \mathpzc{b})$ is weakly-free. However, the converse of this result is not true in general, \textit{i.e.,} every weakly-free $\mathtt{R}\check{\mathtt{R}}$-linear code of block-length $(\mathpzc{a}, \mathpzc{b})$ need not  be an $h$-Galois $\mathtt{R}\check{\mathtt{R}}$-LCD code. The following example illustrates this. 
\begin{example}\label{Example4.1}
Let $\mathbb{F}_9$ be the finite field of order $9$ with a primitive element $\zeta.$  Let $\mathtt{R}$ be the quasi-Galois ring $\mathbb{F}_9[u]/\langle u^2\rangle,$ and let   $\check{\mathtt{R}}=\mathbb{F}_9.$  Note that the ring $\mathtt{R}$ is the Galois extension of $\mathtt{S}=\mathbb{F}_3[u]/\langle u^2 \rangle$ of degree $w=2,$ and the map $\sigma:\mathtt{R}\rightarrow\mathtt{R},$ defined as $\sigma(a_0+a_1u)=a_0^3+a_1^3u$ for all $a_0+a_1u\in\mathtt{R}$ with $a_0,a_1 \in \mathbb{F}_9,$ is the Frobenius $\mathtt{S}$-automorphism of   $\mathtt{R}.$ Let us take $\mathbf{M}=\mathtt{R}\oplus\check{\mathtt{R}}^2.$   Let $\mathtt{C}_1(\subseteq\mathbf{M})$ be an $\mathtt{R}\check{\mathtt{R}}$-linear code  with a generator matrix
$[0~|~2~~\zeta ].$
Clearly, the code $\mathtt{C}_1$ is of the type $\{0,0;1\},$ and hence it is a weakly-free code. Note that $\mathtt{C}_1=\{(0|0,0),(0|2,\zeta),(0|2\zeta,\zeta^2),(0|2\zeta^2,\zeta^3),(0|2\zeta^3,2),(0|1,2\zeta),(0|\zeta,2\zeta^2),(0|\zeta^2,2\zeta^3),(0|\zeta^3,1)\}.$ 
One can easily verify that $\langle(0|2,\zeta),c\rangle_1=0$ for all $c\in\mathtt{C}_1.$ This implies that the codeword $(0|2,\zeta)\in\mathtt{C}_1\cap\mathtt{C}_1^{\perp_1},$ and hence the code $\mathtt{C}_1$ is not $1$-Galois $\mathtt{R}\check{\mathtt{R}}$-LCD. 
\end{example}We  will show, in Corollary \ref{corH}, that when $|\overline{\mathtt{R}}|>4,$
 every weakly-free $\mathtt{R}\check{\mathtt{R}}$-linear code  of block-length $(\mathpzc{a},\mathpzc{b})$ is monomially equivalent to an $h$-Galois  $\mathtt{R}\check{\mathtt{R}}$-LCD code.  Throughout this paper, let $\sigma^h(V)$ denote the matrix whose $(i,j)$-th entry is $\sigma^h({v}_{i,j})$ if ${v}_{i,j}$ is the $(i,j)$-th entry of the matrix $V$. Now, let  $\mathtt{C}(\subseteq \mathbf{M})$ be a weakly-free $\mathtt{R}\check{\mathtt{R}}$-linear code  with a generator matrix $\mathtt{G}$ of the form \eqref{Equation4.3}.  Bajalan \etal~\cite[Th. 1]{Bajalan}  assumed that $A\sigma^h(C)^T+\iota\big(B\sigma(D)^T\big)\in M_{k_0\times \ell_0}(\gamma \mathtt{R})$ and derived three equivalent necessary and sufficient conditions under which the code $\mathtt{C}$ is  $h$-Galois LCD, where the map $\iota:\check{\mathtt{R}}\rightarrow\mathtt{R},$ defined as $a\mapsto\gamma^{e-s}\iota(a)$ for all $a\in\check{\mathtt{R}},$ is the embedding  of $\check{\mathtt{R}}$ into $\mathtt{R}$.  The following example illustrates that the condition $A\sigma^h(C)^T+\iota\big(B\sigma^h(D)^T\big)\in M_{k_0\times \ell_0}(\gamma \mathtt{R})$ is not necessary for the existence of an $h$-Galois $\mathtt{R}\check{\mathtt{R}}$-LCD code of block-length $(\mathpzc{a}, \mathpzc{b}).$ 
\begin{example}\label{Example4.2}
Let $\mathbb{F}_8=\{0,1,\xi,\xi^2,\xi^3=\xi+1,\xi^4,\xi^5,\xi^6\}$ be the finite field of order $8$ with a primitive element $\xi.$ Let $\mathtt{R}$ be the quasi-Galois ring $\mathbb{F}_8[u]/\langle u^2\rangle,$ and let  $\check{\mathtt{R}}=\mathbb{F}_8.$ Note that the ring $\mathtt{R}$ is the Galois extension of $\mathtt{S}=\mathbb{F}_2[u]/\langle u^2\rangle$ of degree $w=3,$ and the map $\sigma:\mathtt{R}\rightarrow\mathtt{R},$ defined as $\sigma(a_0+a_1u)=a_0^2+a_1^2u$ for all $a_0+a_1u\in\mathtt{R}$ with $a_0,a_1 \in \mathbb{F}_8,$  is the Frobenius $\mathtt{S}$-automorphism of $\mathtt{R}.$ Let us define $\mathbf{M}=\mathtt{R}^3\oplus \check{\mathtt{R}}^3.$ Let  $\mathtt{C}_2(\subseteq \mathbf{M})$ be an $\mathtt{R}\check{\mathtt{R}}$-linear code with a generator matrix
\vspace{-2mm}$$\left[\begin{array}{cc|cc}
    I_1 & A & 0 & B\\ 
    0 &  uC & I_2 & D
\end{array}\right]=\left[\begin{array}{ccc|ccc}
    1 & 1 & 1 & 0 & 0 & 1\\ 
    0 & u & 0 & 1 & 0 & 0\\
    0 & 0 & u & 0 & 1 & 0
\end{array}\right],\vspace{-2mm}$$ where $A=[ 1 ~ 1],$ $B=[1],$ $C=\left[\begin{array}{cccc}
    1 &0 \\
     0 &1
\end{array}\right]$ and $D=\left[\begin{array}{c}
    0   \\
    0  
\end{array}\right].$
It is easy to see that  $A\sigma(C)^T+\iota\big(B\sigma(D)^T\big)=[
    1 ~~ 1]\notin M_{1\times2}(u\mathtt{R}).$ However, one can easily see that the code $\mathtt{C}_2$ is $1$-Galois $\mathtt{R}\check{\mathtt{R}}$-LCD.
\end{example}

Thus it is essential to derive necessary and sufficient conditions under which a weakly-free  $\mathtt{R}\check{\mathtt{R}}$-linear code of block-length $(\mathpzc{a}, \mathpzc{b})$ is an $h$-Galois $\mathtt{R}\check{\mathtt{R}}$-LCD code, without assuming the constraint $A\sigma^h(C)^T+\iota\big(B\sigma^h(D)^T\big)\in M_{k_0\times \ell_0}(\gamma \mathtt{R})$.  To do this, we will first define an operation on the set $\mathfrak{M}$ of matrices whose rows are elements of $\mathbf{M}.$ Here, we first note that any  matrix in $\mathfrak{M}$ is of the form $[E~| ~F],$ where $E$ is a matrix over $\mathtt{R}$ with $\mathpzc{a}$ columns and $F$ is a matrix over $\check{\mathtt{R}}$ with $\mathpzc{b}$ columns. Now, for $G=[E~|~F]$ and $H=[U~|~V]$ in $\mathfrak{M},$ we define \vspace{-1mm}
$$G \diamond \sigma^h(H)^T= E\sigma^h(U)^T+\gamma^{e-s} F\sigma^h(V)^T,\vspace{-1mm}$$ 
where $Z^{T}$ denotes the  transpose of the matrix $Z.$ Clearly, $G \diamond \sigma^h(H)^T$ is a matrix over $\mathtt{R}.$ We further see that if $\mathtt{G}$ is a matrix of the form \eqref{Equation4.3}, then 
\vspace{-2mm}\begin{align}\label{Equation4.7}
\mathtt{G}\diamond\sigma^h(\mathtt{G})^T& =\left[\begin{array}{cccc}
    I_{k_0} & A & 0 & B \\
    0 & \gamma^{e-s} C & I_{\ell_0} & D 
\end{array}\right]
\left[\begin{array}{cc}
    I_{k_0} & 0 \\
    \sigma^h(A)^T & \gamma^{e-s}\sigma^h(C)^T\\
    0 & \gamma^{e-s} I_{\ell_0}\\
    \gamma^{e-s} \sigma^h(B)^T & \gamma^{e-s} \sigma^h(D)^T
\end{array}\right]\nonumber\\
& =\left[\begin{array}{cc}
     I_{k_0}+A\sigma^h(A)^T+\gamma^{e-s} B\sigma^h(B)^T & \gamma^{e-s}\big(A\sigma^h(C)^T+ B\sigma^h(D)^T\big)  \\  \gamma^{e-s}\big(C\sigma^h(A)^T+D\sigma^h(B)^T\big) & \gamma^{e-s}\big(I_{\ell_0}+D\sigma^h(D)^T+ \gamma^{e-s}C\sigma^h(C)^T\big)
\end{array}\right].
\vspace{-2mm}\end{align}
 Clearly, the square matrix $\mathtt{G}\diamond\sigma^h(\mathtt{G})^T\in M_{(k_0+\ell_0)\times (k_0+\ell_0)}(\mathtt{R})$ generates a linear code of length $k_0+\ell_0$ over $\mathtt{R}.$ A generator matrix of a linear code $\mathcal{C}$ of length $n$ over $\mathtt{R}$ is defined as a matrix over $\mathtt{R}$ whose rows are independent and generate the code $\mathcal{C}$ over $\mathtt{R}.$ We further see, by Theorem 3.5 of Norton and S{\u{a}}l{\u{a}}gean \cite{Norton}, that a linear code over a chain ring has a unique type. In the following theorem, we derive a necessary and sufficient condition under which an $\mathtt{R}\check{\mathtt{R}}$-linear code of block-length $(\mathpzc{a},\mathpzc{b})$ with a generator matrix of the form \eqref{Equation4.3} is an $h$-Galois $\mathtt{R}\check{\mathtt{R}}$-LCD code, without assuming the constraint  $A\sigma^h(C)^T+\iota\big(B\sigma^h(D)^T\big)\in M_{k_0\times \ell_0}(\gamma \mathtt{R}).$
 \begin{thm}\label{Theorem4.1}
Let $\mathtt{C}(\subseteq\mathbf{M})$ be an $\mathtt{R}\check{\mathtt{R}}$-linear code 
with a generator matrix $\mathtt{G}.$ The code
$\mathtt{C}$ is $h$-Galois $\mathtt{R}\check{\mathtt{R}}$-LCD if and only if the code $\mathtt{C}$ is weakly-free and the matrix $\mathtt{G}\diamond\sigma^h(\mathtt{G})^T$ is a generator matrix of a linear code of the type $\{k_0,0,0,\ldots,0,\underbrace{\ell_0}_{(e-s+1)\text{-th}},0,0,\ldots,0\}$ and length $k_0+\ell_0$ over $\mathtt{R}.$
\end{thm}
\begin{proof} 
To prove the result, we first suppose that  the code $\mathtt{C}$ is $h$-Galois $\mathtt{R}\check{\mathtt{R}}$-LCD. Here, we see, by Lemma \ref{Lemma4.4}, that the code $\mathtt{C}$ is weakly-free
and its generator matrix  $\mathtt{G}$ is of  the form \eqref{Equation4.3}, whose first $k_0$ rows have period $\gamma^{e}$ and the remaining $\ell_0$ rows have period $\gamma^{s}.$ One can see, by \eqref{Equation4.7}, that 
\vspace{-2mm}\begin{eqnarray*}
\mathtt{G}\diamond\sigma^h(\mathtt{G})^T&=&\left[\begin{array}{cc}
     I_{k_0}+A\sigma^h(A)^T+\gamma^{e-s} B\sigma^h(B)^T & \gamma^{e-s}\big(A\sigma^h(C)^T+ B\sigma^h(D)^T\big)  \\  \gamma^{e-s}\big(C\sigma^h(A)^T+D\sigma^h(B)^T\big) & \gamma^{e-s}\big(I_{\ell_0}+D\sigma^h(D)^T+ \gamma^{e-s}C\sigma^h(C)^T\big)
\end{array}\right].
\vspace{-4mm}\end{eqnarray*}
We next assert that the matrix $\mathtt{G}\diamond\sigma^h(\mathtt{G})^T$  is a generator matrix of a linear code of the type $\{k_0,0,0,\ldots,0,\underbrace{\ell_0}_{(e-s+1)\text{-th}},\\0,0,\ldots,0\}$ and length $k_0+\ell_0$ over $\mathtt{R}.$ To prove this assertion,  it suffices to show  that the rows of the matrix $\mathtt{G}\diamond\sigma^h(\mathtt{G})^T$ are independent over $\mathtt{R}$ and that the first $k_0$ rows have period $\gamma^e$ and the remaining $\ell_0$ rows have period $\gamma^{s}.$ For this, let $R_i$ denote the $i$-th row of the matrix $\mathtt{G}\diamond\sigma^h(\mathtt{G})^T$ for $1\leq i\leq k_0+\ell_0.$ 

Suppose, on the contrary, that  there exists an integer $i$ such that $1 \leq i \leq k_0$ and  $\gamma^vR_i=0$ for some positive integer $v<e.$ 
Thus there exists a $(k_0+\ell_0)$-tuple $w_1=(0,0,\ldots,0,\underbrace{\gamma^v}_{i\text{-th}},0,0,\ldots, 0)$ over $\mathtt{R}$ satisfying $w_1\mathtt{G}\diamond\sigma^h(\mathtt{G})^T=0,$ which implies that $w_1\mathtt{G}\in\mathtt{C}^{\perp_h} \cap \mathtt{C}.$ However, since the first $k_0$ rows of $\mathtt{G}$ have period $\gamma^{e},$ we must have $w_1\mathtt{G}\ne0,$  which is a contradiction.  From this, it follows that the row $R_i$ has period $\gamma^{e}$ for $1\leq i\leq k_0.$ 

Working as above, one can show, for $k_0+1\leq j\leq k_0+\ell_0,$ that the $j$-th row $R_j$ has period $\gamma^{s}.$
Now, it remains to show that the rows of the matrix $\mathtt{G}\diamond\sigma^h(\mathtt{G})^T$ are independent over $\mathtt{R}.$ To prove this,   suppose that there exist $c_1,c_2,\ldots,c_{k_0+\ell_0}\in\mathtt{R}$ satisfying $\sum\limits_{i=1}^{k_0+\ell_0}c_iR_i=0.$ This gives  $(c_1,c_2,\ldots,c_{k_0+\ell_0})\mathtt{G}\diamond\sigma^h(\mathtt{G})^T=0,$ which implies that $(c_1, c_2,\cdots,c_{k_0+\ell_0})\mathtt{G} \in \mathtt{C}\cap\mathtt{C}^{\perp_h}.$  Since $\mathtt{C}$ is an $h$-Galois $\mathtt{R}\check{\mathtt{R}}$-LCD code, we must have $(c_1, c_2,\cdots,c_{k_0+\ell_0})\mathtt{G}=0.$ Since the rows of the matrix $\mathtt{G}$ are independent over $\mathtt{R},$ the first $k_0$ rows of $\mathtt{G}$ have period $\gamma^e$ and the remaining $\ell_0$ rows of $\mathtt{G}$ have period $\gamma^{s},$ we obtain  $c_i\in\gamma^e\mathtt{R}$ for $1\leq i\leq k_0$ and $c_j\in \gamma^{s}\mathtt{R}$ for $k_0+1\leq j\leq k_0+\ell_0.$ From this, it follows that $c_iR_i=0$ for $1\leq i\leq k_0+\ell_0.$ This shows that the rows of the matrix $\mathtt{G}\diamond\sigma^h(\mathtt{G})^{T}$ are independent over $\mathtt{R}.$ This completes the proof of the assertion.

Conversely, suppose that the $\mathtt{R}\check{\mathtt{R}}$-linear code $\mathtt{C}(\subseteq\mathbf{M})$ is weakly-free and the matrix $\mathtt{G}\diamond\sigma^h(\mathtt{G})^T$ is a generator matrix of a linear code of the type $\{k_0,0,0,\ldots,0,\underbrace{\ell_0}_{(e-s+1)\text{-th}},0,0,\ldots,0\}$ and length $k_0+\ell_0$ over $\mathtt{R}.$ Here, we will show that $\mathtt{C}\cap\mathtt{C}^{\perp_h}=\{0\}.$
For this, let $\mathtt{c}\in\mathtt{C}\cap\mathtt{C}^{\perp_h}.$ Thus there exists $y=(y_1,y_2,\ldots,y_{k_0+\ell_0})\in\mathtt{R}^{k_0+\ell_0}$ such that $\mathtt{c}=y\mathtt{G}.$ Since $\mathtt{c}\in\mathtt{C}^{\perp_h},$
we must have $ \mathtt{c}\diamond \sigma^h(\mathtt{G})^T=0,$ which   implies that $y\mathtt{G}\diamond\sigma(\mathtt{G})^T=0.$ As the matrix $\mathtt{G}\diamond\sigma^h(\mathtt{G})^T$ is a generator matrix of a linear code of the type $\{k_0,0,0,\ldots,0,\underbrace{\ell_0}_{(e-s+1)\text{-th}},0,0,\ldots,0\}$ and length $k_0+\ell_0$ over $\mathtt{R},$ we must have $y_i\in \gamma^e\mathtt{R}$ for $1\leq i\leq k_0$ and $y_j\in \gamma^{s}\mathtt{R}$ for $k_0+1\leq j\leq k_0+\ell_0$. This implies that  $\mathtt{c}=y\mathtt{G}=0.$ This shows that the code  $\mathtt{C}$ is $h$-Galois $\mathtt{R}\check{\mathtt{R}}$-LCD. 
\end{proof}
\begin{remark}\label{Rk1} From the above theorem, it follows that  a weakly-free $\mathtt{R}\check{\mathtt{R}}$-linear code  $\mathtt{C}$ of 
block-length $(\mathpzc{a}, \mathpzc{b})$ and of the type $\{k_0,\underbrace{0,0,\ldots, 0}_{e-1 \text{ times}}; \ell_0,\underbrace{0,0, \ldots, 0}_{s-1 \text{ times}}\}$ with a generator matrix $\mathtt{G}$  is $h$-Galois $\mathtt{R}\check{\mathtt{R}}$-LCD if and only if  there exists an invertible $(k_0+\ell_0)\times (k_0+\ell_0)$ matrix $P $ over $\mathtt{R}$ such that $\big(\mathtt{G} \diamond \sigma(\mathtt{G})^T\big) P =\left[\begin{array}{c|c}
    I_{k_0} & 0 \\
    0 & \gamma^{e-s} I_{\ell_0}
\end{array} \right].$ This shows that the statements \textbf{1.} and \textbf{3.} in Theorem 1 of Bajalan \textit{et al.} \cite{Bajalan} are equivalent, even when   $A\sigma^h(C)^T+\iota\big(B\sigma^h(D)^T\big)\in M_{k_0\times \ell_0}(\gamma \mathtt{R})$ does not hold.
\end{remark}
In the following theorem, we derive a necessary and sufficient  condition under which the matrix $\mathtt{G}\diamond\sigma(\mathtt{G})^T$ of the form \eqref{Equation4.7} is a generator matrix of a linear code of the type $\{k_0,0,0,\ldots,0,\underbrace{\ell_0}_{(e-s+1)\text{-th}},0,0,\ldots,0\}$ and length $k_0+\ell_0$ over $\mathtt{R}.$ Here, we first recall that $\overline{\mathtt{R}}$ is  the residue field of both $\mathtt{R}$ and $\check{\mathtt{R}}.$  Let us assume that the residue field $\overline{\mathtt{R}}$ is of order $q,$ where $q$ is a prime power. For any matrix $V$ over $\mathtt{R}$ (\textit{resp.} $\check{\mathtt{R}}$) with the $(i,j)$-th entry as $v_{i,j},$ let $\overline{V}$ denote the matrix over $\overline{\mathtt{R}}$ whose $(i,j)$-th entry is $\overline{v}_{i,j}=v_{i,j} +\gamma \mathtt{R}$ (\textit{resp.} $\overline{v}_{i,j}=v_{i,j} +\gamma \check{\mathtt{R}}$) for each $i$ and $j.$
 
\begin{thm}\label{Lemma4.7}
Let $\mathtt{G}$ be a matrix of the form \eqref{Equation4.3}, (note that the matrix $\mathtt{G}\diamond\sigma^h(\mathtt{G})^T$ is given by \eqref{Equation4.7}). The matrix $\mathtt{G}\diamond\sigma^h(\mathtt{G})^T$ is a generator matrix of a linear code of the type $\{k_0,0,0,\ldots,0,\underbrace{\ell_0}_{(e-s+1)\text{-th}},0,0,\ldots,0\}$ and length $k_0+\ell_0$ over $\mathtt{R}$ if and only if both the matrices $
I_{k_0}+\overline{A}~\overline{\sigma^h(A)}^T
$ and $
I_{\ell_0}+\overline{D}~\overline{\sigma^h(D)}^T
$
are invertible over  $\overline{\mathtt{R}}.$
\end{thm}
\begin{proof} 
To prove the result, we first assume that  the matrix $\mathtt{G}\diamond\sigma^h(\mathtt{G})^T$ is a generator matrix of a linear code of the type $\{k_0,0,0,\ldots,0,\underbrace{\ell_0}_{(e-s+1)\text{-th}},0,0,\ldots,0\}$ and length $k_0+\ell_0$ over $\mathtt{R}.$ Here, working as in Theorem 1 of Bajalan \etal~\cite{Bajalan},  we see that the matrix  $
I_{k_0}+\overline{A}~ \overline{\sigma^h(A)}^T$ is invertible.  
We will next show that the matrix $
I_{\ell_0}+\overline{D}~\overline{\sigma^h(D)}^T
$ is invertible. For this, it suffices to show that the rows of the matrix $
I_{\ell_0}+\overline{D}~\overline{\sigma^h(D)}^T
$ are linearly independent over $\overline{\mathtt{R}}.$ 

Suppose, on the contrary, that the rows of the matrix $
I_{\ell_0}+\overline{D}~\overline{\sigma^h(D)}^T
$ are linearly dependent over $\overline{\mathtt{R}}.$ Thus there exist $y_1,y_2,\ldots,y_{\ell_0}\in\overline{\mathtt{R}},$ such that $y_1\ne0$ and  $y_1W^{(1)}+y_2W^{(2)}+\cdots+y_{\ell_0}W^{(\ell_0)}=0,$ where $W^{(j)}$ denotes the $j$-th row of the matrix $
I_{\ell_0}+\overline{D}~\overline{\sigma^h(D)}^T
$ for $1\leq j\leq \ell_0.$ This implies that $b_1V^{(1)}+b_2V^{(2)}+\cdots+b_{\ell_0}V^{(\ell_0)}\in\gamma\mathtt{R}^{\ell_0},$ where $V^{(j)}$ denotes the $j$-th row of the matrix $
I_{\ell_0}+D\sigma^h(D)^T
$ satisfying $\overline{V^{(j)}}=W^{(j)}$ and $b_j\in\mathtt{R}$ satisfies $\overline{b}_j=y_j$  for $1\leq j\leq \ell_0.$ From this, it follows that  \vspace{-1mm}$$b_1\gamma^{e-s} V'^{(1)}+b_2\gamma^{e-s} V'^{(2)}+\cdots+b_{\ell_0}\gamma^{e-s} V'^{(\ell_0)}\in\gamma^{e-s+1}\mathtt{R}^{\ell_0}, \vspace{-1mm}$$ where $\gamma^{e-s} V'^{(j)}$ denotes the $j$-th row of the matrix $
\gamma^{e-s}\big(I_{\ell_0}+D\sigma^h(D)^T+\gamma^{e-s} C\sigma^h(C)^T\big)
$ for $1\leq j\leq \ell_0.$ This implies that $\gamma^{s-1}b_1\gamma^{e-s} V'^{(1)}+\gamma^{s-1}b_2\gamma^{e-s} V'^{(2)}+\cdots+\gamma^{s-1}b_{\ell_0}\gamma^{e-s} V'^{(\ell_0)}=0.$ We further observe that  \begin{equation}\label{EQDD}\gamma^{s-1}b_1\gamma^{e-s} U^{(1)}+\gamma^{s-1}b_2\gamma^{e-s} U^{(2)}+\cdots+\gamma^{s-1}b_{\ell_0}\gamma^{e-s} U^{(\ell_0)}=(\gamma^{e-1}s_1,\gamma^{e-1}s_2,\ldots,\gamma^{e-1}s_{k_0},0,0,\ldots,0), \end{equation} where $\gamma^{e-s} U^{(j)}$ denotes the $j$-th row of the matrix $[
\gamma^{e-s}\big(C\sigma^h(A)^T+D\sigma^h(B)^T\big) ~~~ \gamma^{e-s}\big(I_{\ell_0}+D\sigma^h(D)^T+\gamma^{e-s} C\sigma^h(C)^T\big)
]$ for $1\leq j\leq \ell_0$ and $s_i\in\mathtt{R}$ for $1\leq i\leq k_0.$ Now since $\overline{b}_1=y_1\ne0,$ we note that $b_1\notin\gamma\mathtt{R}.$ Further, as $\gamma^{e-s} U^{(1)}$ has period $\gamma^{s},$ we must have  $\gamma^{s-1}b_1\gamma^{e-s} U^{(1)}\ne0.$ Since the first $k_0$ rows of $\mathtt{G}\diamond\sigma^h(\mathtt{G})^T$ are independent over $\mathtt{R}$ and have period $\gamma^e,$  we see, by applying Lemma \ref{Lemma4.5}(c),  that the rows $Y^{(1)}, Y^{(2)}, \ldots, Y^{(k_0)}$ of the matrix $
I_{k_0}+A\sigma^h(A)^T+\gamma^{e-s} B\sigma^h(B)^T
$ are  independent over $\mathtt{R}$ and have period $\gamma^e.$ This implies that the rows $Y^{(1)}, Y^{(2)}, \ldots, Y^{(k_0)}$  of the matrix $
I_{k_0}+A\sigma(A)^T+\gamma^{e-s} B\sigma^h(B)^T
$ must generate $\mathtt{R}^{k_0}.$ Thus there exist $g_1,g_2,\ldots,g_{k_0}\in\mathtt{R}$ such that  \begin{equation}\label{EQG}g_1Y^{(1)}+g_2Y^{(2)}+\cdots+g_{k_0}Y^{(k_0)}=(-\gamma^{e-1}s_1,-\gamma^{e-1}s_2,\ldots,-\gamma^{e-1}s_{k_0}). \end{equation}
We next assert that  $g_i\in\gamma^{e-1}\mathtt{R}$ for $1\leq i\leq k_0.$ Suppose, if possible, that $g_i\notin\gamma^{e-1}\mathtt{R}$ for some $i.$ Now equation \eqref{EQG} implies that $\gamma g_1Y^{(1)}+\gamma g_2Y^{(2)}+\cdots+\gamma g_{k_0}Y^{(k_0)}=(0,0,\ldots,0)$ and $\gamma g_iY^{(i)}\neq0,$ which contradicts the fact that the rows $Y^{(1)}, Y^{(2)}, \ldots, Y^{(k_0)}$  of the matrix $
I_{k_0}+A\sigma^h(A)^T+\gamma^{e-s} B\sigma^h(B)^T
$ are independent over $\mathtt{R}.$ This shows that $g_i\in\gamma^{e-1}\mathtt{R}$ for $1\leq i\leq k_0.$ 

Now, if $Q^{(i)}$ denotes the $i$-th row of the matrix  $[
I_{k_0}+A\sigma^h(A)^T+\gamma^{e-s} B\sigma^h(B)^T ~~~ \gamma^{e-s}\big(A\sigma^h(C)^T+B\sigma^h(D)^T\big)
]$ for $1\leq i\leq k_0,$ 
then we have  \vspace{-1mm}$$g_1 Q^{(1)}+g_2 Q^{(2)}+\cdots+g_{k_0}Q^{(k_0)}=(-\gamma^{e-1}s_1,-\gamma^{e-1}s_2,\ldots,-\gamma^{e-1}s_{k_0},0,0\ldots,0). \vspace{-1mm}$$ From this and by \eqref{EQDD}, we obtain  $g_1 Q^{(1)}+g_2 Q^{(2)}+\cdots+g_{k_0} Q^{(k_0)}+\gamma^{s-1}b_1\gamma^{e-s} U^{(1)}+\gamma^{s-1}b_2\gamma^{e-s} U^{(2)}+\cdots+\gamma^{s-1}b_{\ell_0}\gamma^{e-s} U^{(\ell_0)}=0,$ where $\gamma^{s-1}b_1\gamma^{e-s} U^{(1)}\ne0.$ This contradicts the fact that the rows of the matrix $\mathtt{G}\diamond\sigma^h(\mathtt{G})^T$ are independent over $\mathtt{R}.$ Thus the rows of the matrix $
I_{\ell_0}+\overline{D}~\overline{\sigma^h(D)}^T
$ are linearly independent over $\overline{\mathtt{R}},$ which implies that the matrix $
I_{\ell_0}+\overline{D}~\overline{\sigma^h(D)}^T
$ is invertible.

On the other hand, working as in Theorem 1 of Bajalan \etal~\cite{Bajalan} and by applying Lemma \ref{Lemma4.5}, the converse part follows immediately.  
\end{proof}
\begin{remark}\label{Rk2} While proving \textbf{1.} $\implies$ \textbf{2.} in the proof of Theorem 1 of Bajalan \textit{et al.} \cite{Bajalan},   the matrix   
$
I_{\ell_0}+\overline{D}~\overline{\sigma^h(D)}^T
$
is shown to be invertible over  $\overline{\mathtt{R}}$ (or equivalently, the matrix 
$
I_{\ell_0}+D\sigma^h(D)^T
$
is shown to be invertible over  $\check{\mathtt{R}}$) by strongly using the constraint  $A\sigma^h(C)^T+\iota\big(B\sigma^h(D)^T\big)\in M_{k_0\times \ell_0}(\gamma \mathtt{R}).$ That is,  Theorem 1 of Bajalan \textit{et al.} \cite{Bajalan}  proves the equivalence of  the statements \textbf{1.}, \textbf{2.} and \textbf{3.} using the constraint $A\sigma^h(C)^T+\iota\big(B\sigma^h(D)^T\big)\in M_{k_0\times \ell_0}(\gamma \mathtt{R}).$ However, by Theorems \ref{Theorem4.1} and \ref{Lemma4.7} and Remark \ref{Rk1}, one can easily see that the statements \textbf{1.}, \textbf{2.} and \textbf{3.} in  Theorem 1 of Bajalan \textit{et al.} \cite{Bajalan} are equivalent, even when   the condition $A\sigma^h(C)^T+\iota\big(B\sigma^h(D)^T\big)\in M_{k_0\times \ell_0}(\gamma \mathtt{R})$ does not hold. 
\end{remark}

Further, a linear code $C_1$ of length $n$ over $\overline{\mathtt{R}}$ is defined as an $\overline{\mathtt{R}}$-linear subspace of $\overline{\mathtt{R}}^n.$ Here we note, by Theorem XV.2 of \cite{MacDonald}, that the automorphism $\sigma^h$ of $\mathtt{R}$ induces an automorphism $\overline{\sigma}^h$ of $\overline{\mathtt{R}},$ defined as $\overline{\sigma}^h(\overline{b})=\overline{\sigma^h(b)}$ for all $\overline{b}\in\overline{\mathtt{R}}.$ The map $[\cdot,\cdot]_h:\overline{\mathtt{R}}^n\times\overline{\mathtt{R}}^n\rightarrow \overline{\mathtt{R}},$ defined as $$[c,d]_h=c_1\overline{\sigma}^h(d_1)+c_2\overline{\sigma}^h(d_2)+\cdots+c_n\overline{\sigma}^h(d_n)$$ 
for all $c=(c_1,\ldots,c_n),d=(d_1,\ldots,d_n)\in\overline{\mathtt{R}}^n,$ is a $\overline{\sigma}^h$-sesquilinear form and is called the $h$-Galois inner product on $\overline{\mathtt{R}}^n$ \cite{Liu2018}. For a linear code $C_1$ of length $n$ over $\overline{\mathtt{R}},$ the $h$-Galois dual code of $C_1,$ denoted by $C_1^{\perp_h},$ is defined as $C_1^{\perp_h}=\{d\in\overline{\mathtt{R}}^n:[d,c]_h=0\text{ for all }c\in C_1\}.$ Clearly, $C_1^{\perp_h}$ is also a linear code. Further, the linear code $C_1$ is said to be $h$-Galois LCD if it satisfies $C_1\cap C_1^{\perp_h}=\{0\}.$ In particular, $h$-Galois LCD codes over $\overline{\mathtt{R}}$ coincide with Euclidean LCD codes when $h=0,$ while  $h$-Galois LCD codes over $\overline{\mathtt{R}}$ match with Hermitian LCD codes when $w$ is even and $h=\frac{w}{2}.$ 

Throughout this paper, for every weakly-free $\mathtt{R}\check{\mathtt{R}}$-linear code $\mathtt{C}(\subseteq\mathbf{M})$ with a generator matrix $\mathtt{G}$ of the form \eqref{Equation4.3}, let $\mathtt{C}^{(X)}$ denote the linear code of length $\mathpzc{a}$ over $\mathtt{R}$ with a generator matrix $[
I_{k_0} ~~ A
],$ and let $\mathtt{C}^{(Y)}$ denote the  linear code of length $\mathpzc{b}$ over $\check{\mathtt{R}}$ with a generator matrix $[
I_{\ell_0} ~~ D].$ Further, the matrices $A\in M_{k_0\times (\mathpzc{a}-k_0)}(\mathtt{R}),$ $ \gamma^{e-s} C\in M_{\ell_0\times (\mathpzc{a}-k_0)}(\mathtt{R}),$ $B\in M_{k_0\times (\mathpzc{b}-\ell_0)}(\check{\mathtt{R}})$ and $D\in M_{\ell_0\times (\mathpzc{b}-\ell_0)}(\check{\mathtt{R}})$ can be written as  \vspace{-2mm}$$A=\sum\limits_{i=0}^{e-1}\gamma^iA_i,~~B=\sum\limits_{j=0}^{s-1}\gamma^jB_j,~~\gamma^{e-s} C=\sum\limits_{j=0}^{s-1}\gamma^{e-s+j}C_j~\text{ and ~}D=\sum\limits_{j=0}^{s-1}\gamma^jD_j,\vspace{-2mm}$$ where $A_i\in M_{k_0\times(\mathpzc{a}-k_0)}(\mathcal{T}),$ $B_j\in M_{k_0\times(\mathpzc{b}-\ell_0)}(\mathcal{T}),$ $C_j\in M_{\ell_0\times(\mathpzc{a}-k_0)}(\mathcal{T})$ and $D_j\in M_{\ell_0\times(\mathpzc{b}-\ell_0)}(\mathcal{T})$ for $0\leq i\leq e-1$ and $0\leq j\leq s-1,$ (recall that $\mathcal{T}$ is  the Teichm$\ddot{u}$ller set of both $\mathtt{R}$ and $\check{\mathtt{R}}$). One can easily see that $\overline{A}=\overline{A}_0$ and $\overline{D}=\overline{D}_0.$ 
Now, let $\overline{\mathtt{C}^{(X)}}$ denote the linear code of length $\mathpzc{a}$ over $\overline{\mathtt{R}}$ with a generator matrix $[
I_{k_0} ~~ \overline{A}
],$ and let $\overline{\mathtt{C}^{(Y)}}$ denote the  linear code of length $\mathpzc{b}$ over $\overline{\check{\mathtt{R}}}=\overline{\mathtt{R}}$ with a generator matrix $[
I_{\ell_0}~~ \overline{D}
].$  By Theorem 2.4 of Liu~\etal~\cite{Liu2018}, we see that the codes $\overline{\mathtt{C}^{(X)}}$ and $\overline{\mathtt{C}^{(Y)}}$ are $h$-Galois LCD if and only if  the matrices $
I_{k_0}+\overline{A}~\overline{\sigma^h(A)}^T
=[
I_{k_0} ~~ \overline{A}
][
I_{k_0} ~~ \overline{\sigma^h(A)}]^T$ and $
I_{\ell_0}+\overline{D}~\overline{\sigma^h(D)}^T
=[
I_{\ell_0} ~~\overline{D}
][
I_{\ell_0} ~~ \overline{\sigma^h(D)}
]^T$ are invertible over $\overline{\mathtt{R}}$. From this, we deduce the following:

 \begin{thm}\label{Theorem4.2}
Let $\mathtt{C}(\subseteq\mathbf{M})$ be an $\mathtt{R}\check{\mathtt{R}}$-linear code with a generator matrix $\mathtt{G}$ of the form \eqref{Equation4.3}.  The code $\mathtt{C}$ is $h$-Galois $\mathtt{R}\check{\mathtt{R}}$-LCD if and only if the codes $\overline{\mathtt{C}^{(X)}}$ and $\overline{\mathtt{C}^{(Y)}}$ with their respective generator matrices $[
I_{k_0} ~~ \overline{A}
]$ and $[
I_{\ell_0}~~ \overline{D}
]$ are $h$-Galois LCD codes over $\overline{\mathtt{R}}.$
\end{thm}
\begin{proof}
It follows from Lemma \ref{Lemma4.4}, Theorems \ref{Theorem4.1} and \ref{Lemma4.7} and by applying Theorem 2.4 of Liu \etal~\cite{Liu2018}.
\end{proof}
Now, the following example illustrates the above theorem.
\begin{example}\label{Example4.3}
Let $\mathbf{M}=\mathbb{Z}_{4}^5\oplus \mathbb{Z}_2^3.$ Let  $\mathtt{C}_3(\subseteq \mathbf{M})$ be a $\mathbb{Z}_4\mathbb{Z}_2$-linear code with a generator matrix
\vspace{-1mm}$$\left[\begin{array}{ccccc|ccc}
    1 & 1 & 3 & 0 & 0 & 0 & 0 & 1\\ 
    0 & 0 & 2 & 0 & 2 & 1 & 0 & 0\\
    0 & 2 & 0 & 2 & 0 & 0 & 1 & 0
\end{array}\right].\vspace{-1mm}$$
Here, we see that $\overline{\mathtt{C}_3^{(X)}}=\{(0,0,0,0,0),(1,1,1,0,0)\}$ and $\overline{\mathtt{C}_3^{(Y)}}=\{(0,0,0),(1,0,0),(0,1,0),(1,1,0)\}$ are Euclidean LCD codes over the residue field $\mathbb{Z}_2.$ This, by Theorem \ref{Theorem4.2}, implies that the code $\mathtt{C}_3$ is Euclidean $\mathbb{Z}_4\mathbb{Z}_{2}$-LCD. On the other hand, one can directly verify that $\mathtt{C}_{3}\cap \mathtt{C}_{3}^{\perp_0}=\{0\}.$
\end{example}
\begin{remark}\label{Rk3}By Theorem \ref{Theorem4.2}, we see that the statements \textbf{1.} and \textbf{4.} of Theorem 1 of Bajalan \textit{et al.} \cite{Bajalan} are equivalent, even when the condition    $A\sigma^h(C)^T+\iota\big(B\sigma^h(D)^T\big)\in M_{k_0\times \ell_0}(\gamma \mathtt{R})$ does not hold.
\end{remark}

We next recall, by Remark 1 of Bajalan \etal~\cite{Bajalan2023}, that a weakly-free $\mathtt{R}\check{\mathtt{R}}$-linear code $\mathtt{D}$ of block-length $(\mathpzc{a},\mathpzc{b})$ and of the type $\{k_0,\underbrace{0,0,\ldots, 0}_{e-1 \text{ times}}; \ell_0,\underbrace{0,0, \ldots, 0}_{s-1 \text{ times}}\}$ has a generator matrix $\mathtt{H}$ of the form 
\vspace{-5mm}\begin{equation}\label{Eq03.4}
\mathtt{H}=\left[\begin{array}{c|c}
    \mathcal{A}&\mathcal{B} \\ 
    \gamma^{e-s}\mathcal{E}& \mathcal{F} 
\end{array}\right],
\vspace{-2mm}\end{equation}
where rows of the matrix $\mathcal{A}\in M_{k_0\times \mathpzc{a}}(\mathtt{R})$ are independent over $\mathtt{R}$ and have period $\gamma^e,$ rows of the matrix $\mathcal{F}\in M_{\ell_0\times \mathpzc{b}}(\check{\mathtt{R}})$ are independent over $\check{\mathtt{R}}$ and have period $\gamma^s,$ and the matrices $ \mathcal{B}\in M_{k_0\times \mathpzc{b}}(\check{\mathtt{R}})$ and  $\gamma^{e-s}\mathcal{E}\in M_{\ell_0\times \mathpzc{a}}(\mathtt{R}).$ Next, let $\mathtt{D}^{(X)}$ be a linear code of length $\mathpzc{a}$ over $\mathtt{R}$ with a generator matrix $\mathcal{A},$ and let $\mathtt{D}^{(Y)}$ be a linear code of length $\mathpzc{b}$ over $\check{\mathtt{R}}$ with a generator matrix $\mathcal{F}.$ Further, let $\overline{\mathtt{D}^{(X)}}$ and $\overline{\mathtt{D}^{(Y)}}$ be  linear codes of lengths $\mathpzc{a}$ and $\mathpzc{b}$ over $\overline{\mathtt{R}}$ with generator matrices $\overline{\mathcal{A}}$ and $\overline{\mathcal{F}},$ respectively. By Theorem 2.4 of Liu \etal~\cite{Liu2018}, we see that the codes $\overline{\mathtt{D}^{(X)}}$ and $\overline{\mathtt{D}^{(Y)}}$ are $h$-Galois LCD if and only if  the matrices $
\overline{\mathcal{A}}~\overline{\sigma^h(\mathcal{A})}^T
$ and $
\overline{\mathcal{F}}~\overline{\sigma^h(\mathcal{F})}^T
$ are invertible. We next observe the following:
\begin{thm}\label{Theorem04.2}
Let $\mathtt{D}(\subseteq\mathbf{M})$ be a weakly-free $\mathtt{R}\check{\mathtt{R}}$-linear code of the type $\{k_0,\underbrace{0,0,\ldots, 0}_{e-1 \text{ times}}; \ell_0,\underbrace{0,0, \ldots, 0}_{s-1 \text{ times}}\}$  with a generator matrix $\mathtt{H}$ of the form \eqref{Eq03.4}. The code $\mathtt{D}$ is $h$-Galois $\mathtt{R}\check{\mathtt{R}}$-LCD if and only if the codes $\overline{\mathtt{D}^{(X)}}$ and $\overline{\mathtt{D}^{(Y)}}$ with their respective generator matrices $\overline{\mathcal{A}}$ and $\overline{\mathcal{F}}$ are $h$-Galois LCD codes over $\overline{\mathtt{R}}.$
\end{thm}
\begin{proof} Working as in Theorems \ref{Theorem4.1}-\ref{Theorem4.2} and by applying Theorem 2.4 of Liu \etal~\cite{Liu2018} and  Lemma \ref{Lemma4.5}, we get the desired result.
\end{proof}

As a consequence of Theorem \ref{Theorem04.2}, we deduce the following:

\begin{cor}\label{corH}
When $q=|\overline{\mathtt{R}}|>4,$ every weakly-free $\mathtt{R}\check{\mathtt{R}}$-linear code  of block-length $(\mathpzc{a},\mathpzc{b})$ is monomially equivalent to an $h$-Galois  $\mathtt{R}\check{\mathtt{R}}$-LCD code.
\end{cor}
\begin{proof}To prove the result, we assume, 
without any loss of generality,  that the code $\mathtt{C}$ has a generator matrix $\mathtt{G}$ of the form \eqref{Equation4.3}. If the code $\mathtt{C}$ is $h$-Galois $\mathtt{R}\check{\mathtt{R}}$-LCD, then we are through.  

Now, suppose that the code $\mathtt{C}$ is not $h$-Galois $\mathtt{R}\check{\mathtt{R}}$-LCD. Here, by Theorem \ref{Theorem4.2}, we see that at least one of the codes $\overline{\mathtt{C}^{(X)}}$ and $\overline{\mathtt{C}^{(Y)}}$ is not an $h$-Galois LCD code over $\overline{\mathtt{R}}.$ Working as in Theorem 5.1 of Yadav and Sharma \etal~\cite{Yadav1}, we see that there exist an $\mathpzc{a}$-tuple $\mathtt{w}=(w_1,w_2,\ldots,w_{\mathpzc{a}})\in (\overline{\mathtt{R}}\setminus\{0\})^{\mathpzc{a}}$ and a $\mathpzc{b}$-tuple $\mathtt{v}=(v_1,v_2,\ldots,v_\mathpzc{b})\in (\overline{\mathtt{R}}\setminus\{0\})^{\mathpzc{b}}$ such that the codes  $\overline{\mathtt{C}_\mathtt{w}^{(X)}}=\{(w_1c_1,w_2c_2,\ldots,w_\mathpzc{a}c_\mathpzc{a})\in\overline{\mathtt{R}}^\mathpzc{a}:(c_1,c_2,\ldots,c_\mathpzc{a})\in\overline{\mathtt{C}^{(X)}}\}$ and $\overline{\mathtt{C}_\mathtt{v}^{(Y)}}=\{(v_1d_1,v_2d_2,\ldots,v_\mathpzc{b}d_\mathpzc{b})\in\overline{\mathtt{R}}^\mathpzc{b}:(d_1,d_2,\ldots,d_\mathpzc{b})\in\overline{\mathtt{C}^{(Y)}}\}$ are $h$-Galois LCD codes over $\overline{\mathtt{R}}.$ Thus  there exist an $\mathpzc{a}$-tuple $\mathtt{r}=(r_1,r_2,\ldots,r_{\mathpzc{a}}) \in (\mathtt{R}\setminus\gamma\mathtt{R})^{\mathpzc{a}}$ and a $\mathpzc{b}$-tuple $\mathtt{s}=(s_1,s_2,\ldots,s_\mathpzc{b})\in (\check{\mathtt{R}}\setminus\gamma\check{\mathtt{R}})^{\mathpzc{b}}$ satisfying $\overline{r}_i=w_i$ and $\overline{s}_j=v_j$ for $1\leq i\leq \mathpzc{a}$ and $1\leq j\leq \mathpzc{b}.$  By Theorem \ref{Theorem04.2}, we see that the code $\mathtt{D}=\{(r_1\alpha_1,r_2\alpha_2,\ldots,r_\mathpzc{a}\alpha_\mathpzc{a}|s_1\beta_1,s_2\beta_2,\ldots,s_\mathpzc{b}\beta_{\mathpzc{b}})\in\mathbf{M}:(\alpha_1,\alpha_2,\ldots,\alpha_\mathpzc{a}|\beta_1,\beta_2,\ldots,\beta_{\mathpzc{b}})\in\mathtt{C}\}$ is an $h$-Galois $\mathtt{R}\check{\mathtt{R}}$-LCD code. Note that the codes $\mathtt{C}$ and $\mathtt{D}$ are monomially equivalent.
\end{proof}

\begin{cor}\label{Eucequivalence}
When $q=|\overline{\mathtt{R}}|>3,$ every weakly-free $\mathtt{R}\check{\mathtt{R}}$-linear code  of block-length $(\mathpzc{a},\mathpzc{b})$ is monomially equivalent to a Euclidean $\mathtt{R}\check{\mathtt{R}}$-LCD code.
\end{cor}
\begin{proof}Working as in Corollary \ref{corH} and by applying Theorem 11 and Corollary 13 of Carlet \etal~\cite{Carlet1},  Theorem 2.4 of Liu~\etal~\cite{Liu2018} and Theorem \ref{Theorem04.2}, the desired result follows.
\end{proof}

Now the following  corollary is an immediate consequence of Theorem \ref{Theorem4.2} and provide a method to construct $h$-Galois $\mathtt{R}\check{\mathtt{R}}$-LCD codes of block-length $(\mathpzc{a}, \mathpzc{b}).$ 
\begin{cor}\label{Cor4.3}
An $\mathtt{R}\check{\mathtt{R}}$-linear code  $\mathtt{C}(\subseteq\mathbf{M})$    with a generator matrix
\vspace{-2mm}\begin{align*}
\left[\begin{array}{cc|cc}
    I_{k_0} & \gamma A & 0 & B \\ 
    0 & \gamma^{e-s} C & I_{\ell_0} & \gamma D 
\end{array}\right]
\vspace{-2mm}\end{align*}
is $h$-Galois $\mathtt{R}\check{\mathtt{R}}$-LCD, where
$k_0$ and $\ell_0$ are non-negative integers. 
\end{cor}

The following corollary is an extension of Proposition 3.6 of Benbelkacem \etal~\cite{Benbel}. 
\begin{cor}\label{Cor4.4}
Let $\mathtt{D}(\subseteq\mathbf{M})$ be a weakly-free $\mathtt{R}\check{\mathtt{R}}$-linear code of the type $\{k_0,\underbrace{0,0,\ldots, 0}_{e-1 \text{ times}}; \ell_0,\underbrace{0,0, \ldots, 0}_{s-1 \text{ times}}\}$  with a generator matrix $\mathtt{H}$ of the form \eqref{Eq03.4}.  Let $R^{(i)}$  denote the $i$-th row of the matrix $\mathtt{H}$  for $1\leq i\leq k_0+\ell_0.$ Suppose that $\langle R^{(i)},R^{(i)}\rangle_h\notin\gamma\mathtt{R}$ and $\langle R^{(i)},R^{(j)}\rangle_h\in\gamma\mathtt{R}$ for all $i\ne j$ and $1\leq i,j \leq k_0+\ell_0.$ Then we have $\ell_0=0,$ and the code $\mathtt{D}$ is $h$-Galois $\mathtt{R}\check{\mathtt{R}}$-LCD.
\end{cor}
\begin{proof}As $\langle R^{(i)},R^{(i)}\rangle_h\notin\gamma\mathtt{R}$ for $1\leq i\leq k_0+\ell_0,$ we must have $\ell_0=0.$ This gives $\mathtt{H}=[\mathcal{A}~|~\mathcal{B}].$ Now, to prove this result, we see, by Theorem \ref{Theorem04.2}, that it suffices to show that the matrix $\overline{\mathcal{A}}$ generates an $h$-Galois LCD code of length $\mathpzc{a}$ over $\overline{\mathtt{R}}.$
For this, let $R^{(i)}_X$ denote the $i$-th row of the matrix $\mathcal{A}$ for $1\leq i\leq k_0.$ Here we observe that $\langle R^{(i)}_X,R^{(i)}_X\rangle_h\notin\gamma \mathtt{R},$  which implies that $[\overline{R^{(i)}}_X,\overline{R^{(i)}}_X]_h\ne0$ for $1\leq i\leq k_0.$ Further, one can see that $[\overline{R^{(i)}}_X,\overline{R^{(j)}}_X]_h=0$ for all $i\ne j$ and $1\leq i,j\leq k_0.$ This implies that the matrix $\overline{\mathcal{A}}~\overline{\sigma^h(\mathcal{A})}^T$ is invertible. This, by Theorem 2.4 of Liu~\etal~\cite{Liu2018}, further implies that the matrix $\overline{\mathcal{A}}$ generates an $h$-Galois LCD code of length $\mathpzc{a}$ over $\overline{\mathtt{R}},$ which completes the proof. 
\end{proof}
Note that $h$-Galois  $\mathtt{R}\check{\mathtt{R}}$-LCD codes coincide with Euclidean $\mathtt{R}\check{\mathtt{R}}$-LCD codes when $h=0,$ while $h$-Galois  $\mathtt{R}\check{\mathtt{R}}$-LCD codes coincide with Hermitian $\mathtt{R}\check{\mathtt{R}}$-LCD codes when $w$ is even and $h=\frac{w}{2}.$ We will next count all Euclidean and Hermitian $\mathtt{R}\check{\mathtt{R}}$-LCD codes by applying the results derived in this section. These enumeration formula are useful in the classification of these codes up to monomial equivalence, which we will illustrate in certain specific cases (see Tables \ref{table1}-\ref{table3} and Appendices A-C). 
\section{Enumeration of Euclidean $\mathtt{R}\check{\mathtt{R}}$-LCD codes of block-length $(\mathpzc{a}, \mathpzc{b})$}\label{SectionE}

In this section, we enumerate all  Euclidean $\mathtt{R}\check{\mathtt{R}}$-LCD codes of block-length $(\mathpzc{a}, \mathpzc{b}).$ To do this, we first prove the following lemma.

\begin{lemma}\label{Lemma4.8}
Let $E_1$ be a Euclidean LCD code  of length $\mathpzc{a}$ and  dimension $k_0 $ over $\overline{\mathtt{R}},$ and let $E_2$ be a Euclidean LCD code  of length $\mathpzc{b}$ and  dimension $\ell_0 $ over $\overline{\mathtt{R}}.$ There are precisely  $q^{(\mathpzc{a}-k_0)(s\ell_0+(e-1)k_0)+(\mathpzc{b}-\ell_0)(sk_0+(s-1)\ell_0)}$ distinct Euclidean $\mathtt{R}\check{\mathtt{R}}$-LCD codes $\mathtt{C}$ of block-length $(\mathpzc{a}, \mathpzc{b})$ satisfying $\overline{\mathtt{C}^{(X)}}=E_1$ and $\overline{\mathtt{C}^{(Y)}}=E_2.$
\end{lemma}
\begin{proof}
To prove the result, we  assume, without any loss of generality, that the codes $E_1$ and $E_2$ have generator matrices in the standard form, say  $[
    I_{k_0} ~~ F_1
]$ is a generator matrix of $E_1$ and $[
    I_{\ell_0} ~~ F_2
]$ is a generator matrix of $E_2.$ One can easily see that there exists a bijection between the Teichm$\ddot{u}$ller set $\mathcal{T}$  and the residue field $\overline{\mathtt{R}}$ of $\mathtt{R}.$ Thus there exist unique matrices  $A_0$ and $D_0$ over $\mathcal{T}$ satisfying  $\overline{A_0}=F_1$ and $\overline{D_0}=F_2.$ Now, let us define  \vspace{-2mm}$$A=\sum\limits_{i=0}^{e-1}\gamma^iA_i,~~B=\sum\limits_{j=0}^{s-1}\gamma^jB_j,~~\gamma^{e-s} C=\sum\limits_{j=0}^{s-1}\gamma^{e-s+j}C_j\text{~ and ~}D=\sum\limits_{j=0}^{s-1}\gamma^jD_j,\vspace{-2mm}$$ where   $A_1,A_2,\ldots,A_{e-1}\in M_{k_0\times(\mathpzc{a}-k_0)}(\mathcal{T}),$ $B_0,B_1,\ldots,B_{s-1}\in M_{k_0\times(\mathpzc{b}-\ell_0)}(\mathcal{T}),$ $C_0,C_1,\ldots,C_{s-1}\in M_{\ell_0\times(\mathpzc{a}-k_0)}(\mathcal{T})$ and $D_1,D_2,\ldots,D_{s-1}\in M_{\ell_0\times(\mathpzc{b}-\ell_0)}(\mathcal{T}).$ By Theorem \ref{Theorem4.2}, we see that the $\mathtt{R}\check{\mathtt{R}}$-linear code $\mathtt{C}$ of block-length $(\mathpzc{a}, \mathpzc{b})$ with a generator matrix $\left[\begin{array}{cc|cc}
    I_{k_0} & A & 0 & B \\ 
    0 & \gamma^{e-s} C & I_{\ell_0} & D 
\end{array}\right]$ is a Euclidean $\mathtt{R}\check{\mathtt{R}}$-LCD code  satisfying $\overline{\mathtt{C}^{(X)}}=E_1$ and $\overline{\mathtt{C}^{(Y)}}=E_2.$ We further  observe that each of the distinct choices of the matrices $A_1,A_2,\ldots,A_{e-1},$ $ B_0,B_1,\ldots,B_{s-1}, $
$ C_0,C_1,\ldots,C_{s-1},$ $ D_1,D_2,\ldots, D_{s-1}$  gives rise to a distinct Euclidean $\mathtt{R}\check{\mathtt{R}}$-LCD code $\mathtt{C}$ of block-length $(\mathpzc{a}, \mathpzc{b})$ with a generator matrix $\left[\begin{array}{cc|cc}
    I_{k_0} & A & 0 & B \\ 
    0 & \gamma^{e-s} C & I_{\ell_0} & D 
\end{array}\right]$ satisfying $\overline{\mathtt{C}^{(X)}}=E_1$ and $\overline{\mathtt{C}^{(Y)}}=E_2.$ From this, the desired result follows.
\end{proof}

Now, to count all Euclidean $\mathtt{R}\check{\mathtt{R}}$-LCD codes of block-length $(\mathpzc{a},\mathpzc{b}),$ we need  to count all Euclidean LCD codes of an arbitrary length over a finite field. Yadav and Sharma \cite[Th. 3.2 and 3.3]{Yadav1} obtained an explicit enumeration formula for all Euclidean LCD codes of an arbitrary length over the finite field $\mathbb{F}_{q}$ of order $q,$ where $q$ is any prime power.  To recall this enumeration formula, let $L_q(n,r)$ denote the number of distinct Euclidean LCD codes of length $n$ and dimension $r$ over  $\mathbb{F}_q,$  where $r,n$ are integers satisfying $n \geq 1$ and $0\leq r \leq n.$ It is easy to see that $L_q(n,0)=L_q(n,n)=1.$ Throughout the paper, let $\qbin{n}{r}{q}$ denote the Gaussian binomial coefficient (or the $q$-binomial coefficient), which is given by 
 $\qbin{n}{0}{q}=1$ and $\qbin{n}{r}{q}=\frac{(q^n-1)(q^n-q)\ldots (q^n-q^{r-1})}{(q^r-1)(q^r-q)\cdots (q^r-q^{r-1})}$ for $1 \leq r \leq n.$  By Theorems 3.2 and 3.3 of Yadav and Sharma \cite{Yadav1}, we see that the number $L_q(n,r)$ is given by the following for $1 \leq r \leq n-1$:
\begin{itemize}
    \item[(a)] When $q$ is an even prime power, we have
\vspace{-2mm}\begin{equation}\label{L1}
L_q(n,r)=\left\{\begin{array}{llll}  q^{\frac{(n-r)(r+1)}{2}}\left[{\begin{array}{cc}
        (n-1)/2  \\
         (r-1)/2
    \end{array}}\right]_{q^{2}}  & \text{if both } n \text{ and } r \text{ are odd}; \\

 q^{\frac{nr-r^2+n-1}{2}}\left[{\begin{array}{cc}
        (n-2)/2  \\
         (r-1)/2
    \end{array}}\right]_{q^{2}} & \text{if } n \text{ is even and } r \text{ is odd;}\\

 q^{\frac{r(n-r+1)}{2}}\;\left[{\begin{array}{cc}
        (n-1)/2  \\
         r/2
    \end{array}}\right]_{q^{2}}   & \text{if } n \text{ is odd and } r \text{ is even};\\

   q^{\frac{nr-r^2-2}{2}}\bigl( ( q^{r}+ q-1)\left[{\begin{array}{cc}
        (n-2)/2  \\
         r/2
    \end{array}}\right]_{q^{2}}\\+( q^{n-r+1}- q^{n-r}+1)\left[{\begin{array}{cc}
        (n-2)/2  \\
         (r-2)/2
    \end{array}}\right]_{q^{2}}\bigl) & \text{if both } n \text{ and } r \text{ are even}.
  \end{array} \right.
\vspace{-2mm}\end{equation}
 \item[(b)] When $q$ is an odd prime power, we have
\vspace{-2mm}\begin{equation}\label{L2} L_q(n,r)=\left\{\begin{array}{llll}   q^{\frac{(n-r)(r+1)}{2}}\left[{\begin{array}{cc}
        (n-1)/2  \\
         (r-1)/2
\end{array}}\right]_{q^{2}}   &~\text{if both } n \text{ and } r \text{ are odd}; \\
 q^{\frac{nr-r^2-1}{2}} ( q^{\frac{n}{2}}-1)\left[{\begin{array}{cc}
        (n-2)/2  \\
         (r-1)/2
    \end{array}}\right]_{q^{2}} & \begin{array}{l}\text{if }n \text{ is even and }  r \text{ is odd } \text{with either }\\q \equiv 1 \Mod{4} \text{ or } n \equiv 0 \Mod{4} \text{ and } \\q \equiv 3 \Mod{4};\end{array}\\
 q^{\frac{nr-r^2-1}{2}}( q^{\frac{n}{2}}+1)\left[{\begin{array}{cc}
        (n-2)/2  \\
         (r-1)/2
    \end{array}}\right]_{q^{2}}  & \begin{array}{l} \text{if } n \equiv 2 \Mod{4}, r \text{ is odd and}\\
q \equiv 3 \Mod{4} ;\end{array}
\\
 q^{\frac{r(n-r+1)}{2}} \left[{\begin{array}{cc}
        (n-1)/2  \\
         r/2
    \end{array}}\right]_{q^{2}}&~\text{if }  n \text{ is odd and }r \text{ is even};
\\
   q^{\frac{r(n-r)}{2}}\left[{\begin{array}{cc}
        n/2  \\
         r/2
    \end{array}}\right]_{q^{2}} &~\text{if both }n \text{ and } r \text{ are even}.
  
\end{array} \right.
\vspace{-2mm}\end{equation}
\end{itemize}

In the following theorem, we enumerate all  Euclidean $\mathtt{R}\check{\mathtt{R}}$-LCD codes of block-length $(\mathpzc{a},\mathpzc{b})$ and of the type $\{k_0,\underbrace{0,0,\ldots, 0}_{e-1 \text{ times}}; \ell_0,\underbrace{0,0, \ldots, 0}_{s-1 \text{ times}}\}.$
\begin{thm}\label{Theorem4.3}
For non-negative integers $k_0$ and $\ell_0,$ the number $\mathscr{L}(k_0,\ell_0)$ of distinct Euclidean $\mathtt{R}\check{\mathtt{R}}$-LCD codes  of block-length $(\mathpzc{a},\mathpzc{b})$ and of the type $\{k_0,\underbrace{0,0,\ldots, 0}_{e-1 \text{ times}}; \ell_0,\underbrace{0,0, \ldots, 0}_{s-1 \text{ times}}\}$ is given by
\vspace{-2mm}\begin{align*}
\mathscr{L}(k_0,\ell_0)=L_q(\mathpzc{a},k_0)L_q(\mathpzc{b},\ell_0)q^{(\mathpzc{a}-k_0)(s\ell_0+(e-1)k_0)+(\mathpzc{b}-\ell_0)(sk_0+(s-1)\ell_0)},
\vspace{-2mm}\end{align*}
where the numbers $L_q(\mathpzc{a},k_0)$ and $L_q(\mathpzc{b},\ell_0)$ are given by \eqref{L1} and \eqref{L2}.
\end{thm}
\begin{proof}
It follows from Theorem \ref{Theorem4.2} and  Lemma \ref{Lemma4.8}.
\end{proof}
In the following theorem, we enumerate all Euclidean $\mathtt{R}\check{\mathtt{R}}$-LCD codes  of block-length $(\mathpzc{a},\mathpzc{b}).$
\begin{thm}\label{Theorem4.4}
The number $\mathscr{L}$ of distinct Euclidean $\mathtt{R}\check{\mathtt{R}}$-LCD codes of block-length $(\mathpzc{a},\mathpzc{b})$ is given by
\vspace{-2mm}\begin{align*}
\mathscr{L}=\sum\limits_{k_0=0}^{\mathpzc{a}}\sum\limits_{\ell_0=0}^{\mathpzc{b}}L_q(\mathpzc{a},k_0)L_q(\mathpzc{b},\ell_0)q^{(\mathpzc{a}-k_0)(s\ell_0+(e-1)k_0)+(\mathpzc{b}-\ell_0)(sk_0+(s-1)\ell_0)},
\vspace{-2mm}\end{align*}
where the numbers $L_q(\mathpzc{a},k_0)$ and $L_q(\mathpzc{b},\ell_0)$ are given by \eqref{L1} and \eqref{L2}.
\end{thm}
\begin{proof} To prove the result, we see, by Lemma \ref{Lemma4.4},   that $\mathscr{L}=\sum\limits_{k_0=0}^{\mathpzc{a}}\sum\limits_{\ell_0=0}^{\mathpzc{b}}\mathscr{L}(k_0,\ell_0),$ where $\mathscr{L}(k_0,\ell_0)$ equals the number  of distinct Euclidean $\mathtt{R}\check{\mathtt{R}}$-LCD codes  of block-length $(\mathpzc{a},\mathpzc{b})$ and of the type $\{k_0,\underbrace{0,0,\ldots, 0}_{e-1 \text{ times}}; \ell_0,\underbrace{0,0, \ldots, 0}_{s-1 \text{ times}}\}$ for non-negative integers $k_0\leq \mathpzc{a}$ and $\ell_0 \leq \mathpzc{b}.$
Now the desired result follows immediately by applying Theorem \ref{Theorem4.3}. 
\end{proof}
The enumeration formula obtained in Theorem \ref{Theorem4.4} plays a significant role in classifying Euclidean $\mathtt{R}\check{\mathtt{R}}$-LCD codes  of block-lengths $(\mathpzc{a},\mathpzc{b})$  up to monomial equivalence. We illustrate the same in Tables \ref{table1} and \ref{table2},  where we classify all Euclidean $\mathbb{Z}_4 \mathbb{Z}_{2}$-LCD   and $\mathbb{Z}_9 \mathbb{Z}_{3}$-LCD codes of block-lengths $(1,1),$ $(1,2),$ $(2,1),$ $(2,2),$ $ (3,1)$ and $(3,2),$ by carrying out computations in Magma.  Generator matrices and     Lee distances  of all monomially inequivalent  Euclidean $\mathbb{Z}_4 \mathbb{Z}_{2}$-LCD   and $\mathbb{Z}_9 \mathbb{Z}_{3}$-LCD codes of block-lengths $(1,1),$ $(1,2),$ $(2,1),$ $(2,2),$ $ (3,1)$ and $(3,2)$ are provided in Appendices A and B, respectively.  

\begin{table}[h]
   \centering
\small\begin{tabular}{ |c|c|c|c|c|}
 \hline
   Block-length & Number of non-zero Euclidean $\mathbb{Z}_4 \mathbb{Z}_{2}$-LCD  & Number of  monomially inequivalent non-zero \\
      $(\mathpzc{a},\mathpzc{b})$ &  codes of block-length  $(\mathpzc{a},\mathpzc{b})$ &  Euclidean $\mathbb{Z}_4 \mathbb{Z}_{2}$-LCD codes of block-length  $(\mathpzc{a},\mathpzc{b})$ \\  \hline
      $(1,1)$ & $5$ & $5$ \\ \hline
     $(1,2)$ & $17$ & $11$ \\ \hline
     $(2,1)$ & $25$ & $15$ \\ \hline
     $(2,2)$ & $113$ & $41$ \\ \hline
     $(3,1)$ & $209$ & $49$ \\ \hline
     $(3,2)$ & $1301$ & $163$ \\ \hline
 \end{tabular}
     \caption{Number of non-zero monomially inequivalent  Euclidean $\mathbb{Z}_4 \mathbb{Z}_{2}$-LCD codes of block-length $(\mathpzc{a},\mathpzc{b}),$ \\where $\mathpzc{a}\in \{1,2,3\}$ and $\mathpzc{b}\in \{1,2\}$}
 \label{table1}\end{table}

    \begin{table}[h] \centering
\small\begin{tabular}{ |c|c|c|c|c|}
 \hline
   Block-length & Number of non-zero Euclidean & Number of monomially inequivalent non-zero  \\
      $(\mathpzc{a},\mathpzc{b})$ &  $\mathbb{Z}_9 \mathbb{Z}_{3}$-LCD codes of block-length  $(\mathpzc{a},\mathpzc{b})$ &  Euclidean $\mathbb{Z}_9 \mathbb{Z}_{3}$-LCD codes of block-length  $(\mathpzc{a},\mathpzc{b})$ \\ \hline
      $(1,1)$ & $7$ & $5$ \\ \hline
     $(1,2)$ & $43$ & $15$ \\ \hline
     $(2,1)$ & $91$ & $19$ \\ \hline
     $(2,2)$ & $883$ & $71$ \\ \hline
     $(3,1)$ & $1351$ & $53$ \\ \hline
     $(3,2)$ & $33751$ & $336$ \\ \hline
  \end{tabular}
\caption{Number of non-zero monomially inequivalent  Euclidean $\mathbb{Z}_9 \mathbb{Z}_{3}$-LCD codes of block-length $(\mathpzc{a},\mathpzc{b}),$ \\where $\mathpzc{a}\in \{1,2,3\}$ and $\mathpzc{b}\in \{1,2\}$}
    \label{table2}\end{table}


\section{Enumeration of Hermitian $\mathtt{R}\check{\mathtt{R}}$-LCD codes of block-length $(\mathpzc{a}, \mathpzc{b})$}\label{SectionH}
Throughout this section, we assume that $w$ is an even integer and $h=\frac{w}{2}.$ In this section, we enumerate all  Hermitian $\mathtt{R}\check{\mathtt{R}}$-LCD codes of block-length $(\mathpzc{a}, \mathpzc{b}).$ To do this, we first prove the following lemma.

\begin{lemma}\label{Lemma4.11}
Let $F_1$ be a Hermitian LCD code  of length $\mathpzc{a}$ and  dimension $k_0 $ over $\overline{\mathtt{R}},$ and let $F_2$ be a Hermitian LCD code  of length $\mathpzc{b}$ and  dimension $\ell_0 $ over $\overline{\mathtt{R}}.$ There are precisely  $q^{(\mathpzc{a}-k_0)(s\ell_0+(e-1)k_0)+(\mathpzc{b}-\ell_0)(sk_0+(s-1)\ell_0)}$ distinct Hermitian $\mathtt{R}\check{\mathtt{R}}$-LCD codes $\mathtt{C}(\subseteq\mathbf{M})$ satisfying $\overline{\mathtt{C}^{(X)}}=F_1$ and $\overline{\mathtt{C}^{(Y)}}=F_2.$
\end{lemma}
\begin{proof}
Working as in Lemma \ref{Lemma4.8}, the desired result follows.
\end{proof}
Now, to count all Hermitian $\mathtt{R}\check{\mathtt{R}}$-LCD codes of block-length $(\mathpzc{a},\mathpzc{b}),$ we need to count all Hermitian LCD codes of an arbitrary length over a finite field. Yadav and Sharma \cite{Yadav1} obtained an explicit enumeration formula for all Hermitian LCD codes of an arbitrary length over the finite field $\mathbb{F}_{q_1^2}$ of order $q_1^2,$ where $q_1$ is any prime power. To recall this enumeration formula, let  $M_{q_1^2}(n,r)$ denote the number of distinct Hermitian LCD codes of length $n$ and dimension $r$ over $\mathbb{F}_{q_1^2},$ where $r,n$ are integers satisfying $n\geq 1$ and $0\leq r\leq n.$ It is easy to see that $M_{q_1^2}(n,0)=M_{q_1^2}(n,n)=1.$ Now by Theorem 3.6 of Yadav and Sharma \cite{Yadav1}, we see, for $1\leq r\leq n-1,$ that 
\vspace{-2mm}\begin{align}\label{M1}
   M_{q_1^2}(n,r)=q_1^{r(n-r)}\prod_{i=0}^{r-1}\Big(\frac{q_1^{n-i}-(-1)^{n-i}}{q_1^{r-i}-(-1)^{r-i}}\Big). 
\vspace{-2mm}\end{align}

 In the following theorem, we enumerate all Hermitian $\mathtt{R}\check{\mathtt{R}}$-LCD codes  of block-length $(\mathpzc{a},\mathpzc{b})$ and of the type $\{k_0,\underbrace{0,0,\ldots, 0}_{e-1 \text{ times}}; \ell_0,\underbrace{0,0, \ldots, 0}_{s-1 \text{ times}}\}.$

\begin{thm}\label{Theorem4.7}
For non-negative integers $k_0$ and $\ell_0,$ the number $\mathscr{I}(k_0,\ell_0)$ of distinct Hermitian $\mathtt{R}\check{\mathtt{R}}$-LCD codes  of block-length $(\mathpzc{a},\mathpzc{b})$ and of the type $\{k_0,\underbrace{0,0,\ldots, 0}_{e-1 \text{ times}}; \ell_0,\underbrace{0,0, \ldots, 0}_{s-1 \text{ times}}\}$ is given by
\vspace{-2mm}\begin{align*}
\mathscr{I}(k_0,\ell_0)=M_{q_1^2}(\mathpzc{a},k_0)M_{q_1^2}(\mathpzc{b},\ell_0)(q_1^2)^{(\mathpzc{a}-k_0)(s\ell_0+(e-1)k_0)+(\mathpzc{b}-\ell_0)(sk_0+(s-1)\ell_0)},
\vspace{-2mm}\end{align*}
where the residue field $\overline{\mathtt{R}}$ of $\mathtt{R}$ is of order $q=q_1^2$ for some prime power $q_1$ and the numbers $M_{q_1^2}(\mathpzc{a},k_0)$ and $M_{q_1^2}(\mathpzc{b},\ell_0)$ are given by \eqref{M1}.
\end{thm}
\begin{proof}
It follows from Theorem \ref{Theorem4.2} and Lemma \ref{Lemma4.11}.
\end{proof}
In the following theorem, we enumerate all  Hermitian $\mathtt{R}\check{\mathtt{R}}$-LCD codes  of block-length $(\mathpzc{a},\mathpzc{b}).$

\begin{thm}\label{Theorem4.8}
The number $\mathscr{I}$ of distinct Hermitian $\mathtt{R}\check{\mathtt{R}}$-LCD codes of block-length $(\mathpzc{a},\mathpzc{b})$ is given by
\vspace{-2mm}\begin{align*}
\mathscr{I}=\sum\limits_{k_0=0}^{\mathpzc{a}}\sum\limits_{\ell_0=0}^{\mathpzc{b}}M_{q_1^2}(\mathpzc{a},k_0)M_{q_1^2}(\mathpzc{b},\ell_0)q^{(\mathpzc{a}-k_0)(s\ell_0+(e-1)k_0)+(\mathpzc{b}-\ell_0)(sk_0+(s-1)\ell_0)},
\vspace{-2mm}\end{align*}
where the residue field $~\overline{\mathtt{R}}$ of $\mathtt{R}$ is of order $q=q_1^2$ for some prime power $q_1$ and the numbers $M_{q_1^2}(\mathpzc{a},k_0)$ and $M_{q_1^2}(\mathpzc{b},\ell_0)$ are given by \eqref{M1}.
\end{thm}
\begin{proof}
To prove the result, we see, by Lemma \ref{Lemma4.4}, that $\mathscr{I}=\sum\limits_{k_0=0}^{\mathpzc{a}}\sum\limits_{\ell_0=0}^{\mathpzc{b}}\mathscr{I}(k_0,\ell_0),$ where $\mathscr{I}(k_0,\ell_0)$ equals the number  of distinct Hermitian $\mathtt{R}\check{\mathtt{R}}$-LCD codes  of block-length $(\mathpzc{a},\mathpzc{b})$ and of the type $\{k_0,\underbrace{0,0,\ldots, 0}_{e-1 \text{ times}}; \ell_0,\underbrace{0,0, \ldots, 0}_{s-1 \text{ times}}\}$  for non-negative integers $k_0\leq \mathpzc{a}$ and $\ell_0 \leq \mathpzc{b}.$ Now by Theorem \ref{Theorem4.7}, we get the desired result. 
\end{proof}
The enumeration formula obtained in Theorem \ref{Theorem4.8} is useful in classifying Hermitian $\mathtt{R}\check{\mathtt{R}}$-LCD codes  of block-length $(\mathpzc{a},\mathpzc{b})$  up to monomial equivalence. Using Magma, we illustrate the same in Table \ref{table3},  where we classify all Hermitian $\frac{\mathbb{F}_{4}[u]}{\langle u^2\rangle} \;\mathbb{F}_{4}$-LCD    codes   of block-length $(\mathpzc{a},\mathpzc{b})$ with $\mathpzc{a},\mathpzc{b} \in \{1,2\}.$   Generator matrices and Lee distances of all monomially inequivalent Hermitian $\frac{\mathbb{F}_{4}[u]}{\langle u^2\rangle} \;\mathbb{F}_{4}$-LCD   codes   of block-lengths $(1,1),$ $(1,2),$ $(2,1)$ and $(2,2)$ are provided in Appendix C. 
 \begin{table}[h]\begin{center}\small\begin{tabular}{ |c|c|c|c|c|}
 \hline
   Block-length & Number of  non-zero Hermitian  & Number of monomially inequivalent non-zero \\
      $(\mathpzc{a},\mathpzc{b})$ &  $\frac{\mathbb{F}_{4}[u]}{\langle u^2\rangle} \;\mathbb{F}_{4}$-LCD codes of block-length $(\mathpzc{a},\mathpzc{b})$ &  Hermitian $\frac{\mathbb{F}_{4}[u]}{\langle u^2\rangle} \;\mathbb{F}_{4}$-LCD codes of block-length $(\mathpzc{a},\mathpzc{b})$ \\ \hline
      $(1,1)$ & $9$ & $5$ \\ \hline
     $(1,2)$ & $65$ & $11$ \\ \hline
     $(2,1)$ & $225$ & $15$ \\ \hline
     $(2,2)$ & $3777$ & $43$ \\ \hline
  \end{tabular}\end{center}
\vspace{-2mm}\caption{Number of non-zero monomially inequivalent  Hermitian $\frac{\mathbb{F}_{4}[u]}{\langle u^2\rangle} \;\mathbb{F}_{4}$-LCD codes of block-length $(\mathpzc{a},\mathpzc{b}),$\\ where $\mathpzc{a}, \mathpzc{b} \in \{1,2\}$ }
    \label{table3}
\vspace{-2mm}\end{table}\normalsize
We next recall that if an $\mathtt{R}\check{\mathtt{R}}$-linear code $\mathtt{C}$ is  $h$-Galois $\mathtt{R}\check{\mathtt{R}}$-LCD, then the code $\mathtt{C}$ and its $h$-Galois dual code $\mathtt{C}^{\perp_h}$ form an LCP of codes. In the next section, we will study and characterize LCPs of $\mathtt{R}\check{\mathtt{R}}$-linear codes of an arbitrary block-length. 
We will also study a direct sum masking scheme constructed using an LCP of $\mathtt{R}\check{\mathtt{R}}$-linear codes and obtain its security threshold against fault injection and side-channel attacks.  We will also discuss another application of LCPs of  $\mathtt{R}\check{\mathtt{R}}$-linear codes to the noiseless two-user adder channel.
\section{LCPs of $\mathtt{R}\check{\mathtt{R}}$-linear codes}\label{LCP}
  In this section, we first recall, from Section \ref{Preliminaries}, that two $\mathtt{R}\check{\mathtt{R}}$-linear codes $\mathtt{C}$ and $\mathtt{D}$ of block-length $(\mathpzc{a},\mathpzc{b})$ form a linear complementary pair of codes (or an LCP of codes) if their direct sum $\mathtt{C}\oplus \mathtt{D}=\mathbf{M}.$ By Lemma 2,  Theorem 1 and Remark 1 of Bajalan~\etal~\cite{Bajalan2023}, we observe the following:

\begin{lemma}\cite{Bajalan2023}\label{Lem6.1}
Let $(\mathtt{C},\mathtt{D})$ be an LCP of $\mathtt{R}\check{\mathtt{R}}$-linear codes of block-length $(\mathpzc{a},\mathpzc{b}).$ The following hold.
\begin{enumerate}\item[(a)] Both the codes $\mathtt{C}$ and $\mathtt{D}$ are weakly-free. 
\item[(b)] If the code $\mathtt{C}$ is of the type $\{k_0,\underbrace{0,0,\ldots, 0}_{e-1 \text{ times}}; \ell_0,\underbrace{0,0, \ldots, 0}_{s-1 \text{ times}}\}$ for some non-negative integers $k_0$ and $\ell_0,$ then the code $\mathtt{D}$ is of the type $\{\mathpzc{a}-k_0,\underbrace{0,0,\ldots, 0}_{e-1 \text{ times}}; \mathpzc{b}-\ell_0,\underbrace{0,0, \ldots, 0}_{s-1 \text{ times}}\}.$
\item[(c)] Further, the codes $\mathtt{C}$ and $\mathtt{D}$ have generator matrices $\mathcal{G}$ and $\mathcal{H}$  of the forms
\vspace{-2mm}\begin{align}\label{Eq6.1}
\mathcal{G}=\left[\begin{array}{c|c}
    \mathscr{A}&\mathscr{B} \\ 
    \gamma^{e-s}\mathscr{C}& \mathscr{D} 
\end{array}\right]\text{ ~ and ~ }\mathcal{H}=\left[\begin{array}{c|c}
    \mathscr{E}&\mathscr{F} \\ 
    \gamma^{e-s}\mathscr{G}& \mathscr{H} 
\end{array}\right],
\vspace{-2mm}\end{align}
respectively, where the matrices $\mathscr{A}\in M_{k_0\times \mathpzc{a}}(\mathtt{R})$ and $\mathscr{E}\in M_{(\mathpzc{a}-k_0)\times \mathpzc{a}}(\mathtt{R}) $ are such that their rows are independent over $\mathtt{R}$ and have period $\gamma^e,$ the matrices $\mathscr{D}\in M_{\ell_0\times \mathpzc{b}}(\check{\mathtt{R}})$ and  $\mathscr{H}\in M_{(\mathpzc{b}-\ell_0)\times \mathpzc{b}}(\check{\mathtt{R}})$ are such that their rows are independent over $\check{\mathtt{R}}$ and have period $\gamma^s,$ and the matrices $ \mathscr{B}\in M_{k_0\times \mathpzc{b}}(\check{\mathtt{R}}),$ $\gamma^{e-s}\mathscr{C}\in M_{\ell_0\times \mathpzc{a}}(\mathtt{R}),$  $\mathscr{F}\in M_{(\mathpzc{a}-k_0)\times \mathpzc{b}}(\check{\mathtt{R}})$ and $ \gamma^{e-s}\mathscr{G}\in M_{(\mathpzc{b}-\ell_0)\times \mathpzc{a}}(\mathtt{R}).$  \end{enumerate}
\end{lemma}
\begin{proof}Part (a) follows from Lemma 2 of Bajalan~\etal~\cite{Bajalan2023} and using the fact that  $\mathbf{M}$ is weakly-free. To prove (b), let us suppose that the codes $\mathtt{C}$ and $\mathtt{D}$ are of the types $\{k_0,\underbrace{0,0,\ldots, 0}_{e-1 \text{ times}}; \ell_0,\underbrace{0,0, \ldots, 0}_{s-1 \text{ times}}\}$ and $\{k'_0,\underbrace{0,0,\ldots, 0}_{e-1 \text{ times}}; \ell'_0,\underbrace{0,0, \ldots, 0}_{s-1 \text{ times}}\},$ respectively, where $k_0,k'_0, \ell_0,\ell'_0$ are non-negative integers. Since $\mathbf{M}$ is of the type $\{\mathpzc{a},\underbrace{0,0,\ldots, 0}_{e-1 \text{ times}}; \mathpzc{b},\underbrace{0,0, \ldots, 0}_{s-1 \text{ times}}\}$ and $\mathtt{C}\oplus \mathtt{D}=\mathbf{M},$ we see, by Theorem 1 of Bajalan~\etal~\cite{Bajalan2023}, that $k_0+k'_0=\mathpzc{a}$ and $\ell_0+\ell'_0=\mathpzc{b},$ from which part (b) follows immediately. Part (c) follows immediately by part (b) and Remark 1 of Bajalan~\etal~\cite{Bajalan2023}.  \end{proof}

From now on, we assume, throughout this section, that $\mathtt{C}$ and $\mathtt{D}$ are weakly-free $\mathtt{R}\check{\mathtt{R}}$-linear codes of block-length $(\mathpzc{a},\mathpzc{b})$ with generator matrices $\mathcal{G}$ and $\mathcal{H}$ (as defined by \eqref{Eq6.1}), respectively. Bajalan \etal~\cite[Th. 1]{Bajalan2023}  recently proved that the codes $\mathtt{C}$ and $\mathtt{D}$ form an LCP if and only if $|\mathtt{C}||\mathtt{D}|=|\mathbf{M}|$ and $\iota(\mathcal{W})$ is a non-singular matrix over $\mathtt{R},$ which holds if and only if  the matrix $\mathcal{W}$ generates $\mathbf{M}$  and $\iota(\mathcal{W})$ is a non-singular matrix over $\mathtt{R},$ where $\mathcal{W}=\left[\begin{array}{c}
    \mathcal{G} \\ 
    \mathcal{H} 
\end{array}\right]$ and the map $\iota$ is an embedding  of $\check{\mathtt{R}}$ into $\mathtt{R},$ defined as  $a\mapsto\gamma^{e-s}\iota(a)$ for all $a\in\check{\mathtt{R}}.$  In another recent work, Liu and Hu \cite[Th. 3.19]{Liu2024} derived  a similar characterization for LCPs of $\mathbb{Z}_4\mathbb{Z}_2$-linear codes.  

In this section, we  will also derive necessary and sufficient conditions under which the  codes $\mathtt{C}$ and $\mathtt{D}$ form an LCP. Towards this, we see, by Proposition 2 of Bajalan~\etal~\cite{Bajalan2023}, that  the $0$-Galois (Euclidean) dual codes $\mathtt{C}^{\perp_0}$ and $\mathtt{D}^{\perp_0}$ are of the types $\{\mathpzc{a}-k_0,\underbrace{0,0,\ldots, 0}_{e-1 \text{ times}}; \mathpzc{b}-\ell_0,\underbrace{0,0, \ldots, 0}_{s-1 \text{ times}}\}$ and $\{k_0,\underbrace{0,0,\ldots, 0}_{e-1 \text{ times}}; \ell_0,\underbrace{0,0, \ldots, 0}_{s-1 \text{ times}}\}$  and have generator matrices $\hat{\mathcal{G}}$ and $\hat{\mathcal{H}}$  of the forms
\vspace{-2mm}\begin{align}\label{Eq6.2}
\hat{\mathcal{G}}=\left[\begin{array}{c|c}
    \mathscr{N}&\mathscr{K} \\ 
    \gamma^{e-s}\mathscr{P}& \mathscr{Q} 
\end{array}\right]\text{ ~ and ~}\hat{\mathcal{H}}=\left[\begin{array}{c|c}
    \mathscr{R}&\mathscr{S} \\ 
    \gamma^{e-s}\mathscr{T}& \mathscr{U} 
\end{array}\right],
\vspace{-2mm}\end{align}
respectively, where the matrices $\mathscr{N}\in M_{(\mathpzc{a}-k_0)\times \mathpzc{a}}(\mathtt{R})$ and $\mathscr{R}\in M_{k_0\times \mathpzc{a}}(\mathtt{R})$ are such that their rows are independent over $\mathtt{R}$ and have period $\gamma^e,$ the matrices $\mathscr{Q}\in M_{(\mathpzc{b}-\ell_0)\times \mathpzc{b}}(\check{\mathtt{R}})$ and $\mathscr{U}\in M_{\ell_0\times \mathpzc{b}}(\check{\mathtt{R}})$ are such that their rows are independent over $\check{\mathtt{R}}$ and have period $\gamma^s,$ and the matrices $\mathscr{K}\in M_{(\mathpzc{a}-k_0)\times \mathpzc{b}}(\check{\mathtt{R}}),$ $\gamma^{e-s}\mathscr{P}\in M_{(\mathpzc{b}-\ell_0)\times \mathpzc{a}}(\mathtt{R}),$ $\mathscr{S}\in M_{k_0\times \mathpzc{b}}(\check{\mathtt{R}})$ and $\gamma^{e-s}\mathscr{T}\in M_{\ell_0\times \mathpzc{a}}(\mathtt{R}).$ The matrices $\hat{\mathcal{G}}$ and $\hat{\mathcal{H}}$ are called parity-check matrices of the codes $\mathtt{C}$ and $\mathtt{D},$ respectively. It is easy to see that $m\diamond \hat{\mathcal{G}}^T=0$ if and only if $m\in\mathtt{C}.$  Similarly, one can see that $m\diamond\hat{\mathcal{H}}^T=0$ if and only if $m\in\mathtt{D}.$   From this, it follows that  \vspace{-2mm}\begin{equation}\label{PCM}\mathcal{G}\diamond\hat{\mathcal{G}}^T=0\text{~ and ~}\mathcal{H}\diamond\hat{\mathcal{H}}^T=0.\vspace{-1mm}\end{equation}
We next observe that the matrices $\mathcal{G}\diamond \hat{\mathcal{H}}^T \in M_{(k_0+\ell_0)\times(k_0+\ell_0)}(\mathtt{R})$ and $\mathcal{H}\diamond\hat{\mathcal{G}}^T \in M_{(\mathpzc{a}-k_0+\mathpzc{b}-\ell_0)\times(\mathpzc{a}-k_0+\mathpzc{b}-\ell_0)}(\mathtt{R})$ are given by
\vspace{-2mm}\begin{align}\label{Eq6.3}
&\mathcal{G}\diamond\hat{\mathcal{H}}^T =\left[\begin{array}{cc}
     \mathscr{A}\mathscr{R}^T+\gamma^{e-s}\mathscr{B}\mathscr{S}^T & \gamma^{e-s}(\mathscr{A}\mathscr{T}^T+\mathscr{B}\mathscr{U}^T)  \\  \gamma^{e-s}(\mathscr{C}\mathscr{R}^T+\mathscr{D}\mathscr{S}^T) & \gamma^{e-s}(\mathscr{D}\mathscr{U}^T+\gamma^{e-s}\mathscr{C}\mathscr{T}^T)
\end{array}\right] \text{~~and}\\
&\mathcal{H}\diamond\hat{\mathcal{G}}^T =\left[\begin{array}{cc}
     \mathscr{E}\mathscr{N}^T+\gamma^{e-s}\mathscr{F}\mathscr{K}^T & \gamma^{e-s}(\mathscr{E}\mathscr{P}^T+\mathscr{F}\mathscr{Q}^T)  \\  \gamma^{e-s}(\mathscr{G}\mathscr{N}^T+\mathscr{H}\mathscr{K}^T) & \gamma^{e-s}(\mathscr{H}\mathscr{Q}^T+\gamma^{e-s}\mathscr{G}\mathscr{P}^T)\label{Eq6.4}
\end{array}\right].
\vspace{-2mm}\end{align}
Clearly, the square matrices $\mathcal{G}\diamond\hat{\mathcal{H}}^T $ and $\mathcal{H}\diamond\hat{\mathcal{G}}^T$ generate linear codes of lengths $k_0+\ell_0$ and $\mathpzc{a}-k_0+\mathpzc{b}-\ell_0$ over $\mathtt{R},$  respectively. In the following proposition, we derive two equivalent necessary and sufficient conditions under which the   codes $\mathtt{C}$ and $\mathtt{D}$  form an LCP by investigating the linear codes generated by the matrices $\mathcal{G}\diamond\hat{\mathcal{H}}^T$ and $\mathcal{H}\diamond\hat{\mathcal{G}}^T$ over $\mathtt{R}.$
\begin{prop}\label{Thm6.1}
Let $\mathtt{C}$ and $\mathtt{D}$ be $\mathtt{R}\check{\mathtt{R}}$-linear codes of block-length $(\mathpzc{a},\mathpzc{b})$ with generator matrices $\mathcal{G}$ and $\mathcal{H}$ (as defined by \eqref{Eq6.1}) and parity-check matrices $\hat{\mathcal{G}}$ and $\hat{\mathcal{H}}$  (as defined by \eqref{Eq6.2}), respectively. The following three statements are equivalent:
\begin{itemize}
\item[(a)] The codes $\mathtt{C}$ and $\mathtt{D}$ form an LCP.
\item[(b)] The matrix $\mathcal{G}\diamond\hat{\mathcal{H}}^T$ is a generator matrix of a linear code of the type $\{k_0,0,0,\ldots,0,\underbrace{\ell_0}_{(e-s+1)\text{-th}},0,0,\ldots,0\}$ and length $k_0+\ell_0$ over $\mathtt{R}.$
\item[(c)] The matrix $\mathcal{H}\diamond\hat{\mathcal{G}}^T$ is a generator matrix of a linear code of the type $\{\mathpzc{a}-k_0,0,0,\ldots,0,\underbrace{\mathpzc{b}-\ell_0}_{(e-s+1)\text{-th}},0,0,\ldots,0\}$ and length $\mathpzc{a}-k_0+\mathpzc{b}-\ell_0$ over $\mathtt{R}.$
\end{itemize}
\end{prop}
\begin{proof} Working as in Theorem \ref{Theorem4.1}, one can show that the statements (a) and (b) are equivalent and that the statements (a) and (c) are equivalent. From this, the desired result follows. 
\vspace{-3mm}\end{proof}

\begin{remark}\label{Rk6.2} Note that the matrix $\mathcal{G}\diamond\hat{\mathcal{H}}^T$ (given by \eqref{Eq6.3}) is a generator matrix of a linear code of the type $\{k_0,0,0,\ldots,0,\underbrace{\ell_0}_{(e-s+1)\text{-th}},0,0,\ldots,0\}$ and length $k_0+\ell_0$ over $\mathtt{R}$ if and only if there exists an invertible $(k_0+\ell_0)\times(k_0+\ell_0)$ matrix $P_1$ over $\mathtt{R}$ such that $(\mathcal{G}\diamond\hat{\mathcal{H}}^T) P_1 =\left[\begin{array}{c|c}
    I_{k_0} & 0 \\
    0 & \gamma^{e-s} I_{\ell_0}
\end{array} \right].$ Similarly, the matrix $\mathcal{H}\diamond\hat{\mathcal{G}}^T$ (given by \eqref{Eq6.4}) is a generator matrix of a linear code of the type $\{\mathpzc{a}-k_0,0,0,\ldots,0,\underbrace{\mathpzc{b}-\ell_0}_{(e-s+1)\text{-th}},0,0,\ldots,0\}$ and length $\mathpzc{a}-k_0+\mathpzc{b}-\ell_0$ over $\mathtt{R}$ if and only if there exists an invertible $(\mathpzc{a}-k_0+\mathpzc{b}-\ell_0)\times(\mathpzc{a}-k_0+\mathpzc{b}-\ell_0)$ matrix $P_2$ over $\mathtt{R}$ such that $(\mathcal{H}\diamond\hat{\mathcal{G}}^T) P_2 =\left[\begin{array}{c|c}
    I_{\mathpzc{a}-k_0} & 0 \\
    0 & \gamma^{e-s} I_{\mathpzc{b}-\ell_0}
\end{array} \right].$   
\end{remark}

In the following proposition, we derive a necessary and sufficient  condition under which the matrix $\mathcal{G}\diamond \hat{\mathcal{H}}^T$ (given by \eqref{Eq6.3}) is a generator matrix of a linear code of the type $\{k_0,0,0,\ldots,0,\underbrace{\ell_0}_{(e-s+1)\text{-th}},0,0,\ldots,0\}$ and length $k_0+\ell_0$ over $\mathtt{R}.$ In an analogous way,  we derive a necessary and sufficient  condition under which the matrix $\mathcal{H}\diamond \hat{\mathcal{G}}^T$  (given by \eqref{Eq6.4}) is a generator matrix of a linear code of the type $\{\mathpzc{a}-k_0,0,0,\ldots,0,\underbrace{\mathpzc{b}-\ell_0}_{(e-s+1)\text{-th}},0,0,\ldots,0\}$ and length $\mathpzc{a}-k_0+\mathpzc{b}-\ell_0$ over $\mathtt{R}.$
\begin{prop}\label{Thm6.2}
Let $\mathcal{G}$ and $\mathcal{H}$ be the matrices as defined by \eqref{Eq6.1}, and let $\hat{\mathcal{G}}$ and $\hat{\mathcal{H}}$ be the matrices as defined by \eqref{Eq6.2}. The following hold.
\begin{itemize}
\item[(a)] The matrix $\mathcal{G}\diamond\hat{\mathcal{H}}^T$ is a generator matrix of a linear code of the type $\{k_0,0,0,\ldots,0,\underbrace{\ell_0}_{(e-s+1)\text{-th}},0,0,\ldots,0\}$ and length $k_0+\ell_0$ over $\mathtt{R}$ if and only if the matrices $\overline{\mathscr{A}}~\overline{\mathscr{R}}^T$ and $\overline{\mathscr{D}}~\overline{\mathscr{U}}^T$
are invertible over  $\overline{\mathtt{R}}.$
\item[(b)] The matrix $\mathcal{H}\diamond\hat{\mathcal{G}}^T$ is a generator matrix of a linear code of the type $\{\mathpzc{a}-k_0,0,0,\ldots,0,\underbrace{\mathpzc{b}-\ell_0}_{(e-s+1)\text{-th}},0,0,\ldots,0\}$ and length $\mathpzc{a}-k_0+\mathpzc{b}-\ell_0$ over $\mathtt{R}$ if and only if the matrices $\overline{\mathscr{E}}~\overline{\mathscr{N}}^T$ and $\overline{\mathscr{H}}~\overline{\mathscr{Q}}^T$
are invertible over  $\overline{\mathtt{R}}.$
\end{itemize}
\end{prop}
\begin{proof}
Working as in Theorem \ref{Lemma4.7}, we get the desired result.
\end{proof}
Two linear codes $U_1$ and $U_2$ of length $n$ over $\overline{\mathtt{R}}$ form an LCP if their direct sum $U_1\oplus U_2=\overline{\mathtt{R}}^n.$ A parity-check matrix of a linear code $U$ of length $n$ over $\overline{\mathtt{R}}$ is defined as a generator matrix of its $0$-Galois (or Euclidean) dual code $U^{\perp_0}.$  Now, the following lemma states two equivalent necessary and sufficient conditions under which two linear codes over $\overline{\mathtt{R}}$ form an LCP. 
\begin{lemma}\cite{Hu2021}\label{Lem6.2}
Let $U_1$ and $U_2$ be two linear codes of length $n$ over $\overline{\mathtt{R}}$  with generator matrices $G_1$ and $G_2$ and parity-check matrices $H_1$ and $H_2$, respectively.  The following three statements are equivalent:
\begin{itemize}
\item[(a)] The codes $U_1$ and $U_2$ form an LCP.
\item[(b)] The matrix $G_1H_2^T$ is invertible.
\item[(c)] The matrix $H_1G_2^T$ is invertible.
\end{itemize}
\end{lemma}
\begin{proof}
It follows from Theorem 2.10 of Hu and Liu \cite{Hu2021}.
\end{proof}

For $\mathtt{R}\check{\mathtt{R}}$-linear codes $\mathtt{C}$ and $\mathtt{D}$ with generator matrices $\mathcal{G}$ and $\mathcal{H}$ (as defined by \eqref{Eq6.1}) respectively, let $\overline{\mathtt{C}^{(X)}}$ and $\overline{\mathtt{D}^{(X)}}$ be linear codes of length $\mathpzc{a}$ over $\overline{\mathtt{R}}$ with generator matrices  $\overline{\mathscr{A}}$ and $\overline{\mathscr{E}},$ respectively, and let   $\overline{\mathtt{C}^{(Y)}}$  and $\overline{\mathtt{D}^{(Y)}}$ be  linear codes of length $\mathpzc{b}$ over $\overline{\check{\mathtt{R}}}=\overline{\mathtt{R}}$ with generator matrices $\overline{\mathscr{D}}$ and $\overline{\mathscr{H}},$ respectively.  We next make the following observation.
\begin{lemma}\label{Lem6.3}
Let $\mathtt{C}$ and $\mathtt{D}$ be $\mathtt{R}\check{\mathtt{R}}$-linear codes of block-length $(\mathpzc{a},\mathpzc{b})$ with generator matrices $\mathcal{G}$ and $\mathcal{H}$ (as defined by \eqref{Eq6.1}) and parity-check matrices $\hat{\mathcal{G}}$ and $\hat{\mathcal{H}}$ (as defined by \eqref{Eq6.2}), respectively. Then the linear codes $\overline{\mathtt{C}^{(X)}},$  $\overline{\mathtt{C}^{(Y)}},$ $\overline{\mathtt{D}^{(X)}}$ and $\overline{\mathtt{D}^{(Y)}}$ have parity-check matrices 
$\overline{\mathscr{N}},$ $\overline{\mathscr{Q}},$ $\overline{\mathscr{R}}$ and $\overline{\mathscr{U}},$  respectively.  
\end{lemma}
\begin{proof}
 Since $\hat{\mathcal{G}}$ is a parity-check matrix of the code $\mathtt{C},$ we have $\mathcal{G}\diamond\hat{\mathcal{G}}^T=0,$ which implies that $\mathscr{A}\mathscr{N}^T+\gamma^{e-s}\mathscr{B}\mathscr{K}^T=0$ and $\gamma^{e-s}(\mathscr{D}\mathscr{Q}^T+\gamma^{e-s}\mathscr{C}\mathscr{P}^T)=0$ in $\mathtt{R}.$ From this, we get $\overline{\mathscr{A}}~\overline{\mathscr{N}}^T=0$ and $\overline{\mathscr{D}}~\overline{\mathscr{Q}}^T=0$ in $\overline{\mathtt{R}}.$ Now, let $E_1$ and $E_2$ be linear codes of lengths $\mathpzc{a}$ and $\mathpzc{b}$ over $\overline{\mathtt{R}}$ generated by the matrices $\overline{\mathscr{N}}$ and $\overline{\mathscr{Q}},$ respectively. It is easy to see that $E_1\subseteq\overline{\mathtt{C}^{(X)}}^{\perp_0}$ and $E_2\subseteq\overline{\mathtt{C}^{(Y)}}^{\perp_0}.$ Further, since the rows of the matrix $\mathscr{N}$ are independent over $\mathtt{R}$ and have period $\gamma^e,$ we see that the rows of the matrix $\overline{\mathscr{N}}$ are linearly independent over $\overline{\mathtt{R}}$, and hence the linear code $E_1$ has dimension $\mathpzc{a}-k_0.$ Similarly, as the rows of the matrix $\mathscr{Q}$ are independent over $\check{\mathtt{R}}$ and have period $\gamma^s,$ we see that the rows of the matrix $\overline{\mathscr{Q}}$ are linearly independent over $\overline{\mathtt{R}}$, and hence the linear code $E_2$ has dimension $\mathpzc{b}-\ell_0.$ Now, since the dual codes $\overline{\mathtt{C}^{(X)}}^{\perp_0}$ and $\overline{\mathtt{C}^{(Y)}}^{\perp_0}$ have dimensions $\mathpzc{a}-k_0$ and $\mathpzc{b}-\ell_0$ respectively, we get $\overline{\mathtt{C}^{(X)}}^{\perp_0}=E_1$ and $\overline{\mathtt{C}^{(Y)}}^{\perp_0}=E_2.$ From this, it follows that the codes $\overline{\mathtt{C}^{(X)}}$ and $\overline{\mathtt{C}^{(Y)}}$ have parity-check matrices $\overline{\mathscr{N}}$ and $\overline{\mathscr{Q}},$  respectively. 
 
 Working as above, one can show that $\overline{\mathscr{R}}$ and $\overline{\mathscr{U}}$ are parity-check matrices of the linear codes $\overline{\mathtt{D}^{(X)}}$ and $\overline{\mathtt{D}^{(Y)}},$ respectively. 
\end{proof}
In the following theorem, we  provide another characterization of LCPs of $\mathtt{R}\check{\mathtt{R}}$-linear codes of block-length $(\mathpzc{a},\mathpzc{b}).$
\begin{thm}\label{Thm6.3}
Let $\mathtt{C}$ and $\mathtt{D}$ be $\mathtt{R}\check{\mathtt{R}}$-linear codes of block-length $(\mathpzc{a},\mathpzc{b})$ with generator matrices $\mathcal{G}$ and $\mathcal{H}$ (as defined by \eqref{Eq6.1}) and parity-check matrices $\hat{\mathcal{G}}$ and $\hat{\mathcal{H}}$ (as defined by \eqref{Eq6.2}), respectively. The following  two statements are equivalent:
\begin{itemize}
\item[(a)] The codes $\mathtt{C}$ and $\mathtt{D}$ form an LCP.
\item[(b)] The codes $\overline{\mathtt{C}^{(X)}}$ and $\overline{\mathtt{D}^{(X)}}$ form an LCP and the codes $\overline{\mathtt{C}^{(Y)}}$ and $\overline{\mathtt{D}^{(Y)}}$ form an LCP.
\end{itemize}
\end{thm}
\begin{proof}
It follows from  Propositions \ref{Thm6.1} and \ref{Thm6.2} and Lemmas \ref{Lem6.2} and \ref{Lem6.3}.
\end{proof}
Note that the above theorem provides a characterization of LCPs of $\mathtt{R}\check{\mathtt{R}}$-linear codes, which is different from the one provided in Theorem 1 of Bajalan \etal~\cite{Bajalan2023} and Theorem 3.19 of Liu and Hu \cite{Liu2024}. 

Now, an $\mathtt{R}\check{\mathtt{R}}$-linear code $\mathtt{A}$ of block-length $(\mathpzc{a},\mathpzc{b})$ is said to be separable if it has a generator matrix of the form $\mathcal{V}=\left[\begin{array}{c|c}
    \mathscr{W}_1&0 \\ 
    0& \mathscr{W}_2 
\end{array}\right],$ where the matrices $\mathscr{W}_1\in M_{n_1\times\mathpzc{a}}(\mathtt{R})$ and $\mathscr{W}_2\in M_{n_2\times\mathpzc{b}}(\check{\mathtt{R}})$ for some non-negative integers $n_1$ and $n_2.$  Recently, Bajalan \etal~\cite[Lem. 6]{Bajalan2023}  derived a necessary and sufficient condition under which two separable $\mathtt{R}\check{\mathtt{R}}$-linear codes  form an LCP. In another recent work, Liu and Hu \cite[Th. 4.13]{Liu2024} also provided a similar characterization of LCPs of separable $\mathbb{Z}_4\mathbb{Z}_2$-linear codes.  In the following corollary, we deduce Lemma 6 of Bajalan \etal~\cite{Bajalan2023} from Theorem \ref{Thm6.3} as a special case.
\begin{cor}\label{cor6.1}
Let $\mathtt{A}_1$ and $\mathtt{A}_2$ be linear codes of length $\mathpzc{a}$ over $\overline{\mathtt{R}}$  with generator matrices $\mathcal{V}_1$ and $\mathcal{V}_2,$  respectively. Also, let $\mathtt{B}_1$ and $\mathtt{B}_2$ be linear codes of length $\mathpzc{b}$ over $\overline{\mathtt{R}}$  with generator matrices $\mathcal{W}_1$ and $\mathcal{W}_2,$  respectively. Next, let $\mathtt{A}$ and $\mathtt{B}$ be $\mathtt{R}\check{\mathtt{R}}$-linear codes  of block-length $(\mathpzc{a}, \mathpzc{b})$  with generator matrices $\mathcal{Y}_1=\left[\begin{array}{c|c}
    \mathscr{Y}_1&0 \\ 
    0& \mathscr{Z}_1 
\end{array}\right]$ and $\mathcal{Y}_2=\left[\begin{array}{c|c}
    \mathscr{Y}_2&0 \\ 
    0& \mathscr{Z}_2 
\end{array}\right]$  respectively, where $\overline{\mathscr{Y}}_1=\mathcal{V}_1,$ $\overline{\mathscr{Y}}_2=\mathcal{V}_2,$ $\overline{\mathscr{Z}}_1=\mathcal{W}_1$ and $\overline{\mathscr{Z}}_2=\mathcal{W}_2.$ Then the codes $\mathtt{A}$ and $\mathtt{B}$  form an LCP if and only if the codes $\mathtt{A}_1$ and $\mathtt{A}_2$ form an LCP over $\overline{\mathtt{R}}$ and the codes $\mathtt{B}_1$ and $\mathtt{B}_2$ form an LCP over $\overline{\mathtt{R}}.$
\end{cor}
\begin{proof}
It follows immediately from Theorem \ref{Thm6.3}.
\end{proof}

Liu and Hu \cite[Th. 3.20]{Liu2024} derived a sufficient condition under which two $\mathbb{Z}_4\mathbb{Z}_2$-linear codes form an LCP. We will now apply Theorem \ref{Thm6.3} to extend this  result to $\mathtt{R}\check{\mathtt{R}}$-linear codes in the following corollary, which  also provides a method to construct LCPs of $\mathtt{R}\check{\mathtt{R}}$-linear codes.

\begin{cor}
Let $\mathtt{C}$ and $\mathtt{D}$ be $\mathtt{R}\check{\mathtt{R}}$-linear codes such that the code $\mathtt{C}$ has a generator matrix $\mathcal{G}$ (as defined  by \eqref{Eq6.1}) and the code $\mathtt{D}$ has a parity-check matrix $\hat{\mathcal{H}}$ (as defined by \eqref{Eq6.2}). 
For $1\leq i\leq k_0+\ell_0,$ let $R^{(i)}$ and $S^{(i)}$   denote the $i$-th rows of the matrices $\mathcal{G}$ and $\hat{\mathcal{H}},$ respectively. Suppose that $\langle R^{(i)},S^{(i)}\rangle_0\notin\gamma\mathtt{R}$ and $\langle R^{(i)},S^{(j)}\rangle_0\in\gamma\mathtt{R}$ for all $i\ne j$ and $1\leq i,j \leq k_0+\ell_0.$ Then we have $\ell_0=0,$ and the codes $\mathtt{C}$ and $\mathtt{D}$ form an LCP of $\mathtt{R}\check{\mathtt{R}}$-linear codes.
\end{cor}
\begin{proof}
As $\langle R^{(i)},S^{(i)}\rangle_h\notin\gamma\mathtt{R}$ for $1\leq i\leq k_0+\ell_0,$ we must have $\ell_0=0.$ This gives $\mathcal{G}=[\mathscr{A}|~\mathscr{B}]$ and $\hat{\mathcal{H}}=[\mathscr{R}|~\mathscr{S}].$ Now, to prove that the codes $\mathtt{C}$ and $\mathtt{D}$ form an LCP, we see, by Lemmas \ref{Lem6.2} and \ref{Lem6.3} and Theorem \ref{Thm6.3}, that it is enough to show that the matrix $\overline{\mathscr{A}}~\overline{\mathscr{R}}^T$ is invertible.
Towards this, let $R^{(i)}_X$ and $S^{(i)}_X$  denote the $i$-th rows of the matrices $\mathscr{A}$ and $\mathscr{R}$ respectively, where $1\leq i\leq k_0.$ We  observe that $\langle R^{(i)}_X,S^{(i)}_X\rangle_0\notin\gamma \mathtt{R},$  which implies that $[\overline{R^{(i)}}_X,\overline{S^{(i)}}_X]_0\ne0$ for $1\leq i\leq k_0.$ Further, one can easily see that $[\overline{R^{(i)}}_X,\overline{S^{(j)}}_X]_0=0$ for all $i\ne j$ and $1\leq i,j\leq k_0.$ This implies that the matrix $\overline{\mathscr{A}}~\overline{\mathscr{R}}^T$ is invertible,  from which the desired result follows.
\end{proof}
We now proceed to discuss two applications of LCPs of $\mathtt{R}\check{\mathtt{R}}$-linear codes.  First of all, in the following section, we will present a direct sum masking scheme based on LCPs of $\mathtt{R}\check{\mathtt{R}}$-linear codes and obtain its security threshold against fault injection attacks (FIA)  and side-channel attacks (SCA) \cite{Bringer}. 
\subsection{A direct sum masking scheme based on LCPs of $\mathtt{R}\check{\mathtt{R}}$-linear codes}\label{DSM}
Bringer~\etal~\cite{Bringer} introduced and studied an orthogonal direct sum masking scheme   using LCD codes over finite fields to protect smart cards or trusted platform modules against FIA and SCA. Later,  Ngo \etal~\cite{Ngo} introduced and studied a direct sum masking scheme using LCPs of codes over finite fields, which we outline as follows:  Suppose that   $U_1$ and $U_2$ are two linear codes of length $n$ over $\mathbb{F}_q$ with dimensions $k$ and $n-k$ and generator matrices $G_1$ and $G_2,$ respectively. Suppose that $U_1$ and $U_2$ form an LCP.   Let $u_1 \in \mathbb{F}_{q}^k$ be the sensitive data. In a direct sum masking scheme, the  sensitive data $u_1 \in \mathbb{F}_{q}^k$ is first encoded to the codeword $c=u_1G_1 \in U_1$ and then a random word $u_2 \in \mathbb{F}_{q}^{n-k}$ is chosen to create the word $d=u_2 G_2 \in U_2,$ which is known as the mask and acts as an intentionally added noise.  Since the encoded sensitive information $c$ and the mask $d$ belong to two complementary subspaces $U_1$ and $U_2$ of $\mathbb{F}_q^n$, it is possible to recover both $c$ and $d$ knowing the word $z=c+d \in \mathbb{F}_{q}^n.$  This technique is referred to as a direct sum masking scheme. 
In the same work \cite{Ngo},  Ngo \etal~ showed that the level of resistance (or the security threshold) of the direct sum masking scheme based on the LCP $(U_1,U_2)$ of codes against both SCA and FIA is given by $\min\{d_H(U_1),d_H(U_2^{\perp_0})\},$ where $d_H(U_1)$ and $d_H(U_2^{\perp_0})$ are  Hamming distances of the codes $U_1$ and  $U_2^{\perp_0},$ respectively, (note that $U_2^{\perp_0}$ is the  $0$-Galois (or Euclidean) dual code of the code $U_2$).  This work  can  be similarly extended to LCPs of codes over finite commutative chain rings. 

In this section, we will introduce and study  a direct sum masking scheme using  LCPs of $\mathtt{R}\check{\mathtt{R}}$-linear codes and obtain its security threshold against FIA and SCA.
 For this, we assume, throughout this section, that the $\mathtt{R}\check{\mathtt{R}}$-linear codes $\mathtt{C}$ and $\mathtt{D}$ of block-length $(\mathpzc{a},\mathpzc{b})$ form an LCP, \textit{i.e.,} we have $\mathtt{C}\oplus\mathtt{D}=\mathbf{M}.$ Additionally, suppose that the codes $\mathtt{C}$ and $\mathtt{D}$ have generator matrices $\mathcal{G}$ and $\mathcal{H}$ (as defined by \eqref{Eq6.1}) and parity-check matrices $\hat{\mathcal{G}}$ and $\hat{\mathcal{H}}$ (as defined by \eqref{Eq6.2}), respectively. As $\mathtt{C}\oplus \mathtt{D}=\mathbf{M},$ each element $m \in \mathbf{M}$ can be uniquely expressed as $m=a+b,$ where $a \in \mathtt{C}$ and $b \in \mathtt{D}.$ Further, since $\mathcal{G}$ and $\mathcal{H}$ are generator matrices of the codes $\mathtt{C}$ and $\mathtt{D},$ respectively, there exist unique words $v\in\mathtt{R}^{k_0}\oplus\check{\mathtt{R}}^{\ell_0} $ and $v'\in\mathtt{R}^{\mathpzc{a}-k_0}\oplus\check{\mathtt{R}}^{\mathpzc{b}-\ell_0}$ satisfying $a=v \mathcal{G}$ and $b=v'\mathcal{H}.$ From this, it follows that for each $m \in \mathbf{M},$ there exist unique words  $v\in\mathtt{R}^{k_0}\oplus\check{\mathtt{R}}^{\ell_0} $ and $v'\in\mathtt{R}^{\mathpzc{a}-k_0}\oplus\check{\mathtt{R}}^{\mathpzc{b}-\ell_0}$ such that $m=v\mathcal{G}+v' \mathcal{H}.$ 

Now, let $x=(x_1|x_2)\in\mathtt{R}^{k_0}\oplus\check{\mathtt{R}}^{\ell_0}$  be the sensitive data to be  protected from FIA and SCA, where $x_1\in\mathtt{R}^{k_0}$ and $x_2\in\check{\mathtt{R}}^{\ell_0}.$ To protect $x,$ we will first choose a random word $y=(y_1|y_2)\in\mathtt{R}^{\mathpzc{a}-k_0}\oplus\check{\mathtt{R}}^{\mathpzc{b}-\ell_0},$ where $y_1\in\mathtt{R}^{\mathpzc{a}-k_0}$ and $y_2\in\check{\mathtt{R}}^{\mathpzc{b}-\ell_0}.$ The sensitive word $x$ is first encoded to the word $c=(c_1|c_2)=x\mathcal{G}\in \mathtt{C},$ which is also called the encoded sensitive data. Then the mask $d=(d_1|d_2)=y\mathcal{H}\in \mathtt{D}$ is added to the encoded word $c$ to create the word $z=c+d \in \mathbf{M},$ which is called the direct sum representation. As $\mathtt{C}\oplus \mathtt{D}=\mathbf{M},$ it is theoretically  possible to recover both the sensitive data $x$ and the randomly chosen word $y$ knowing the word $z.$  We will next provide  methods to recover both the sensitive data $x$ and the random word $y$ from the word $z.$ Towards this, let us first define two maps $\Omega_1:\mathtt{R}^{k_0}\oplus\gamma^{e-s}\mathtt{R}^{\ell_0}\rightarrow\mathtt{R}^{k_0}\oplus\check{\mathtt{R}}^{\ell_0}$ and $\Omega_2:\mathtt{R}^{\mathpzc{a}-k_0}\oplus\gamma^{e-s}\mathtt{R}^{\mathpzc{b}-\ell_0}\rightarrow\mathtt{R}^{\mathpzc{a}-k_0}\oplus\check{\mathtt{R}}^{\mathpzc{b}-\ell_0}$ as
\vspace{-2mm}\begin{align}
\Omega_1(u_1,u_2,\ldots,u_{k_0},\gamma^{e-s}u_{k_0+1},\ldots,\gamma^{e-s}u_{k_0+\ell_0})&=(u_1,u_2,\ldots,u_{k_0}|u_{k_0+1},\ldots,u_{k_0+\ell_0})\label{Eq6.5}~\text{and}\\
\Omega_2(w_1,w_2,\ldots,w_{\mathpzc{a}-k_0},\gamma^{e-s}w_{\mathpzc{a}-k_0+1},\ldots,\gamma^{e-s}w_{\mathpzc{a}-k_0+\mathpzc{b}-\ell_0})&=(w_1,w_2,\ldots,w_{\mathpzc{a}-k_0}|w_{\mathpzc{a}-k_0+1},\ldots,w_{\mathpzc{a}-k_0+\mathpzc{b}-\ell_0})
\vspace{-2mm}\end{align}
for all $(u_1,u_2,\ldots,u_{k_0},\gamma^{e-s}u_{k_0+1},\ldots,\gamma^{e-s}u_{k_0+\ell_0})\in\mathtt{R}^{k_0}\oplus\gamma^{e-s}\mathtt{R}^{\ell_0}$ and $(w_1,w_2,\ldots,w_{\mathpzc{a}-k_0},\gamma^{e-s}w_{\mathpzc{a}-k_0+1},\ldots,\\\gamma^{e-s}w_{\mathpzc{a}-k_0+\mathpzc{b}-\ell_0})\in\mathtt{R}^{\mathpzc{a}-k_0}\oplus\gamma^{e-s}\mathtt{R}^{\mathpzc{b}-\ell_0}.$ This also gives rise to  the projection maps $\Omega_1^{(X)}:\mathtt{R}^{k_0}\oplus\gamma^{e-s}\mathtt{R}^{\ell_0}\rightarrow\mathtt{R}^{k_0}$ and $\Omega_1^{(Y)}:\mathtt{R}^{k_0}\oplus\gamma^{e-s}\mathtt{R}^{\ell_0}\rightarrow\check{\mathtt{R}}^{\ell_0}$  of the map $\Omega_1$ as 
\vspace{-2mm}\begin{align*}
\Omega_1^{(X)}(u_1,u_2,\ldots,u_{k_0},\gamma^{e-s}u_{k_0+1},\ldots,\gamma^{e-s}u_{k_0+\ell_0})&=(u_1,u_2,\ldots,u_{k_0})~\text{ and}\\
\Omega_1^{(Y)}(u_1,u_2,\ldots,u_{k_0},\gamma^{e-s}u_{k_0+1},\ldots,\gamma^{e-s}u_{k_0+\ell_0})&=(u_{k_0+1},\ldots,u_{k_0+\ell_0})
\vspace{-2mm}\end{align*}
for all  $(u_1,u_2,\ldots,u_{k_0},\gamma^{e-s}u_{k_0+1},\ldots,\gamma^{e-s}u_{k_0+\ell_0})\in\mathtt{R}^{k_0}\oplus\gamma^{e-s}\mathtt{R}^{\ell_0}.$

Now, by Remark \ref{Rk6.2}, we know that there exist invertible matrices $P_1\in M_{(k_0+\ell_0)\times(k_0+\ell_0)}(\mathtt{R})$ and $P_2\in M_{(\mathpzc{a}-k_0+\mathpzc{b}-\ell_0)\times(\mathpzc{a}-k_0+\mathpzc{b}-\ell_0)}(\mathtt{R})$ such that $(\mathcal{G}\diamond\hat{\mathcal{H}}^T) P_1 =\left[\begin{array}{c|c}
    I_{k_0} & 0 \\
    0 & \gamma^{e-s} I_{\ell_0}
\end{array} \right]$ and $(\mathcal{H}\diamond\hat{\mathcal{G}}^T) P_2 =\left[\begin{array}{c|c}
    I_{\mathpzc{a}-k_0} & 0 \\
    0 & \gamma^{e-s} I_{\mathpzc{b}-\ell_0}
\end{array} \right].$
Further, let us define  maps $\Psi_1:\mathbf{M}\rightarrow\mathtt{R}^{k_0}\oplus\gamma^{e-s}\mathtt{R}^{\ell_0}$ and $\Psi_2:\mathbf{M}\rightarrow\mathtt{R}^{\mathpzc{a}-k_0}\oplus\gamma^{e-s}\mathtt{R}^{\mathpzc{b}-\ell_0}$ as 
\vspace{-2mm}\begin{equation}\label{Psi1} 
\Psi_1(m)=(m\diamond\hat{\mathcal{H}}^T)P_1~\text{ ~and ~}\Psi_2(m)=(m\diamond\hat{\mathcal{G}}^T)P_2 \text{~ for all ~}m\in\mathbf{M}.
\vspace{-2mm} \end{equation}

Now, let $z=c+d,$ where $c=x\mathcal{G} \in \mathtt{C}$ and $d=y\mathcal{H} \in \mathtt{D}$ with $x=(x_1|x_2)\in\mathtt{R}^{k_0}\oplus\check{\mathtt{R}}^{\ell_0}$ and $y=(y_1|y_2)\in\mathtt{R}^{\mathpzc{a}-k_0}\oplus\check{\mathtt{R}}^{\mathpzc{b}-\ell_0}.$ By  \eqref{PCM}, we have $\mathcal{G}\diamond \hat{\mathcal{G}}^T=0$ and $\mathcal{H}\diamond \hat{\mathcal{H}}^T=0.$ Further, using equation \eqref{Psi1} and Remark \ref{Rk6.2},  we get \vspace{-3mm}\begin{align*}
\Psi_1(z)&=(z\diamond\hat{\mathcal{H}}^T)P_1=(x\mathcal{G}\diamond\hat{\mathcal{H}}^T+y\mathcal{H\diamond}\hat{\mathcal{H}}^T)P_1=x(\mathcal{G}\diamond\hat{\mathcal{H}}^T)P_1\\& =x\left[\begin{array}{c|c}
    I_{k_0} & 0 \\
    0 & \gamma^{e-s} I_{\ell_0}
\end{array} \right]=(x_1|\gamma^{e-s}x_2).
\vspace{-2mm}\end{align*}
Similarly, we observe that $\Psi_2(z)=(y_1|\gamma^{e-s}y_2).$  From this,  one can easily observe that
\vspace{-2mm}\begin{equation*}
\Omega_1\circ\Psi_1(z)=(x_1|x_2)=x~\text{~and~}~
\Omega_2\circ\Psi_2(z)=(y_1|y_2)=y.
\vspace{-2mm}\end{equation*}

\noindent\textbf{Computations}: We need to carry out certain computations within the direct sum representation. Since $x$ is the sensitive data, all operations must be performed on the state $z$ only, without involving $x.$ Here, we will see how the operations  performed on $x$ are translated to the corresponding operations on the direct sum representation $z.$  Following are the three main operations involved in the computations:
\begin{itemize}
\item[(a)] \textbf{Key addition}

A key addition to be performed on $x=(x_1|x_2)$ is a map $(x_1|x_2)\mapsto(x_1|x_2)+(f_1|f_2),$ where $(f_1|f_2)\in\mathtt{R}^{k_0}\oplus\check{\mathtt{R}}^{\ell_0}$ is the key. As  $(z_1|z_2)+(f_1|f_2)\mathcal{G}=\big((x_1|x_2)+(f_1|f_2)\big)\mathcal{G}+(y_1|y_2)\mathcal{H},$ the corresponding secure key addition on $z=(z_1|z_2)$ is the map  \vspace{-1mm}$$\vspace{-1mm}z=(z_1|z_2)\mapsto (z_1|z_2)+(f_1|f_2)\mathcal{G}.$$
\item[(b)] \textbf{Linear operation}

A linear operation to be performed on $x=(x_1|x_2)$ is a map $(x_1|x_2)\mapsto(x_1L_1|x_2L_2),$ where $L_1\in M_{k_0\times k_0}(\mathtt{R})$ and $L_2\in M_{\ell_0\times \ell_0}(\check{\mathtt{R}}).$  As $\Omega_1^{(X)}\big((z\diamond\hat{\mathcal{H}}^T)P_1\big)=x_1,$ $\Omega_2^{(Y)}\big((z\diamond\hat{\mathcal{H}}^T)P_1\big)=x_2$ and $\Omega_2\big((z\diamond\hat{\mathcal{G}}^T)P_2\big)=y,$ the corresponding masked linear operation on $z=(z_1|z_2)$ is given by \vspace{-1mm}$$\vspace{-2mm}z=(z_1|z_2)\mapsto \Big(\Omega_1^{(X)}\big((z\diamond\hat{\mathcal{H}}^T)P_1\big)L_1,\Omega_1^{(Y)}\big((z\diamond\hat{\mathcal{H}}^T)P_1\big)L_2\Big)\mathcal{G}+\Omega_2\big((z\diamond\hat{\mathcal{G}}^T)P_2\big)\mathcal{H}.$$

\item[(c)] \textbf{Non-linear operation}

A non-linear operation to be performed on $x=(x_1|x_2)$ is of the form $(x_1|x_2)\mapsto\big(S_1(x_1)|S_2(x_2)\big),$ where $S_1:\mathtt{R}^{k_0}\rightarrow\mathtt{R}^{k_0}$ and $S_2:\check{\mathtt{R}}^{\ell_0}\rightarrow\check{\mathtt{R}}^{\ell_0}$ are some non-linear operations. As $\Omega_1^{(X)}\big((z\diamond\hat{\mathcal{H}}^T)P_1\big)=x_1,$ $\Omega_2^{(Y)}\big((z\diamond\hat{\mathcal{H}}^T)P_1\big)=x_2$ and $\Omega_2\big((z\diamond\hat{\mathcal{G}}^T)P_2\big)=y,$ the corresponding  masked non-linear operation on $z=(z_1|z_2)$ is given by \vspace{-2mm}$$\vspace{-2mm}z=(z_1|z_2)\mapsto \bigg(S_1\Big(\Omega_1^{(X)}\big((z\diamond\hat{\mathcal{H}}^T)P_1\big)\Big),  S_2\Big(\Omega_1^{(Y)}\big((z\diamond\hat{\mathcal{H}}^T)P_1\big)\Big) \bigg) \mathcal{G} + \Omega_2\big((z\diamond\hat{\mathcal{G}}^T)P_2\big)\mathcal{H}.$$

\end{itemize}

We will now study the security of this direct sum masking scheme against FIA and SCA. To do this, we first extend the Hamming weight function $w_H$ defined on $\mathtt{R}^\mathpzc{a}$ and $\check{\mathtt{R}}^{\mathpzc{b}}$  to the Hamming weight function $\textbf{w}_H$ on $\mathbf{M}$ as \vspace{-3mm}$$\vspace{-3mm}\textbf{w}_H(\mathtt{c}_1,\mathtt{c}_2,\ldots,\mathtt{c}_\mathpzc{a}|\mathtt{d}_1,\mathtt{d}_2,\ldots,\mathtt{d}_\mathpzc{b})=\sum_{i=1}^\mathpzc{a}w_H(\mathtt{c}_i)+\sum\limits_{j=1}^\mathpzc{b}w_H(\mathtt{d}_i)$$ for all $(\mathtt{c}_1,\mathtt{c}_2,\ldots,\mathtt{c}_\mathpzc{a}|\mathtt{d}_1,\mathtt{d}_2,\ldots,\mathtt{d}_\mathpzc{b})\in\mathbf{M}.$ The Hamming distance between $m_1,m_2 \in \mathbf{M},$ denoted by $d_H(m_1,m_2),$ is defined as $d_H(m_1,m_2)=\textbf{w}_H(m_1-m_2).$ The Hamming distance of an $\mathtt{R}\check{\mathtt{R}}$-linear code $\mathtt{E}\subseteq\mathbf{M}$ is defined as  $d_H(\mathtt{E})=\min\{d_H(m_1,m_2): m_1, m_2\in \mathtt{E} \text{ and }m_1 \neq m_2\}$. It is easy to see that $d_H(\mathtt{E})= \min\{\textbf{w}_H(m): m (\neq 0)\in \mathtt{E}\}.$
\vspace{2mm}\\
\textbf{Security against FIA}
\vspace{1mm}\\
We will first study the security of  the direct sum masking scheme  constructed using the LCP of the codes $\mathtt{C}$ and $\mathtt{D}$ against FIA. Note that, throughout the computation, the state $z$ is masked by the same word $d=y\mathcal{H}$ for some $y\in\mathtt{R}^{\mathpzc{a}-k_0}\oplus\check{\mathtt{R}}^{\mathpzc{b}-\ell_0},$ which is randomly chosen in the beginning of the computation. Hence the value of the mask $d=y\mathcal{H}$ can be checked from time to time during the computations. 

During an FIA, the state $z=(z_1|z_2)$ may be modified to $z+\epsilon=(z_1+\epsilon_1|z_2+\epsilon_2)$ by injecting the fault  $\epsilon=(\epsilon_1|\epsilon_2)\in\mathbf{M},$ where $\epsilon_1\in\mathtt{R}^\mathpzc{a}$ and $\epsilon_2\in\check{\mathtt{R}}^\mathpzc{b}.$ Here, the main task is to establish whether such a modification due to FIA has occurred or not. For this, we note that $\mathbf{M}=\mathtt{C}\oplus\mathtt{D}$ and $\epsilon\in\mathbf{M}.$ So we can write $\epsilon=\epsilon_1'\mathcal{G}+\epsilon_2'\mathcal{H}$ for some $\epsilon_1'\in\mathtt{R}^{k_0}\oplus\check{\mathtt{R}}^{\ell_0}$ and $\epsilon_2'\in\mathtt{R}^{\mathpzc{a}-k_0}\oplus\check{\mathtt{R}}^{\mathpzc{b}-\ell_0}.$ This implies that $z+\epsilon=(x+\epsilon_1')\mathcal{G}+(y+\epsilon_2')\mathcal{H}.$ Since the variable $y$ is random and does not contain any sensitive information, the first step to determine whether or not any modification due to FIA has occurred on the state $z$ is to compute $\Omega_2\circ\Psi_2(z+\epsilon)=y+\epsilon_2'.$ If $\epsilon_2'\ne0,$ we can conclude that some modification on the original state $z$ has occurred due to FIA. On the other hand, if $\epsilon_2'=0,$ we have $z+\epsilon=(x+\epsilon_1')\mathcal{G}+y\mathcal{H}.$  Here, we see that the modification cannot be detected if $\epsilon\in\mathtt{C}.$ From this, it follows that if the injected fault $\epsilon\in\mathbf{M}$ satisfies $\textbf{w}_H(\epsilon)<d_H(\mathtt{C}),$ then the resulting FIA is detected.  From this, we deduce the following:
\begin{thm}\label{Thm6.4}
Let $\mathtt{C}$ and $\mathtt{D}$ form an LCP of $\mathtt{R}\check{\mathtt{R}}$-linear codes of block-length $(\mathpzc{a},\mathpzc{b})$ used in the direct sum masking scheme against FIA. The injected  fault $\epsilon \in \mathbf{M}$ is detected if $\mathbf{w}_H(\epsilon)<d_H(\mathtt{C}).$   
\end{thm}
In particular, if $\mathtt{C}$ and $\mathtt{D}$ form an LCP of separable $\mathtt{R}\check{\mathtt{R}}$-linear codes of block-length $(\mathpzc{a},\mathpzc{b}),$ then one can easily observe that $d_H(\mathtt{C})= min\{d_H(\mathtt{C}^{(X)}),d_H(\mathtt{C}^{(Y)})\}.$ From this and by Theorem \ref{Thm6.4}, we deduce the following:
\begin{cor}\label{Cor6.3}
Let $\mathtt{C}$ and $\mathtt{D}$ be an LCP of separable $\mathtt{R}\check{\mathtt{R}}$-linear codes of block-length $(\mathpzc{a},\mathpzc{b})$ used in the direct sum masking scheme against FIA. The injected fault $\epsilon$  is detected if $\mathbf{w}_H(\epsilon)<min\{d_H(\mathtt{C}^{(X)}),d_H(\mathtt{C}^{(Y)})\}.$   
\end{cor}
\vspace{2mm}
\noindent \textbf{Security against SCA}
\vspace{1mm}\\
We will now study the security of the direct sum masking scheme  constructed using the LCP of the codes $\mathtt{C}$ and $\mathtt{D}$ in protection against SCA. In this context, we focus exclusively on vertical attacks, where the attacker must gather a substantial number of traces to extract the sensitive data $x$. Note that both the state $z$ and the mask $d=y\mathcal{H}$ are not sensitive. To do this, let us first suppose that $z=(z_1|z_2)\in\mathbf{M}$ is of the form $z=(z_{1,1},z_{1,2},\ldots,z_{1,\mathpzc{a}}|z_{2,1},z_{2,2},\ldots,z_{2,\mathpzc{b}}),$ where $z_{1,i}\in\mathtt{R}$ for $1\leq i\leq \mathpzc{a}$ and $z_{2,j}\in\check{\mathtt{R}}$ for $1\leq j\leq \mathpzc{b}$. 

Now, for a positive integer $\delta,$ we say that there is a leakage of order $\delta$ in an SCA if there exists $T_1\subseteq\{1,2,\ldots,\mathpzc{a}\}$ and $T_2\subseteq\{1,2,\ldots,\mathpzc{b}\}$ such that $|T_1|+|T_2|=\delta$ and the coordinates $z_{1,i_1},z_{1,i_2},\ldots,z_{1,i_{|T_1|}},z_{2,j_1},z_{2,j_2},\ldots,z_{2,j_{|T_2|}}$ of the state $z$ are leaked. We further assume that $T_1$ and $T_2$ are ordered sets, given by $T_1=\{i_1,i_2,\ldots,i_{|T_1|}\}$  and $T_2=\{j_1,j_2,\ldots,j_{|T_2|}\},$ where  $i_1 < i_2 < \cdots < i_{|T_1|}$ and $j_1 < j_2 < \cdots < j_{|T_2|}.$ Furthermore,  for a vector $(\alpha_1|\alpha_2)=(\alpha_{1,1},\alpha_{1,2},\ldots,\alpha_{1,\mathpzc{a}}|\alpha_{2,1},\alpha_{2,2},\ldots,\alpha_{2,\mathpzc{b}})\in\mathbf{M},$ the vector $(\alpha_1^{(T_1)}|\alpha_2^{(T_2)})$ is defined as $(\alpha_1^{(T_1)}|\alpha_2^{(T_2)})=(\alpha_{1,i_1},\alpha_{1,i_2},\ldots,\alpha_{1,i_{|T_1|}}|\alpha_{2,j_1},\alpha_{2,j_2},\ldots,\alpha_{2,j_{|T_2|}}).$ For a matrix $[E~|~F]\in\mathfrak{M},$ let $[E^{(T_1)}~|~F^{(T_2)}]$ denote the matrix whose $i$-th row is $(\beta_1^{(T_1)}|\beta_2^{(T_2)})$ if $(\beta_1|\beta_2)$ is the $i$-th row of the matrix $[E~|~F].$ 
Our aim is to obtain an upper bound on $\delta$ such that any leakage of the order less than $\delta $ would give no advantage to the adversary in extracting or obtaining even a partial information about the sensitive data $x.$ To do this, we first prove the following lemma.
\begin{lemma}\label{Lem6.4}
Let $\mathtt{D}\subseteq\mathbf{M}$ be an $\mathtt{R}\check{\mathtt{R}}$-linear code with a generator matrix $\mathcal{H}$ as defined  by \eqref{Eq6.1}. Let $\mathtt{D}^{(X)}$ be a linear code of length $\mathpzc{a}$ over $\mathtt{R}$ with a generator matrix $\mathscr{E},$ and let $\mathtt{D}^{(Y)}$ be a linear code of length $\mathpzc{b}$ over $\check{\mathtt{R}}$ with a generator matrix $\mathscr{H}.$ For positive integers $t_1$ and $t_2$ satisfying $t_1\leq \mathpzc{a}-k_0$ and $t_2\leq \mathpzc{b}-\ell_0,$ we have the following:
\begin{enumerate}
\item[(a)] Any $(t_1-1)$ columns of the matrix $\mathscr{E}$ are independent over $\mathtt{R}$ and have period $\gamma^e$ if and only if $\nolinebreak{d_H(\mathtt{D}^{(X)\perp_0})\geq t_1}.$
\item[(b)] Any $(t_2-1)$ columns of the matrix $\mathscr{H}$ are independent over $\check{\mathtt{R}}$ and have period $\gamma^s$ if and only if $\nolinebreak{d_H(\mathtt{D}^{(Y)\perp_0})\geq t_2}.$
\end{enumerate}
\end{lemma}
\begin{proof}
Working as in Theorem 4.5.6 of \cite{Ling}, we get the desired result.
\end{proof}

Now, suppose that there is a leakage of order $\delta=|T_1|+|T_2|,$ where $T_1\subseteq\{1,2,\ldots,\mathpzc{a}\}$ and $T_2\subseteq\{1,2,\ldots,\mathpzc{b}\}$ are ordered sets satisfying $|T_1|<d_H(\mathtt{D}^{(X)\perp_0})$ and $|T_2|<d_H(\mathtt{D}^{(Y)\perp_0}).$ Here, we first observe that 
\vspace{-1mm}\begin{align}\label{eq6.9}
(z_1^{(T_1)}|z_2^{(T_2)})=x\left[\begin{array}{c|c}
    \mathscr{A}^{(T_1)}&\mathscr{B}^{(T_2)} \\ 
    \gamma^{e-s}\mathscr{C}^{(T_1)}& \mathscr{D}^{(T_2)} 
\end{array}\right]+y\left[\begin{array}{c|c}
\mathscr{E}^{(T_1)}&\mathscr{F}^{(T_2)} \\ 
    \gamma^{e-s}\mathscr{G}^{(T_1)}& \mathscr{H}^{(T_2)}
\end{array}\right].
\vspace{-1mm}\end{align}
Now since $|T_1|<d_H(\mathtt{D}^{(X)\perp_0}),$ we see, by Lemma \ref{Lem6.4}(a), that the columns of the matrix $\mathscr{E}^{(T_1)}$ are independent over $\mathtt{R}$ and have period $\gamma^e.$ Similarly, since $|T_2|<d_H(\mathtt{D}^{(Y)\perp_0}),$ we see, by Lemma \ref{Lem6.4} (b), that the columns of the matrix $\mathscr{H}^{(T_2)}$ are independent over $\check{\mathtt{R}}$ and have period $\gamma^s.$ From this, one can easily observe  that the matrix $\left[\begin{array}{c|c}
\mathscr{E}^{(T_1)}&\mathscr{F}^{(T_2)} \\ 
    \gamma^{e-s}\mathscr{G}^{(T_1)}& \mathscr{H}^{(T_2)}
\end{array}\right]$ generates $\mathtt{R}^{|T_1|}\oplus\check{\mathtt{R}}^{|T_2|}.$ From this and using the fact that the word $y\in\mathtt{R}^{\mathpzc{a}-k_0}\oplus\check{\mathtt{R}}^{\mathpzc{b}-\ell_0}$ is randomly chosen, it follows that even if there is a leakage of $\delta$ coordinates $(z_1^{(T_1)}|z_2^{(T_2)})$ of $(z_1|z_2),$   the word $(z_1^{(T_1)}|z_2^{(T_2)})-y\left[\begin{array}{c|c}
\mathscr{E}^{(T_1)}&\mathscr{F}^{(T_2)} \\ 
    \gamma^{e-s}\mathscr{G}^{(T_1)}& \mathscr{H}^{(T_2)}
\end{array}\right]$ is uniformly distributed in $\mathtt{R}^{|T_1|}\oplus\check{\mathtt{R}}^{|T_2|}.$ Now by  equation \eqref{eq6.9}, we see that the vector $(x_1^{(T_1)}|x_2^{(T_2)})=x\left[\begin{array}{c|c}
    \mathscr{A}^{(T_1)}&\mathscr{B}^{(T_2)} \\ 
    \gamma^{e-s}\mathscr{C}^{(T_1)}& \mathscr{D}^{(T_2)} 
\end{array}\right]$ is uniformly distributed in $\mathtt{R}^{|T_1|}\oplus\check{\mathtt{R}}^{|T_2|}$ even if there is a leakage of $\delta$ coordinates $(z_1^{(T_1)}|z_2^{(T_2)})$ of $(z_1|z_2).$ This shows that this direct sum masking scheme protects against SCA with any leakage of order $\delta=|T_1|+|T_2|,$ where $T_1\subseteq\{1,2,\ldots,\mathpzc{a}\}$ and $T_2\subseteq\{1,2,\ldots,\mathpzc{b}\}$ satisfy $|T_1|<d_H(\mathtt{D}^{(X)\perp_0})$ and $|T_2|<d_H(\mathtt{D}^{(Y)\perp_0}).$ From this, we deduce the following:
\begin{thm}\label{Thm6.5}
Let $\mathtt{C}$ and $\mathtt{D}$ form an LCP of $\mathtt{R}\check{\mathtt{R}}$-linear codes of block-length $(\mathpzc{a},\mathpzc{b})$ used in the direct sum masking scheme against SCA. The direct sum masking scheme protects against SCA with any leakage of order $\delta=|T_1|+|T_2|,$ where $T_1\subseteq\{1,2,\ldots,\mathpzc{a}\}$ and $T_2\subseteq\{1,2,\ldots,\mathpzc{b}\}$ satisfy $|T_1|<d_H(\mathtt{D}^{(X)\perp_0})$ and $|T_2|<d_H(\mathtt{D}^{(Y)\perp_0}).$   
\end{thm}

In the following theorem, we obtain the security threshold of this direct sum masking scheme against FIA and SCA.  
\begin{thm}\label{Thm6.04}
An LCP of $\mathtt{R}\check{\mathtt{R}}$-linear codes $\mathtt{C}$ and $\mathtt{D}$ enables a direct sum masking scheme, which offers simultaneous protection against both FIA and SCA with a security threshold, given by \vspace{-1mm}$$\vspace{-1mm}min\{d_H(\mathtt{C}),d_H(\mathtt{D}^{(X)\perp_0}),d_H(\mathtt{D}^{(Y)\perp_0})\}.$$
\end{thm}
\begin{proof}It follows immediately from Theorems \ref{Thm6.4} and \ref{Thm6.5}.
\end{proof}
In particular, if $\mathtt{C}$ and $\mathtt{D}$ form an LCP of separable $\mathtt{R}\check{\mathtt{R}}$-linear codes of block-length $(\mathpzc{a},\mathpzc{b}),$ then we have the following:
\begin{cor}\label{Cor6.4}
An LCP of separable $\mathtt{R}\check{\mathtt{R}}$-linear codes $\mathtt{C}$ and $\mathtt{D}$ enables a direct sum masking scheme, which offers simultaneous protection against both FIA and SCA with a security threshold, given by \vspace{-1mm}$$\vspace{-1mm}min\{d_H(\mathtt{C}^{(X)}),d_H(\mathtt{C}^{(Y)}),d_H(\mathtt{D}^{(X)\perp_0}),d_H(\mathtt{D}^{(Y)\perp_0})\}.$$
\end{cor}
\begin{proof}It follows immediately from Corollary \ref{Cor6.3} and Theorem \ref{Thm6.5}.
\end{proof}

Working as in Ngo \etal~\cite{Ngo}, one can easily see that the security threshold against FIA and SCA of the direct sum making scheme constructed using an LCP $(\mathcal{C}, \mathcal{D})$ of linear codes over finite commutative chain rings  is $\min\{d_H(\mathcal{C}), d_H(\mathcal{D}^{\perp_0})\}.$  Now,  by applying Corollary \ref{Cor6.4},  we observe that using an LCP of separable $\mathtt{R}\check{\mathtt{R}}$-linear codes $\mathtt{C}$ and $\mathtt{D}$  to construct a direct sum masking scheme for protecting the sensitive data $x=(x_1|x_2)$ against FIA and SCA is equivalent to using two different LCPs $\left({\mathtt{C}^{(X)}}, {\mathtt{D}^{(X)}}\right)$ and $\left({\mathtt{C}^{(Y)}}, {\mathtt{D}^{(Y)}}\right)$ of linear codes over $\mathtt{R}$ and $\check{\mathtt{R}}$ to protect the individual parts $x_1$ and $x_2$ of the sensitive data $x$ against FIA and SCA, respectively. In fact, both the schemes have the same security threshold against FIA and SCA. 

For an LCP of $\mathtt{R}\check{\mathtt{R}}$-linear codes  $\mathtt{C}$ and $\mathtt{D}$ with generator matrices $\mathcal{G}$ and $\mathcal{H}$ (as defined by \eqref{Eq6.1}) respectively, there exists an LCP of separable $\mathtt{R}\check{\mathtt{R}}$-linear codes  $\mathtt{C}_{sep}$ and $\mathtt{D}_{sep}$ with generator matrices of the forms
\vspace{-1mm}\begin{align}\label{Eq6.1sep}
\mathcal{G}_{sep}=\left[\begin{array}{c|c}
    \mathscr{A}&0 \\ 
    0 & \mathscr{D} 
\end{array}\right]\text{ ~ and ~ }\mathcal{H}_{sep}=\left[\begin{array}{c|c}
    \mathscr{E}&0\\ 
     0 & \mathscr{H} 
\end{array}\right],
\end{align} respectively. On the other hand, since the ring $\check{\mathtt{R}}$ can be embedded into $\mathtt{R}, $ one can view the matrices $\mathcal{G}$ and $\mathcal{H}$ as matrices over $\mathtt{R}.$ Let $\mathtt{C}_{emb}$ and $\mathtt{D}_{emb}$ be the linear codes of length $\mathpzc{n}=\mathpzc{a}+\mathpzc{b}$ over $\mathtt{R}$ with generator matrices $\mathcal{G}$ and $\mathcal{H},$ respectively, where the matrices $\mathcal{G}$ and $\mathcal{H}$ are viewed over $\mathtt{R}.$ 
We now provide an example to illustrate that the direct sum masking scheme constructed using the LCP $(\mathtt{C}, \mathtt{D})$ of non-separable $\mathtt{R}\check{\mathtt{R}}$-linear codes  outperforms  its counterparts constructed using the corresponding LCP $(\mathtt{C}_{sep}, \mathtt{D}_{sep})$ of separable $\mathtt{R}\check{\mathtt{R}}$-linear codes  and the corresponding LCP $(\mathtt{C}_{emb}, \mathtt{D}_{emb})$ of linear codes over $\mathtt{R}$ in terms of the security threshold against FIA and SCA.
\begin{example}\label{Eg.6.1}
Let $\mathtt{C}'$ and $\mathtt{D}'$ be $\mathbb{Z}_4\mathbb{Z}_2$-linear codes of block-length $(6,5)$ with generator matrices
$$\vspace{-2mm}\mathcal{G}'=\left[\begin{array}{c|c}
    \mathscr{A}' & \mathscr{B}' \\ 
    2\mathscr{C}' & \mathscr{D}'
\end{array}\right]\text{ ~ and ~ } \mathcal{H}'=\left[\begin{array}{c|c}
    \mathscr{E}' & \mathscr{F}'\\ 
    2\mathscr{G}' & \mathscr{H}'
\end{array}\right],$$
respectively, where $~\mathscr{A}'=[1~1~1~1~1~0], ~~ \mathscr{B}'=[0 ~ 1 ~ 1 ~ 0 ~ 1],~~ \mathscr{C}'=[0 ~ 0 ~ 1 ~ 1 ~ 1 ~ 1],~~ \mathscr{D}'=[1 ~ 1 ~ 1 ~ 0 ~ 0],$ \vspace{2mm}\\$ \mathscr{E}'=\left[\begin{array}{cccccc}
1 & 0 & 0 & 0 & 0 & 3\\
0 & 1 & 0 & 0 & 0 & 3\\
0 & 0 & 1 & 0 & 0 & 3\\
0 & 0 & 0 & 1 & 0 & 3\\
0 & 0 & 0 & 0 & 1 & 3
\end{array}\right],~~ \mathscr{F}'=\left[\begin{array}{ccccc}
 1 & 0 & 0 & 0 & 0\\
1 & 0 & 0 & 0 & 0\\
1 & 0 & 0 & 0 & 0\\
1 & 0 & 0 & 0 & 0\\
1 & 0 & 0 & 0 & 0
\end{array}\right],~~ \mathscr{G}'=\left[\begin{array}{cccccc}
 0 & 0 & 0 & 0 & 0 & 1\\
 0 & 0 & 0 & 0 & 0 & 1\\
 0 & 0 & 0 & 0 & 0 & 1\\
 0 & 0 & 0 & 0 & 0 & 1
\end{array}\right]$ ~ and \\$ \mathscr{H}'=\left[\begin{array}{ccccc}
 1 & 1 & 0 & 0 & 0\\
 0 & 1 & 1 & 0 & 0\\
 0 & 0 & 1 & 1 & 0\\
 0 & 0 & 0 & 1 & 1
\end{array}\right].$ Here, $\mathtt{C}'^{(X)}$ and $\mathtt{D}'^{(X)}$ are linear codes of length $6$ over $\mathbb{Z}_4$ with generator matrices $\mathscr{A}'$ and $\mathscr{E}',$ 
respectively. Similarly, $\mathtt{C}'^{(Y)}$ and $\mathtt{D}'^{(Y)}$ are linear codes of length $5$ over $\mathbb{Z}_2$ with generator matrices $\mathscr{D}'$ and $\mathscr{H}',$ respectively. By Theorem \ref{Thm6.3}, the codes $\mathtt{C}'$ and $\mathtt{D}'$ form an LCP of $\mathbb{Z}_4\mathbb{Z}_2$-linear codes. Note that both the codes $\mathtt{C}'$ and $\mathtt{D}'$ are not separable.

By direct computations, we observe that the Euclidean dual codes $\mathtt{D}'^{(X)\perp_0}$ and $\mathtt{D}'^{(Y)\perp_0}$ have generator matrices $[1~1~1~1~1~1]$ and $[1~1~1~1~1],$ respectively. Further, we see, using Magma, that $d_H(\mathtt{C}')=5,$ $ d_H(\mathtt{D}'^{(X)\perp_0})=6$ and $ d_H(\mathtt{D}'^{(Y)\perp_0})=5.$ Hence, by Theorem \ref{Thm6.04}, the direct sum masking scheme constructed   using the LCP $(\mathtt{C}',\mathtt{D}')$ has security threshold $min\{d_H(\mathtt{C}'),d_H(\mathtt{D}'^{(X)\perp_0}),d_H(\mathtt{D}'^{(Y)\perp_0})\}=5.$ 

Further, corresponding to the LCP $(\mathtt{C}', \mathtt{D}'),$ let $(\mathtt{C}'_{sep}, \mathtt{D}'_{sep})$ be the LCP of separable $\mathbb{Z}_4\mathbb{Z}_2$-linear codes  $\mathtt{C}'_{sep}$ and $ \mathtt{D}'_{sep}$ with generator matrices
\vspace{-3mm}$$\vspace{-2mm}\mathcal{G}'_{sep}=\left[\begin{array}{c|c}
   \mathscr{A}'&0\\
   0&\mathscr{D}'
\end{array}\right]\text{ ~~and ~~}\mathcal{H}'_{sep}=\left[\begin{array}{c|c}
    \mathscr{E}' & 0 \\
    0&\mathscr{H}'
\end{array}\right],$$ respectively.
Here, we observe that $d_H(\mathtt{C}'^{(X)})=5$ and $d_H(\mathtt{C}'^{(Y)})=3.$ This,  by Corollary \ref{Cor6.4},  implies that the direct sum masking scheme constructed using the LCP $(\mathtt{C}'_{sep}, \mathtt{D}'_{sep})$ of separable codes has security threshold $min\{d_H(\mathtt{C}'^{(X)}),d_H(\mathtt{C}'^{(Y)}),d_H(\mathtt{D}'^{(X)\perp_0}),d_H(\mathtt{D}'^{(Y)\perp_0})\}=3.$

On the other hand, the ring $\mathbb{Z}_2$ can be embedded into $\mathbb{Z}_4,$ so one can also view the matrices $\mathcal{G}'$ and $\mathcal{H}'$ as matrices over $\mathbb{Z}_4.$  Let $\mathtt{C}_{emb}'$ and $\mathtt{D}_{emb}'$ be the linear codes of length $11$ over $\mathbb{Z}_4$ with generator matrices $\mathcal{G}'$ and $\mathcal{H}',$ respectively, where the matrices $\mathcal{G}'$ and $\mathcal{H}'$ are viewed over $\mathbb{Z}_4.$  Using  Magma, we see that the codes $\mathtt{C}_{emb}'$ and $\mathtt{D}_{emb}'$ form an LCP of linear codes over $\mathbb{Z}_4$ with $d_H(\mathtt{C}_{emb}')=3$ and  $d_H(\mathtt{D}_{emb}'^{\perp_0})=6.$ Working as in Ngo \etal~\cite{Ngo}, the direct sum making scheme constructed using the LCP $(\mathtt{C}_{emb}', \mathtt{D}_{emb}')$  has security threshold $min\{d_H(\mathtt{C}_{emb}'), d_H(\mathtt{D}_{emb}'^{\perp_0})\}=3.$ 

Thus the direct sum masking scheme constructed using the LCP $(\mathtt{C}',\mathtt{D}')$ of non-separable $\mathbb{Z}_4\mathbb{Z}_2$-linear codes of block-length $(6,5)$ outperforms the direct sum masking schemes constructed using the LCP $(\mathtt{C}'_{sep}, \mathtt{D}'_{sep})$ of separable $\mathbb{Z}_4\mathbb{Z}_2$-linear codes and  the LCP $(\mathtt{C}_{emb}', \mathtt{D}_{emb}')$ of linear codes over $\mathbb{Z}_4$ in terms of security threshold against FIA and SCA.
\end{example}

In the next section, we will discuss another application of  LCPs of $\mathtt{R}\check{\mathtt{R}}$-linear codes  in coding for the noiseless two-user adder channel.

\subsection{An application of LCPs of $\mathtt{R}\check{\mathtt{R}}$-linear codes in coding for the noiseless two-user adder channel}\label{tuser}
In a noiseless two-user adder channel, the two users of this channel send their codewords simultaneously through the channel, and the  receiver receives only the sum of the transmitted codewords. The problem for the receiver is to recover the individual codewords transmitted from the received sum. Massey \cite{Massey} showed that binary LCD codes provide an optimum linear coding solution for the noiseless two-user binary adder channel. Recently, Liu and Hu \cite{Liu2024} also provided a solution to this problem using the LCPs of $\mathbb{Z}_4\mathbb{Z}_2$-linear codes. In this section, we will also provide a coding solution to this problem using LCPs of $\mathtt{R}\check{\mathtt{R}}$-linear codes. For this, we assume, throughout this section, that the codes  $\mathtt{C}$ and $\mathtt{D}$ form an LCP of $\mathtt{R}\check{\mathtt{R}}$-linear codes of block-length $(\mathpzc{a},\mathpzc{b}),$ \textit{i.e.,} we have $\mathtt{C}\oplus\mathtt{D}=\mathbf{M}.$ Additionally, suppose that the codes $\mathtt{C}$ and $\mathtt{D}$ have generator matrices $\mathcal{G}$ and $\mathcal{H}$ (as defined by \eqref{Eq6.1}) and parity-check matrices $\hat{\mathcal{G}}$ and $\hat{\mathcal{H}}$ (as defined by \eqref{Eq6.2}), respectively.

To obtain a projection map from $\mathbf{M}$ onto $\mathtt{C}$ with respect to $\mathtt{D},$  let us first define a map $\Omega_0:\mathtt{R}^{k_0}\oplus\check{\mathtt{R}}^{\ell_0}\rightarrow\mathtt{C}$  as 
\vspace{-2mm}\begin{equation}\label{EQ0}
\Omega_0(u)=u\mathcal{G}~\text{~for all~}~u\in\mathtt{R}^{k_0}\oplus\check{\mathtt{R}}^{\ell_0}.
\vspace{-2mm}\end{equation}
 Then the following lemma provides a projection map from $\mathbf{M}$ onto $\mathtt{C}$ with respect to $\mathtt{D}.$

\begin{lemma}\label{Lem6.5}
The map $\Omega_0\circ\Omega_1\circ\Psi_1$ is a projection map from $\mathbf{M}$ onto $\mathtt{C}$ with respect to $\mathtt{D},$ where the maps $\Omega_1:\mathtt{R}^{k_0}\oplus\gamma^{e-s}\mathtt{R}^{\ell_0}\rightarrow\mathtt{R}^{k_0}\oplus\check{\mathtt{R}}^{\ell_0}$ and $\Psi_1:\mathbf{M}\rightarrow\mathtt{R}^{k_0}\oplus\gamma^{e-s}\mathtt{R}^{\ell_0}$ are as defined by \eqref{Eq6.5} and \eqref{Psi1}, respectively.
\end{lemma}
\begin{proof} To prove the result, let $z\in\mathbf{M}$ be fixed. As $\mathbf{M}=\mathtt{C}\oplus \mathtt{D},$ there exist  unique $c\in\mathtt{C}$ and $d\in\mathtt{D}$ such that $z=c+d.$ Further, as $\mathcal{G}$ and $\mathcal{H}$ are generator matrices of $\mathtt{C}$ and $\mathtt{D},$ respectively, there exist $x\in\mathtt{R}^{k_0}\oplus\check{\mathtt{R}}^{\ell_0}$ and $y\in\mathtt{R}^{\mathpzc{a}-k_0}\oplus\check{\mathtt{R}}^{\mathpzc{b}-\ell_0}$ such that $c=x\mathcal{G}$ and $d=y\mathcal{H}.$ From this, we get $z=x\mathcal{G}+y\mathcal{H}.$ Now using \eqref{PCM}, \eqref{Eq6.5}, \eqref{Psi1} and \eqref{EQ0}, we observe that $\Omega_0\circ\Omega_1\circ\Psi_1(z)=c,$ which proves the lemma.
\end{proof}


 Now, suppose that the two users of the two-user adder channel transmit two codewords $c\in\mathtt{C}$ and $d\in\mathtt{D}$ simultaneously through the noiseless channel.  At the receiver's end, the received word $z$ is the sum $c+d$ of these two codewords.  Now, to recover the transmitted codeword $c$ from $z,$ the receiver applies the projection map  $\Omega_0\circ\Omega_1\circ\Psi_1$ on $z,$ which, by applying Lemma \ref{Lem6.5}, gives the transmitted codeword $c.$ Further, the other transmitted codeword $d$ can be obtained by computing $z-(\Omega_0\circ\Omega_1\circ\Psi_1)(z).$ Hence, by using LCPs of $\mathtt{R}\check{\mathtt{R}}$-linear codes, one can recover the codewords $c$ and $d$ from the received word $z.$ The following example illustrates this.
\begin{example}\label{Eg.6.2}
Let $\mathbb{F}_8=\{0,1,\xi,\xi^2,\xi^3=\xi+1,\xi^4,\xi^5,\xi^6\}$ be the finite field of order $8$ with a primitive element $\xi.$  Let $\mathtt{R}$ be the quasi-Galois ring $\mathbb{F}_8[u]/\langle u^2\rangle,$ and let  $\check{\mathtt{R}}=\mathbb{F}_8.$ Let $\mathtt{C}_2$ and $\mathtt{D}_2$ be $\mathtt{R}\check{\mathtt{R}}$-linear codes of block-length $(3,3)$ with generator matrices \vspace{-2mm}$$\vspace{-2mm}\mathcal{G}_2=\left[\begin{array}{ccc|ccc}
    1 & 1 & 1 & 0 & 0 & 1\\ 
    0 & u & 0 & 1 & 0 & 0\\
    0 & 0 & u & 0 & 1 & 0
\end{array}\right]\text{ ~ and ~ } \mathcal{H}_2=\left[\begin{array}{ccc|ccc}
    0 & 1 & 0 & 0 & 0 & 0\\ 
    0 & 0 & 1 & 0 & 0 & 0\\
    0 & 0 & 0 & 0 & 0 & 1
\end{array}\right],$$ respectively. By applying Theorem \ref{Thm6.3}, we see that the codes $\mathtt{C}_2$ and $\mathtt{D}_2$ form an LCP.  Further, one can easily see that the code $\mathtt{C}_2$ is of the type $\{1,0;2\}$ and that the parity-check matrix of the code $\mathtt{D}_2$ is given by \vspace{-2mm}$$\vspace{-2mm}\hat{\mathcal{H}}_2=\left[\begin{array}{ccc|ccc}
    1 & 0 & 0 & 0 & 0 & 0\\ 
    0 & 0 & 0 & 1 & 0 & 0\\
    0 & 0 & 0 & 0 & 1 & 0
\end{array}\right].$$
 This implies that  \vspace{-2mm}$$\vspace{-2mm}\mathcal{G}_2\diamond\hat{\mathcal{H}}_2^T=\left[\begin{array}{ccc}
    1 & 0 & 0 \\ 
    0 & u & 0 \\
    0 & 0 & u 
\end{array}\right],$$ which, by Remark \ref{Rk6.2}, further implies that  $P_1=I_{3}$ is the matrix satisfying $(\mathcal{G}_2\diamond\hat{\mathcal{H}_2}^T) P_1 =\left[\begin{array}{c|c}
    1 & 0 \\
    0 & u I_{2}
\end{array} \right].$  

Now, suppose that the codewords $c=(\xi,\xi+u,\xi|1,0,\xi)\in\mathtt{C}_2$ and $d=(0,1,1|0,0,0)\in\mathtt{D}_2$ are transmitted simultaneously through the noiseless two-user adder channel. The receiver receives the word $z=c+d=(\xi,\xi+u+1,\xi+1|1,0,\xi).$
By applying the projection map $\Omega_0\circ\Omega_1\circ\Psi_1$ on $z$ and using equations \eqref{Eq6.5}, \eqref{Psi1}  and \eqref{EQ0}, we get
\vspace{-2mm}\begin{align*}
\Omega_0\circ\Omega_1\circ\Psi_1(z)&=\Omega_0\circ\Omega_1(z \diamond \hat{\mathcal{H}}_2^T)=\Omega_0\circ\Omega_1\Big([\xi~\xi+u+1~\xi+1~|~1~0~\xi]\diamond\left[\begin{array}{ccc|ccc}
    1 & 0 & 0 & 0 & 0 & 0\\ 
    0 & 0 & 0 & 1 & 0 & 0\\
    0 & 0 & 0 & 0 & 1 & 0
\end{array}\right]^T\Big)
\\&=\Omega_0\circ\Omega_1(\xi,u,0)
=\Omega_0(\xi,1,0)
=[\xi~1~0]\mathcal{G}_2
=(\xi,\xi+u,\xi|1,0,\xi)=c.
\vspace{-2mm}\end{align*}
Further, we see that $z-c=(\xi,\xi+u+1,\xi+1|1,0,\xi)-(\xi,\xi+u,\xi|1,0,\xi)=(0,1,1|0,0,0)=d.$ Thus the transmitted codewords $c$ and $d$ are recovered from the received word $z$ using the LCP $(\mathtt{C}_2, \mathtt{D}_2)$ of $\mathtt{R}\check{\mathtt{R}}$-linear codes.
\end{example}
\section{Conclusion and future work}\label{Conclusion}
In this paper, Galois $\mathtt{R}\check{\mathtt{R}}$-LCD codes of an arbitrary block length are studied and characterized.  Further, all Euclidean and Hermitian $\mathtt{R}\check{\mathtt{R}}$-LCD codes of an arbitrary block-length are enumerated.  With the help of these enumeration formulae  and using Magma,  all Euclidean $\mathbb{Z}_4 \mathbb{Z}_{2}$-LCD   and $\mathbb{Z}_9 \mathbb{Z}_{3}$-LCD codes  of block-lengths $(1,1),$ $(1,2),$ $(2,1),$ $(2,2),$ $(3,1)$ and $(3,2)$   and all Hermitian $\frac{\mathbb{F}_{4}[u]}{\langle u^2\rangle} \;\mathbb{F}_{4}$-LCD codes of block-lengths $(1,1),$ $(1,2),$ $(2,1)$ and $(2,2)$ are classified up to monomial equivalence.  Apart from this, LCPs of $\mathtt{R}\check{\mathtt{R}}$-linear codes of an arbitrary block-length are studied. A necessary and sufficient condition under which a pair of $\mathtt{R}\check{\mathtt{R}}$-linear codes form an LCP is also derived. A direct sum masking scheme based on an LCP of $\mathtt{R}\check{\mathtt{R}}$-linear codes is also studied, and its security threshold against FIA and SCA is also obtained. Another application of LCPs of $\mathtt{R}\check{\mathtt{R}}$-linear codes to the two-user adder channel is also studied.  It would be interesting to study Gray images of LCPs of $\mathtt{R}\check{\mathtt{R}}$-linear codes and derive necessary and sufficient  conditions under which these Gray images form an LCP of codes. Another interesting line of  research  would be to construct LCPs of $\mathtt{R}\check{\mathtt{R}}$-linear codes with optimal security thresholds. 
\section{Acknowledgements}
The first author acknowledges the research support provided by the National Board for Higher Mathematics (NBHM), India,  under Grant No. 0203/13(46)/2021-R\&D-II/13176. The second author acknowledges the research support provided by the Department of Science and Technology (DST), India, under Grant No. DST/INT/RUS/RSF/P-41/2021 with TPN 65025.

\vspace{4mm}\noindent\textbf{\Large Appendices}\vspace{4mm}\\
\noindent \textbf{\large A.} \textbf{\large Classification of Euclidean $\mathbb{Z}_4 \mathbb{Z}_{2}$-LCD    codes   of block-length $(\mathpzc{a},\mathpzc{b}),$ where $\mathpzc{a} \in \{1,2,3\}$ and $\mathpzc{b} \in \{1,2\}$}\\

To classify Euclidean $\mathbb{Z}_4 \mathbb{Z}_{2}$-LCD    codes   of block-length $(\mathpzc{a},\mathpzc{b})$ with $\mathpzc{a} \in \{1,2,3\}$ and $\mathpzc{b} \in \{1,2\},$ we first recall that the  Gray map $\phi:\mathbb{Z}_4\rightarrow\mathbb{Z}_2^2$ is  defined as $\phi(0)=(0,0),$ $\phi(1)=(0,1),$ $\phi(2)=(1,1)$ and $\phi(3)=(1,0)$  and that the Lee weight $w_L(a)$ of an element $a\in \mathbb{Z}_{4}$ is defined as the Hamming weight of its Gray image $\phi(a),$ \textit{i.e.,} we have   $w_L(0)=0,$ $w_L(1)=w_{L}(3)=1$ and $w_{L}(2)=2$ (see Hammons \textit{et al.} \cite{Hammons}). Borges \etal~\cite{Borges1} further extended the Gray map $\phi$  to the Gray map $\Phi:\mathbb{Z}_{4}^{\mathpzc{a}}\oplus \mathbb{Z}_{2}^{\mathpzc{b}}\rightarrow\mathbb{Z}_{2}^{2\mathpzc{a}+\mathpzc{b}}$ as $\Phi(c_1,c_2,\ldots,c_{\mathpzc{a}}|d_1,d_2,\ldots,d_{\mathpzc{b}})=(\phi(c_1),\phi(c_2),\ldots,\phi(c_{\mathpzc{a}}),d_1,d_2,\ldots,d_{\mathpzc{b}})$ for all $(c_1,c_2,\ldots,c_{\mathpzc{a}}|d_1,d_2,\ldots,d_{\mathpzc{b}})\in\mathbb{Z}_{4}^{\mathpzc{a}}\oplus \mathbb{Z}_{2}^{\mathpzc{b}}.$ Thus the Lee weight $w_L$ on $\mathbb{Z}_{4}$ can be naturally extended to the Lee weight $w_L$ on $\mathbb{Z}_{4}^{\mathpzc{a}}\oplus \mathbb{Z}_{2}^{\mathpzc{b}}$ as $w_L(m)=\sum\limits_{i=1}^{\mathpzc{a}}w_{L}(c_i)+\sum\limits_{j=1}^{\mathpzc{b}}w_H(d_j)$ for all $m=(c_1,c_2,\ldots,c_{\mathpzc{a}}|d_1,d_2,\ldots,d_{\mathpzc{b}}) \in \mathbb{Z}_{4}^{\mathpzc{a}}\oplus \mathbb{Z}_{2}^{\mathpzc{b}},$ where $w_H$ denotes the Hamming weight function. The Lee distance of a $\mathbb{Z}_{4} \mathbb{Z}_{2}$-linear code $\mathtt{C}$ of block-length $(\mathpzc{a},\mathpzc{b})$ is defined as $d_L(\mathtt{C})=\min\{w_L(c): c (\neq 0)\in \mathtt{C}\}.$ The Gray image of a $\mathbb{Z}_{4} \mathbb{Z}_{2}$-linear code $\mathtt{C}$ of block-length $(\mathpzc{a},\mathpzc{b})$ is a binary code of length $2 \mathpzc{a}+\mathpzc{b},$ which need not be a linear code. We are now ready to classify  all Euclidean $\mathbb{Z}_4 \mathbb{Z}_{2}$-LCD codes   of block-length $(\mathpzc{a},\mathpzc{b}),$ where $\mathpzc{a} \in \{1,2,3\}$ and $\mathpzc{b} \in \{1,2\}.$ From now on,  generator matrices of $\mathbb{Z}_{4} \mathbb{Z}_{2}$-linear codes with Gray images as non-trivial optimal binary Euclidean LCD codes achieving the Griesmer bound are marked $\ast.$ 
\begin{itemize}
\item[I.] There are  5 monomially inequivalent non-zero Euclidean $\mathbb{Z}_4\mathbb{Z}_2$-LCD codes of block-length $(1,1).$ Among these codes, there are
\vspace{-2mm}\begin{enumerate}
\item[$\blacklozenge$] 3 Euclidean $\mathbb{Z}_4\mathbb{Z}_2$-LCD codes of Lee distance $1$ with generator matrices $[1~|~0],$ $[0~|~1]$ and $\left[\begin{array}{c|c}
     1&0  \\  
     0&1 
\end{array}\right];$
\item[$\blacklozenge$] 1 Euclidean $\mathbb{Z}_4\mathbb{Z}_2$-LCD code of Lee distance $2$ with a generator matrix $[1~|~1];$ and
\item[$\blacklozenge$] 1 Euclidean $\mathbb{Z}_4\mathbb{Z}_2$-LCD code of Lee distance $3$ with a generator matrix $[2~|~1].$ 
\end{enumerate}
\item[II.] There are  11  monomially inequivalent non-zero Euclidean $\mathbb{Z}_4\mathbb{Z}_2$-LCD codes of block-length $(1,2).$ Among these codes, there are
\vspace{-2mm}\begin{enumerate}
\item[$\blacklozenge$] 7 Euclidean $\mathbb{Z}_4\mathbb{Z}_2$-LCD codes of Lee distance 1 with generator matrices $[0~|~I_2],$ $[0~|~1~ 0],$ $[1~|~0~0],$
$\left[\begin{array}{c|cc}
     1&0&1  \\  
     0&1&0 
\end{array}\right],$ $\left[\begin{array}{c|cc}
     2&1&0  \\  
     0&0&1 
\end{array}\right],$ $\left[\begin{array}{c|cc}
     1&0&0  \\  
     0&1&0 
\end{array}\right]$ and  $\left[\begin{array}{c|cc}
     1&0  \\
     0&I_2
\end{array}\right];$ 
\item[$\blacklozenge$] 3 Euclidean $\mathbb{Z}_4\mathbb{Z}_2$-LCD codes of Lee distance 2 with generator matrices $[1~|~1~ 0],$ $[1~|~1~1]$ and $\left[\begin{array}{c|cc}
     2&1&0  \\  
     2&0&1 
\end{array}\right];$ and
\item[$\blacklozenge$] 1 Euclidean $\mathbb{Z}_4\mathbb{Z}_2$-LCD code of Lee distance 3 with a generator matrix $[2~|~1~ 0].$ 
\end{enumerate}
\item[III.] There are  15  monomially inequivalent non-zero Euclidean $\mathbb{Z}_4\mathbb{Z}_2$-LCD codes of block-length $(2,1).$ Among these codes, there are
\begin{enumerate}
\vspace{-2mm}\item[$\blacklozenge$] 8 Euclidean $\mathbb{Z}_4\mathbb{Z}_2$-LCD codes of Lee distance 1 with generator matrices $[1~ 0~|~ 0],$ $[0~ 0~|~ 1],$ $[I_2~|~0],$ 
$\left[\begin{array}{cc|c}
     1&0&1  \\
     0&1&0
\end{array}\right],$ $\left[\begin{array}{cc|c}
     1&0&0  \\  
     0&0&1
\end{array}\right],$  $\left[\begin{array}{cc|c}
     1&2&0  \\  
     0&0&1
\end{array}\right],$ $\left[\begin{array}{cc|c}
     1&0&0  \\  
     0&2&1
\end{array}\right]$ and  $\left[\begin{array}{c|c}
     I_2&0  \\
     0&1
\end{array}\right];$
\item[$\blacklozenge$] 5 Euclidean $\mathbb{Z}_4\mathbb{Z}_2$-LCD codes of Lee distance 2 with generator matrices $[1~ 2~|~ 0],$ $[1~ 2~|~ 1],$ $[1~ 0~|~ 1],$
$\left[\begin{array}{cc|c}
     1&0&1  \\
     0&1&1
\end{array}\right]$ and  $\left[\begin{array}{cc|c}
     1&2&0  \\  
     0&2&1
\end{array}\right];$  
\item[$\blacklozenge$] 1 Euclidean $\mathbb{Z}_4\mathbb{Z}_2$-LCD code of Lee distance 3 with a generator matrix 
$[0~ 2~|~ 1];$ and
\item[$\blacklozenge$] 1 Euclidean $\mathbb{Z}_4\mathbb{Z}_2$-LCD code of Lee distance 5 with a generator matrix 
$[2~ 2~|~ 1].$
\end{enumerate}
\item[IV.] There are  41  monomially inequivalent non-zero Euclidean $\mathbb{Z}_4\mathbb{Z}_2$-LCD codes of block-length $(2,2).$ Among these codes, there are
\vspace{-2mm}\begin{itemize}
\item[$\blacklozenge$] 22 Euclidean $\mathbb{Z}_4\mathbb{Z}_2$-LCD codes of Lee distance 1 with generator matrices as listed below:
\\$\bullet$ $[0~0~|~1~0],$ $[1~0~|~0~0]$ and $\left[\begin{array}{c|c}
     I_2&0  \\
     0&I_2
\end{array}\right];$  
\\$\bullet$ $\left[\begin{array}{cc|cc}
     1&x&y&z  \\
     0&u&v&w
\end{array}\right]$ with $(x,y,z,u,v,w)\in\{(0,1,0,1,0,0),(0,0,0,1,1,1),(0,0,0,1,0,0),(0,0,1,0,1,0),\\(0,0,0,0,1,0),(2,0,1,0,1,0),(2,0,0,0,1,0),(0,0,0,2,1,0)\};$
\\$\bullet$ $\left[\begin{array}{cc|cc}
     x&y&1&0  \\
     z&u&0&1  
\end{array}\right]$ with $(x,y,z,u)\in\{(0,0,0,0),(0,0,2,2),(0,2,0,0)\};$
\\$\bullet$ $\left[\begin{array}{cc|cc}
     1&x&y&z  \\
     0&u&v&w\\  
     0&r&s&t
\end{array}\right]$ with $(x,y,z,u,v,w,r,s,t)\in\{(0,0,0,1,0,0,0,1,0),(0,0,1,1,0,1,0,1,0),(0,0,1,1,\\0,0,0,1,0),(0,0,0,2,1,0,0,0,1),(2,0,0,0,1,0,2,0,1),(2,0,0,0,1,0,0,0,1),(0,0,0,2,1,0,2,0,1),(0,0,\\0,0,1,0,0,0,1)\};$
\item[$\blacklozenge$] 15 Euclidean $\mathbb{Z}_4\mathbb{Z}_2$-LCD codes of Lee distance 2 with generator matrices as listed below:
\\$\bullet$ $[1~ x~|~ y~ z]$ with $(x,y,z)\in\{(2,1,1),(2,0,0),(2,1,0),(0,1,1),(0,1,0)\};$
\\$\bullet$ $\left[\begin{array}{cc|cc}
     1&x&y&z  \\
     0&u&v&w
\end{array}\right]$ with $(x,y,z,u,v,w)\in\{(0,1,1,1,1,1),(0,1,0,1,1,1),(0,0,1,1,0,1),(0,1,0,1,0,1),\\(2,0,1,2,1,0),(2,0,0,2,1,0),(0,0,1,2,1,0)\};$
\\$\bullet$ $\left[\begin{array}{cc|cc}
     x&y&1&0  \\
     z&u&0&1  
\end{array}\right]$ with $(x,y,z,u)\in\{(2,2,2,2),(2,0,2,0)\};$
\\$\bullet$ $\left[\begin{array}{cc|cc}
     1&2&0&0  \\  
     0&2&1&0  \\ 
     0&2&0&1 
\end{array}\right];$

\item[$\blacklozenge$] 3 Euclidean $\mathbb{Z}_4\mathbb{Z}_2$-LCD codes of Lee distance 3 with generator matrices $[0~2~|~1~0],$
$\left[\begin{array}{cc|cc}
     2&0&1&0  \\
     0&2&0&1  
\end{array}\right]$ and $\left[\begin{array}{cc|cc}
     0&2&1&0  \\
     2&2&0&1  
\end{array}\right];$ and  
\item[$\blacklozenge$] 1 Euclidean $\mathbb{Z}_4\mathbb{Z}_2$-LCD code of Lee distance 5 with a generator matrix $[2~2~|~1~0].$
\end{itemize}
\item[V.] There are  49  monomially inequivalent non-zero Euclidean $\mathbb{Z}_4\mathbb{Z}_2$-LCD codes of block-length $(3,1).$ Among these codes, there are
\vspace{-2mm}\begin{enumerate}
\item[$\blacklozenge$] 22 Euclidean $\mathbb{Z}_4\mathbb{Z}_2$-LCD codes of Lee distance 1 with generator matrices as listed below:
\\$\bullet$ $[0~0~0~|~1],$ $[1~0~0~|~0]$ and $\left[\begin{array}{c|c}
     I_3&0  \\
     0&1  
\end{array}\right];$
\\$\bullet$ $\left[\begin{array}{ccc|c}
     1&x&y&z  \\
     0&u&v&w  
\end{array}\right]$ with $(x,y,z,u,v,w)\in\{(2,2,0,0,0,1),(3,1,0,0,0,1),(0,0,0,2,0,1),(0,0,0,2,2,1),\\(0,0,0,0,0,1),(0,2,0,0,0,1),(0,0,1,1,0,0),(0,2,1,1,0,0),(0,0,0,1,0,0),(0,2,0,1,0,1)\};$
\\$\bullet$ $\left[\begin{array}{ccc|c}
     1&0&x&y  \\
     0&1&z&u  \\
     0&0&v&w
\end{array}\right]$ with $(x,y,z,u,v,w)\in\{(2,0,0,0,0,1),(0,0,0,0,0,1),(2,0,0,0,2,1),(0,0,0,0,2,1),\\(1,0,3,0,0,1),(2,0,2,0,0,1),(0,0,0,0,1,1),(0,2,0,0,1,1),(0,0,0,0,1,0)\};$
\item[$\blacklozenge$] 21 Euclidean $\mathbb{Z}_4\mathbb{Z}_2$-LCD codes of Lee distance 2 with generator matrices as listed below:
\\$\bullet$ $[1~x~y~|~z]$ with $(x,y,z)\in\{(2,2,1),(2,2,0),(0,0,1),(0,2,1),(0,2,0)\};$
\\$\bullet$ $\left[\begin{array}{ccc|c}
     1&x&y&z  \\
     0&u&v&w  
\end{array}\right]$ with $(x,y,z,u,v,w)\in\{(2,2,0,2,2,1),(0,2,0,0,2,1),(0,2,0,2,0,1),(2,2,0,2,0,1),\\(0,2,0,2,2,1),(0,2,0,1,2,0),(0,2,0,1,0,1),(0,2,2,1,2,1),(0,1,1,1,3,0),(0,1,0,1,3,0),(0,0,1,1,0,1),\\(0,2,1,1,2,0),(0,2,1,1,0,1)\};$
\\$\bullet$ $\left[\begin{array}{ccc|c}
     1&0&x&y  \\
     0&1&z&u  \\
     0&0&v&w
\end{array}\right]$ with $(x,y,z,u,v,w)\in\{(2,0,2,0,2,1),(3,0,3,0,2,1),(0,1,0,1,1,1)\};$
\item[$\blacklozenge$] 2 Euclidean $\mathbb{Z}_4\mathbb{Z}_2$-LCD codes of Lee distance 3 with generator matrices $[0~0~2~|~1]$ and $\left[\begin{array}{ccc|c}
     1&3&1&0  \\
     0&2&0&1 
\end{array}\right];$
\item[$\blacklozenge$] 2 Euclidean $\mathbb{Z}_4\mathbb{Z}_2$-LCD codes of Lee distance 4 with generator matrices $[1~1~3~|~0]$ and $[1~1~3~|~1];$
\item[$\blacklozenge$] 1 Euclidean $\mathbb{Z}_4\mathbb{Z}_2$-LCD code of Lee distance 5 with a generator matrix $[2~2~0~|~1];$ and
\item[$\blacklozenge$] 1 Euclidean $\mathbb{Z}_4\mathbb{Z}_2$-LCD code of Lee distance 7 with a generator matrix $[2~2~2~|~1];$
\end{enumerate}
\item[VI.] There are  163  monomially inequivalent non-zero Euclidean $\mathbb{Z}_4\mathbb{Z}_2$-LCD codes of block-length $(3,2).$ Among these codes, there are
\vspace{-2mm}\begin{enumerate}
\item[$\blacklozenge$] 76 Euclidean $\mathbb{Z}_4\mathbb{Z}_2$-LCD codes of Lee distance 1 with generator matrices as listed below:
\\$\bullet$ $[0~0~0~|~1~0],$ $[1~0~0~|~0~0]$ and $\left[\begin{array}{c|c}
     I_3&0  \\
     0&I_2 
\end{array}\right];$
\\$\bullet$  $\left[\begin{array}{ccc|cc}
    x&y&z&1&0\\
    u&v&w&0&1     
\end{array}\right]$ with $(x,y,z,u,v,w)\in\{(0,0,0,0,0,0),(2,2,0,0,0,0),(0,2,0,0,0,0),\\(0,0,0,2,2,2)\};$
\\$\bullet$ $\left[\begin{array}{ccc|cc}
1&x&y&0&z\\
0&u&v&1&w
\end{array}\right]$ with $(x,y,z,u,v,w)\in\{(0,0,1,0,0,0),(0,0,0,2,0,0),(0,2,1,0,0,0),\\(3,1,0,0,0,0),(0,0,0,2,2,0),(0,0,0,0,0,0),(0,2,0,0,0,0),(2,2,1,0,0,0),(3,1,1,0,0,0),\\(2,2,0,0,0,0)\};$
\\$\bullet$ $\left[\begin{array}{ccc|cc}
1&0&x&y&z\\
0&1&u&v&w
\end{array}\right]$ with $(x,y,z,u,v,w)\in\{(0,1,0,0,0,0),(2,1,0,0,0,0),(0,0,0,0,1,1),\\(2,1,1,0,0,0),(0,0,0,0,0,0),(2,0,0,0,0,0)\};$
\\$\bullet$ $\left[\begin{array}{ccc|cc}
1&x&y&0&0\\
0&z&u&1&0\\
0&v&w&0&1
\end{array}\right]$ with $(x,y,z,u,v,w)\in\{(2,0,0,0,2,2),(0,0,2,0,0,0),(3,3,0,0,0,0),\\(0,0,2,0,0,2),(3,3,0,2,0,0),(2,0,0,0,0,0),(2,0,0,2,0,0),(0,0,0,0,2,2),(0,0,0,0,0,0),(0,0,0,2,0,2),\\(2,2,0,0,2,0),(2,2,0,0,2,2),(0,0,2,2,2,0),(2,2,0,0,0,0),(0,0,2,2,2,2),(2,0,0,0,2,0)\};$
\\$\bullet$ $\left[\begin{array}{ccc|cc}
1&0&x&0&y\\
0&1&z&0&u\\
0&0&v&1&w
\end{array}\right]$ with $(x,y,z,u,v,w)\in\{(1,0,3,0,0,0),(2,0,0,1,0,0),(2,1,0,1,0,0),\\(2,1,2,1,0,0),(2,0,0,0,0,0),(2,0,2,0,0,0),(2,1,0,0,0,0),(2,0,0,0,2,0),(0,1,0,1,0,0),(2,1,2,0,0,0),\\(2,1,0,0,2,0),(0,0,0,0,0,0),(0,1,0,0,0,0),(0,0,0,0,2,0),(0,1,0,0,2,0),(1,1,3,1,0,0)\};$
\\$\bullet$ $\left[\begin{array}{ccc|cc}
1&0&0&x&y\\
0&1&0&z&u\\
0&0&1&v&w
\end{array}\right]$ with $(x,y,z,u,v,w)\in\{(0,0,0,0,1,0),(0,0,0,1,0,1),(0,0,1,1,1,1),\\(0,1,0,0,1,0),(0,0,1,1,1,0),(0,0,0,0,0,0),(0,0,0,0,1,1)\};$
\\$\bullet$ $\left[\begin{array}{ccc|cc}
1&0&x&0&0\\
0&1&y&0&0\\
0&0&z&1&0\\
0&0&u&0&1
\end{array}\right]$ with $(x,y,z,u)\in\{(2,0,2,0),(0,0,2,2),(0,2,0,0),(2,0,2,2),(2,2,0,0),\\(1,1,0,0),(0,0,0,0),(0,0,2,0),(2,2,0,2),(3,1,2,0)\};$
\\$\bullet$ $\left[\begin{array}{ccc|cc}
1&0&0&0&x\\
0&1&0&0&y\\
0&0&1&0&z\\
0&0&0&1&u
\end{array}\right]$ with $(x,y,z,u)\in\{(0,0,0,0),(0,1,1,0),(1,1,1,0),(0,0,1,0)\};$
\item[$\blacklozenge$] 71 Euclidean $\mathbb{Z}_4\mathbb{Z}_2$-LCD codes of Lee distance 2 with generator matrices as listed below:  
\\$\bullet$ $[1~x~y~|~z~u]$ with $(x,y,z,u)\in\{(2,2,1,1),(2,2,1,0),(0,0,1,1),(2,2,0,0),(0,0,1,0),(0,2,1,1),\\(0,2,1,0),(0,2,0,0)\};$
\\$\bullet$ $\left[\begin{array}{ccc|cc}
     x&y&z&1&0  \\
     u&v&w&0&1
\end{array}\right]$ with $(x,y,z,u,v,w)\in\{(2,2,2,2,2,2),(2,2,0,2,2,0),(2,0,0,2,0,0)\};$
\\$\bullet$ $\left[\begin{array}{ccc|cc}
1&x&y&0&z\\
0&u&v&1&w
\end{array}\right]$ with $(x,y,z,u,v,w)\in\{(0,2,0,2,0,0),(0,2,1,0,2,0),(0,2,0,2,2,0),\\(2,2,1,2,0,0),(0,2,0,0,2,0),(2,2,1,2,2,0),(2,2,0,2,0,0),(2,2,0,2,2,0),(0,0,1,2,0,0),\\(0,2,1,2,0,0),(0,0,1,2,2,0),(0,2,1,2,2,0)\};$ 
\\$\bullet$ $\left[\begin{array}{ccc|cc}
1&0&x&y&z\\
0&1&u&v&w
\end{array}\right]$ with $(x,y,z,u,v,w)\in\{(2,0,0,2,0,0),(2,1,1,0,1,0),(0,1,0,0,0,1),\\(2,1,0,2,0,0),(2,0,0,0,1,0),(0,1,1,0,1,1),(2,1,0,0,0,1),(2,1,1,0,1,1),(2,1,0,2,0,1),\\(2,1,1,2,1,1),(2,0,0,0,1,1),(1,0,1,3,0,1),(2,0,0,2,1,1),(0,1,0,0,1,1),(1,0,0,3,0,0),\\(2,1,0,0,1,1),(2,1,0,2,1,1),(0,0,1,0,0,1),(2,0,1,0,0,1),(1,1,1,3,1,1),(2,0,1,2,0,1)\};$
\\$\bullet$ $\left[\begin{array}{ccc|cc}
1&x&y&0&0\\
0&z&u&1&0\\
0&v&w&0&1
\end{array}\right]$ with $(x,y,z,u,v,w)\in\{(2,2,0,2,0,2),(3,3,0,2,0,2),(2,2,2,2,2,0),\\(2,2,2,2,2,2),(2,0,0,2,0,2),(0,2,0,2,0,2),(2,0,2,2,2,0),(0,2,2,2,2,0),(2,0,2,2,2,2),\\(2,2,0,2,2,0),(2,0,2,0,0,2)\};$
\\$\bullet$ $\left[\begin{array}{ccc|cc}
1&0&x&0&y\\
0&1&z&0&u\\
0&0&v&1&w
\end{array}\right]$ with $(x,y,z,u,v,w)\in\{(2,0,0,1,2,0),(2,1,0,1,2,0),(2,1,2,1,2,0),\\(2,0,2,0,2,0),(2,1,2,0,2,0),(0,1,0,1,2,0),(3,1,3,1,2,0),(3,0,3,0,2,0)\};$
\\$\bullet$ ${}^{\ast}\left[\begin{array}{ccc|cc}
1&0&0&x&y\\
0&1&0&z&u\\
0&0&1&v&w
\end{array}\right]$ with $(x,y,z,u,v,w)\in\{(0,1,1,0,1,0),(0,1,1,1,0,1),(1,0,1,0,1,0),\\(1,1,1,1,1,0),(1,1,1,1,1,1),(0,1,1,1,1,0)\};$
\\$\bullet$ ${}^{\ast}\left[\begin{array}{ccc|cc}
1&0&2&0&0\\
0&1&2&0&0\\
0&0&2&1&0\\
0&0&2&0&1
\end{array}\right]$ and $\left[\begin{array}{ccc|cc}
1&0&1&0&0\\
0&1&3&0&0\\
0&0&2&1&0\\
0&0&2&0&1
\end{array}\right];$\\
\item[$\blacklozenge$] 10 Euclidean $\mathbb{Z}_4\mathbb{Z}_2$-LCD codes of Lee distance 3 with  generator matrices as listed below:
\\$\bullet$ $[0~0~2~|~1~0]$ and $[1~1~3~|~0~0];$
\\$\bullet$ $\left[\begin{array}{ccc|cc}
     x&y&z&1&0 \\
     u&v&w&0&1
\end{array}\right]$ with $(x,y,z,u,v,w)\in\{(0,2,0,2,2,2),(0,2,0,2,0,0),(0,0,2,2,2,0),\\(2,0,0,2,0,2)\};$
\\$\bullet$ $\left[\begin{array}{ccc|cc}
1&x&y&0&z\\
0&u&v&1&w
\end{array}\right]$ with $(x,y,z,u,v,w)\in\{(3,1,1,2,0,0),(3,1,0,2,0,0),(0,1,0,1,3,1)\};$
\\$\bullet$ $\left[\begin{array}{ccc|cc}
1&1&3&0&0\\
0&2&0&1&0\\
0&0&2&0&1
\end{array}\right];$
\item[$\blacklozenge$] 2 Euclidean $\mathbb{Z}_4\mathbb{Z}_2$-LCD codes of Lee distance 4 with generator matrices $[1~1~3~|~1~0]$ and \\$\left[\begin{array}{ccc|cc}
      0&2&2&1&0\\
     2&2&2&0&1
\end{array}\right];$  
\item[$\blacklozenge$] 3 Euclidean $\mathbb{Z}_4\mathbb{Z}_2$-LCD codes of Lee distance 5 with generator matrices $[2~2~0~|~1~0],$ $[1~1~3~|~1~1]$ and  $^*\left[\begin{array}{ccc|cc}
      2&2&0&1&0\\
     2&0&2&0&1
\end{array}\right];$ and    
\item[$\blacklozenge$] 1 Euclidean $\mathbb{Z}_4\mathbb{Z}_2$-LCD code of Lee distance 7 with a generator matrix $[2~2~2~|~1~0].$
\end{enumerate}
\end{itemize}
\noindent \textbf{\large B.} \textbf{\large Classification of Euclidean $\mathbb{Z}_9\mathbb{Z}_3$-LCD codes of block-length $(\mathpzc{a},\mathpzc{b}),$ where $\mathpzc{a} \in \{1,2,3\}$ and $\mathpzc{b} \in \{1,2\}$}\\

 We will next classify all Euclidean $\mathbb{Z}_9\mathbb{Z}_3$-LCD codes of block-length $(\mathpzc{a},\mathpzc{b})$  up to monomial equivalence, where $\mathpzc{a}\in\{1,2,3\}$ and $\mathpzc{b}\in\{1,2\}.$ For this, we will place the extended Lee weight $w_L$ defined by Yildiz and {\"O}zger \cite{Yildiz}) on $\mathbb{Z}_{9},$  as it is more suitable for coding theory applications as compared to the Lee weight defined by Bhaintwal and Wasan \cite{Bhaintwal}. The extended Lee weight $w_L$ on $\mathbb{Z}_{9}$ is defined as $w_L(0)=0,$ $w_L(1)=w_{L}(8)=1,$ $w_{L}(2)=w_{L}(7)=2$ and $w_{L}(3)=w_{L}(4)=w_{L}(5)=w_{L}(6)=3.$ It can be naturally extended to the Lee weight on $\mathbb{Z}_9^{\mathpzc{a}} \oplus \mathbb{Z}_{3}^{\mathpzc{b}}$ as $w_L(m)=\sum\limits_{i=1}^{\mathpzc{a}}w_{L}(c_i)+\sum\limits_{j=1}^{\mathpzc{b}}w_H(d_j)$ for all $m=(c_1,c_2,\ldots,c_{\mathpzc{a}}|d_1,d_2,\ldots,d_{\mathpzc{b}}) \in \mathbb{Z}_9^{\mathpzc{a}} \oplus \mathbb{Z}_{3}^{\mathpzc{b}}.$  The Lee distance of a $\mathbb{Z}_{9} \mathbb{Z}_{3}$-linear code $\mathtt{C}$ of block-length $(\mathpzc{a},\mathpzc{b})$ is defined as $d_L(\mathtt{C})=\min\{w_L(c): c (\neq 0)\in \mathtt{C}\}.$ 

\begin{itemize}
\item[I.] There are  5 monomially inequivalent non-zero Euclidean $\mathbb{Z}_9\mathbb{Z}_3$-LCD codes of block-length $(1,1).$ Among these codes, there are
\vspace{-2mm}\begin{enumerate}
\item[$\blacklozenge$] 3 Euclidean $\mathbb{Z}_9\mathbb{Z}_3$-LCD codes of Lee distance 1 with generator matrices
 $[1~|~0],$ $[0~|~1]$ and $\left[\begin{array}{c|c}
     1&0  \\
     0&1
\end{array}\right];$ 
\item[$\blacklozenge$] 1 Euclidean $\mathbb{Z}_9\mathbb{Z}_3$-LCD code of Lee distance 2 with a generator matrix $[1~|~1];$ and
\item[$\blacklozenge$] 1 Euclidean $\mathbb{Z}_9\mathbb{Z}_3$-LCD code of Lee distance 4 with a generator matrix $[3~|~1].$
\end{enumerate}
\item[II.] There are 15 monomially inequivalent non-zero Euclidean $\mathbb{Z}_9\mathbb{Z}_3$-LCD codes of block-length $(1,2).$ Among these codes, there are
\vspace{-2mm}\begin{enumerate}
\item[$\blacklozenge$] 8 Euclidean $\mathbb{Z}_9\mathbb{Z}_3$-LCD codes of Lee distance 1 with generator matrices
$[1~|~0~0],$ $[0~|~1~0],$
$\left[\begin{array}{c|cc}
        1 & 0 & 0\\  
        0& 1& 0
    \end{array}\right],$ $\left[\begin{array}{c|cc}
        1 & 0 & 0\\   
        0 & 1 & 1
    \end{array}\right],$ $\left[\begin{array}{c|cc}
        1 & 0 & 1\\   
        0 & 1 & 0
    \end{array}\right],$ $\left[\begin{array}{c|cc}
        3 & 1 & 0\\ 
        0 & 0 & 1
    \end{array}\right],$ $\left[\begin{array}{c|cc}
        0 & 1 & 0\\ 
        0 & 0 & 1
    \end{array}\right]$ and $\left[\begin{array}{c|c}
     1&0  \\
     0&I_2
\end{array}\right];$   
\item[$\blacklozenge$] 4 Euclidean $\mathbb{Z}_9\mathbb{Z}_3$-LCD codes of Lee distance 2 with generator matrices 
$[1~|~0~2],$ $[0~|~1~2],$ 
$\left[\begin{array}{c|cc}
        1 & 0 & 1\\   
        0 & 1 & 2
    \end{array}\right]$  and $\left[\begin{array}{c|cc}
        6 & 1 & 0\\
        3 & 0 & 1
\end{array}\right];$ 
\item[$\blacklozenge$] 1 Euclidean $\mathbb{Z}_9\mathbb{Z}_3$-LCD code of Lee distance 3 with a generator matrix $[1~|~1~2];$
\item[$\blacklozenge$] 1 Euclidean $\mathbb{Z}_9\mathbb{Z}_3$-LCD code of Lee distance 4 with a generator matrix $[6~|~1~0];$ and
\item[$\blacklozenge$] 1 Euclidean $\mathbb{Z}_9\mathbb{Z}_3$-LCD code of Lee distance 5 with a generator matrix $[6~|~1~2].$
\end{enumerate}
\item[III.] There are precisely 19  monomially inequivalent non-zero Euclidean $\mathbb{Z}_9\mathbb{Z}_3$-LCD codes of block-length $(2,1).$ Among these codes, there are
\vspace{-2mm}\begin{enumerate}
\item[$\blacklozenge$] 9 Euclidean $\mathbb{Z}_9\mathbb{Z}_3$-LCD codes of Lee distance 1 with generator matrices $[1~0~|~0],$ $[0~0~|~1],$ 
$\left[\begin{array}{cc|c}
    1 &  0 & 0  \\
    0 & 1 & 0
\end{array}\right],$ $\left[\begin{array}{cc|c}
    1 &  0 & 0  \\
    0 & 1 & 2
\end{array}\right],$ $\left[\begin{array}{cc|c}
        1 & 0 & 0  \\  
         0 & 6 & 1
    \end{array}\right],$ $\left[\begin{array}{cc|c}
        1 & 0 & 0  \\  
         0 & 0 & 1
    \end{array}\right],$ $\left[\begin{array}{cc|c}
        1 & 8 & 0  \\  
         0 & 0 & 1
    \end{array}\right],$ $\left[\begin{array}{cc|c}
        1 & 6 & 0  \\  
         0 & 0 & 1
    \end{array}\right]$ and $\left[\begin{array}{c|c}
     I_2&0  \\
     0&1
\end{array}\right];$ 
\item[$\blacklozenge$] 3 Euclidean $\mathbb{Z}_9\mathbb{Z}_3$-LCD codes of Lee distance 2 with generator matrices  $[1~0~|~2],$ 
$\left[\begin{array}{cc|c}
    1 &  0 & 1  \\
    0 & 1 & 2
\end{array}\right]$ and  $\left[\begin{array}{cc|c}
        1 & 3 & 0  \\  
         0 & 3 & 1
    \end{array}\right];$
\item[$\blacklozenge$] 4 Euclidean $\mathbb{Z}_9\mathbb{Z}_3$-LCD codes of Lee distance 3 with generator matrices
 $[1~4~|~0],$  $[1~6~|~0],$ $[1~6~|~2]$ and 
$\left[\begin{array}{cc|c}
        1 & 2 & 0  \\  
         0 & 6 & 1
    \end{array}\right];$
\item[$\blacklozenge$] 2 Euclidean $\mathbb{Z}_9\mathbb{Z}_3$-LCD codes of Lee distance 4 with generator matrices $[1~4~|~2]$ and $[3~0~|~1];$ and
\item[$\blacklozenge$] 1 Euclidean $\mathbb{Z}_9\mathbb{Z}_3$-LCD codes of Lee distance 7 with a generator matrix $[6~3~|~1].$
\end{enumerate}
\item[IV.] There are precisely 71 monomially inequivalent non-zero Euclidean $\mathbb{Z}_9\mathbb{Z}_3$-LCD codes of block-length $(2,2).$ Among these codes, there are
\vspace{-2mm}\begin{enumerate}
\item[$\blacklozenge$] 30 Euclidean $\mathbb{Z}_9\mathbb{Z}_3$-LCD codes of Lee distance 1 with generator matrices as listed below:
\\$\bullet$ $[1~0~|~0~0],$ $[0~0~|~1~0]$ and $\left[\begin{array}{c|c}
     I_2&0  \\
     0&I_2
\end{array}\right];$
\\$\bullet$ $\left[\begin{array}{cc|cc}
    1& 0& x &y\\
    0 &1 &z& u
\end{array}\right]$ with $(x,y,z,u)\in\{(0,0,0,0),(1,2,0,0),(0,0,0,1)\};$
\\$\bullet$ $\left[\begin{array}{cc|cc}
    1& x& 0 &y\\  
    0 &z &1& u
\end{array}\right]$ with $(x,y,z,u)\in\{(0,0,6,0),(0,0,6,1),(8,1,0,0),(8,0,0,0),(0,1,0,0),(0,0,0,0),\\(0,0,0,1),(6,1,0,0),(6,0,0,0)\};$
\\$\bullet$ $\left[\begin{array}{cc|cc}
    x& y& 1 &0\\
    z &u &0& 1
\end{array}\right]$ with $(x,y,z,u)\in\{(0,0,0,0),(3,6,0,0),(6,0,0,0)\};$
\\$\bullet$ $\left[\begin{array}{cc|cc}
    1& 0& 0 &x\\
    0 &1 &0& y\\  
    0 &0& 1& z
\end{array}\right]$ with $(x,y,z)\in\{(0,0,1),(1,0,1),(0,1,0),(1,2,0),(0,0,0)\};$
\\$\bullet$ $\left[\begin{array}{cc|cc}
    1& x& 0 &0\\  
    0 &y &1& 0\\ 
    0 &z& 0& 1
\end{array}\right]$ with $(x,y,z)\in\{(0,6,6),(6,0,0),(7,0,0),(0,0,6),(6,0,6),(1,0,3),(0,0,0)\};$ 
\item[$\blacklozenge$] 19 Euclidean $\mathbb{Z}_9\mathbb{Z}_3$-LCD codes of Lee distance 2 with generator matrices as listed below:
\\$\bullet$ $[1~0~|~ 0~ 2]$ and $[0~0~|~ 1~ 2];$
\\$\bullet$ $\left[\begin{array}{cc|cc}
    1& x& y &z\\
    0 &u &v& w
\end{array}\right]$ with $(x,y,z,u,v,w)\in\{(0,0,2,1,0,2),(0,1,1,1,1,0),(0,1,0,1,0,2),(0,2,1,1,2,1),\\(6,0,1,0,1,2),(3,0,2,3,1,2),(6,0,0,0,1,1),(0,0,1,6,1,0),(0,0,1,0,1,2),(8,0,1,0,1,2),\\(3,0,0,3,1,0),(8,0,0,0,1,1)\};$
\\$\bullet$  $\left[\begin{array}{cc|cc}
    x& y& 1 &0\\
    z &u &0& 1
\end{array}\right]$ with $(x,y,z,u)\in\{(3,6,6,3),(0,3,0,6)\};$
\\$\bullet$ $\left[\begin{array}{cc|cc}
    1& x& 0 &y\\
    0 &z &u& v\\  
    0 &w& r& 1
\end{array}\right]$ with $(x,y,z,u,v,w,r)\in\{(0,1,1,0,1,0,1),(6,0,6,1,0,6,0),(1,0,6,1,0,3,0)\};$ 
\item[$\blacklozenge$] 12 Euclidean $\mathbb{Z}_9\mathbb{Z}_3$-LCD codes of Lee distance 3 with generator matrices as listed below:
\\$\bullet$ $[1~ x~|~ y~ z]$ with $(x,y,z)\in\{(4,0,0),(0,1,2),(6,0,0),(6,0,2),(6,1,2)\};$
\\$\bullet$ $\left[\begin{array}{cc|cc}
    1& x& y &z\\
    0 &u &v& w
\end{array}\right]$ with $(x,y,z,u,v,w)\in\{(0,1,2,1,1,1),(2,0,1,6,1,2),(3,0,1,3,1,2),(2,0,0,6,1,0),\\(3,0,1,3,1,0),(2,0,0,6,1,1),(3,0,0,3,1,1)\};$ 
\item[$\blacklozenge$] 5 Euclidean $\mathbb{Z}_9\mathbb{Z}_3$-LCD codes of Lee distance 4 with generator matrices $ [3~ 0~|~ 1~ 0],$ $ [1~ 4~|~ 0~ 2],$ $\left[\begin{array}{cc|cc}
    1& 2& 0 &0\\  
    0 &6 &1& 1
\end{array}\right],$ $\left[\begin{array}{cc|cc}
    6& 0& 1 &0\\
0 &3 &0& 1
\end{array}\right]$ and $\left[\begin{array}{cc|cc}
    6& 0& 1 &0\\
6 &6 &0& 1
\end{array}\right];$ 
\item[$\blacklozenge$] 3 Euclidean $\mathbb{Z}_9\mathbb{Z}_3$-LCD codes of Lee distance 5 with generator matrices 
$ [1~ 4~|~ 1~ 2],$ $ [3~ 0~|~ 1~ 2]$ and 
$\left[\begin{array}{cc|cc}
    3& 6& 1 &0\\
3 &3 &0& 1
\end{array}\right];$ 
\item[$\blacklozenge$] 1 Euclidean $\mathbb{Z}_9\mathbb{Z}_3$-LCD code of Lee distance 7 with a generator matrix $ [6~ 3~|~ 1~ 0];$ and
\item[$\blacklozenge$] 1 Euclidean $\mathbb{Z}_9\mathbb{Z}_3$-LCD code of Lee distance 8 with a generator matrix $[6~ 3~|~ 1~ 2].$
\end{enumerate}
\item[V.] There are precisely 53 monomially inequivalent non-zero Euclidean $\mathbb{Z}_9\mathbb{Z}_3$-LCD codes of block-length $(3,1).$ Among these codes, there are
\vspace{-2mm}\begin{enumerate}
\item[$\blacklozenge$] 25 Euclidean $\mathbb{Z}_9\mathbb{Z}_3$-LCD codes of Lee distance 1 with generator matrices as listed below:
\\$\bullet$ $[0~0~0~|1],$ $[1~0~0~|~0]$ and $\left[\begin{array}{c|c}
      I_3&0  \\
     0&1
 \end{array}\right];$
\\$\bullet$ $\left[\begin{array}{ccc|c}
1&x&y&0\\
0&z&u&1
\end{array}\right]$ with $(x,y,z,u)\in\{(0,7,0,0),(4,6,0,0),(3,0,0,0),(0,0,0,0),(3,6,0,0)\};$
\\$\bullet$ $\left[\begin{array}{ccc|c}
1&0&x&y\\
0&1&z&u
\end{array}\right]$ with $(x,y,z,u)\in\{(0,0,0,2),(6,1,0,0),(2,0,0,0),(0,0,0,0),(6,0,0,0),(2,1,0,0)\};$
\\$\bullet$ $\left[\begin{array}{ccc|c}
1&0&x&0\\
0&1&y&0\\
0&0&z&1
\end{array}\right]$ with $(x,y,z)\in\{(6,6,0),(0,3,3),(8,3,0),(0,6,0),(0,0,0),(4,0,3),(0,2,0),(0,0,6)\};$
\\$\bullet$  $\left[\begin{array}{ccc|c}
1&0&0&x\\
0&1&0&y\\
0&0&1&z
\end{array}\right]$ with $(x,y,z)\in\{(1,2,0),(0,0,0),(0,2,0)\};$   
\item[$\blacklozenge$] 14 Euclidean $\mathbb{Z}_9\mathbb{Z}_3$-LCD codes of Lee distance 2 with generator matrices as listed below:
\\$\bullet$ $[1~x~y~|~z]$ with $(x,y,z)\in\{(0,0,2),(1,0,0)\};$
\\$\bullet$ $\left[\begin{array}{ccc|c}
1&0&x&y\\
0&1&z&u
\end{array}\right]$ with $(x,y,z,u)\in\{(6,1,0,2),(3,0,3,0),(2,0,0,2),(3,2,3,2),(8,0,3,0),(2,1,0,2),\\(8,0,3,2),(6,0,0,2),(0,1,0,2)\};$
\\$\bullet$ $\left[\begin{array}{ccc|c}
1&0&x&y\\
0&1&z&u\\
0&0&v&1
\end{array}\right]$ with $(x,y,z,u,v)\in\{(8,0,3,0,6),(6,0,3,0,3),(0,1,0,1,1)\};$ 
\item[$\blacklozenge$] 9 Euclidean $\mathbb{Z}_9\mathbb{Z}_3$-LCD codes of Lee distance 3 with generator matrices as listed below:
\\$\bullet$ $[1~x~y~|~z]$ with $(x,y,z)\in\{(3,0,0),(3,0,2),(3,6,0),(3,6,2),(1,0,2)\};$
\\$\bullet$ $\left[\begin{array}{ccc|c}
1&0&x&y\\
0&1&z&u
\end{array}\right]$ with $(x,y,z,u)\in\{(3,1,3,2),(8,1,3,0),(8,1,3,2),(3,0,3,2)\};$
\item[$\blacklozenge$] 1 Euclidean $\mathbb{Z}_9\mathbb{Z}_3$-LCD code of Lee distance 4 with a generator matrix $[0~3~0~|~1];$
\item[$\blacklozenge$] 2 Euclidean $\mathbb{Z}_9\mathbb{Z}_3$-LCD codes of Lee distance 6 with generator matrices $[1~6~7~|~0]$ and $[1~6~7~|~2];$  

\item[$\blacklozenge$] 1 Euclidean $\mathbb{Z}_9\mathbb{Z}_3$-LCD code of Lee distance 7 with a generator matrix $[6~6~0~|~1];$ and
\item[$\blacklozenge$] 1 Euclidean $\mathbb{Z}_9\mathbb{Z}_3$-LCD code of Lee distance 10 with a generator matrix $[6~6~3~|~1].$
\end{enumerate}
\item[VI.] There are precisely 336 monomially inequivalent non-zero Euclidean $\mathbb{Z}_9\mathbb{Z}_3$-LCD codes of block-length $(3,2).$ Among these codes, there are
\vspace{-2mm}\begin{enumerate}
\item[$\blacklozenge$] 120 Euclidean $\mathbb{Z}_9\mathbb{Z}_3$-LCD codes of Lee distance 1 whose generator matrices are listed below:

$\bullet$ $[0~0~0~|~1~0],$ $[1~0~0~|~0~0],$ $\left[\begin{array}{c|c}
I_3&0\\
0&I_2
\end{array}\right];$ 
\\$\bullet$ $\left[\begin{array}{ccc|cc}
x&y&z&1&0\\
0&0&0&0&1
\end{array}\right]$ with $(x,y,z)\in\{(6,0,0),(0,3,6),(3,3,3),(0,0,0)\};$
\\ $\bullet$ $\left[\begin{array}{ccc|cc}
1&x&y&0&z\\
0&u&v&1&w
\end{array}\right]$ with $(x,y,z,u,v,w)\in\{(0,0,0,6,6,0),(0,0,0,0,0,0),(0,0,0,0,6,1),\\(0,0,0,6,6,1),(3,6,0,0,0,0),(0,0,0,0,0,1),(3,0,0,0,0,0),(0,0,1,0,0,0),(0,7,0,0,0,0),\\(4,6,0,0,0,0),(3,0,1,0,0,0),(0,7,1,0,0,0),(0,0,0,0,6,0),(3,6,1,0,0,0)\};$
\\ $\bullet$ $\left[\begin{array}{ccc|cc}
1&0&x&y&z\\
0&1&0&0&u
\end{array}\right]$ with $(x,y,z,u)\in\{(0,0,0,1),(0,1,2,0),(0,0,0,0),(2,0,1,0),(6,0,1,0),\\(2,0,0,0),(6,0,0,0),(2,1,2,0),(6,1,2,0)\};$
\\$\bullet$ $\left[\begin{array}{ccc|cc}
1&0&x&0&y\\
0&1&z&0&u\\
0&0&v&1&w
\end{array}\right]$ with $(x,y,z,u,v,w)\in\{(0,0,3,1,3,0),(0,0,2,0,0,1),(0,2,6,0,0,0),\\
(0,0,0,1,0,0),(6,0,6,0,0,0),(0,0,6,1,0,1),(6,1,6,1,0,0),(8,2,3,0,0,0),(0,0,3,1,3,1),\\(4,2,0,0,3,0),(0,0,6,0,0,0),(0,2,2,0,0,0),(0,0,0,0,6,0),(0,0,3,0,3,0),(0,0,0,0,0,0),\\(0,0,2,1,0,1),(6,1,6,2,0,0),(0,0,6,0,0,1),(0,0,0,0,6,1),(0,0,3,0,3,1),(0,1,6,2,0,0),\\(0,0,0,0,0,1),(4,0,0,0,3,1),(0,1,0,2,0,0),(0,0,2,0,0,1),(8,1,3,2,0,0),(8,0,3,1,0,0),\\(0,1,2,2,0,0),(0,0,2,1,0,0),(0,0,3,1,3,2),(6,0,6,1,0,0),(0,1,0,0,0,1), (4,1,0,0,3,1),\\(0,0,6,1,0,0),(0,0,0,1,6,0),(8,0,3,0,0,0),(4,0,0,0,3,0)\};$
\\$\bullet$  $\left[\begin{array}{ccc|cc}
1 & x & y & 0 & 0 \\
0 & z & u & 1 & 0 \\
0 & v & w & 0 & 1
\end{array}\right]$ with $(x,y,z,u,v,w)\in\{(0,0,0,6,6,6),(6,6,0,6,0,0),(3,0,0,0,0,0),\\(0,8,0,6,0,0),(3,3,0,0,6,3),(0,0,6,3,6,3),(0,1,0,0,0,0),(0,0,0,0,0,0),(0,0,3,0,0,0),\\(0,5,6,6,0,0),(2,6,0,0,0,0),(6,7,0,0,3,0),(0,0,3,3,6,3),(3,6,0,0,3,6),(6,8,0,0,3,6),\\(0,0,0,3,3,0),(0,3,6,0,0,0),(3,3,0,0,0,0),(0,0,0,0,3,3),(0,0,0,6,0,6),(6,0,6,0,0,0),\\(2,6,0,0,6,0),(3,0,3,6,0,0),(0,1,0,0,6,0),(0,0,0,6,6,6)\};$ 
\\$\bullet$ $\left[\begin{array}{ccc|cc}
1 & 0 & 0 & x & y \\
0 & 1 & 0 & z & u \\
0 & 0 & 1 & v & w
\end{array}\right]$ with $(x,y,z,u,v,w)\in\{(0,0,0,0,1,2),(0,0,1,2,2,1),(0,0,0,0,0,0),\\(0,1,0,0,0,0),(1,2,0,0,0,2),(0,1,0,2,0,0),(0,0,0,2,1,0),(0,0,2,2,1,2)\};$
\\$\bullet$  $\left[\begin{array}{ccc|cc}
1 & 0 & x & 0 & 0 \\
0 & 1 & y & 0 & 0 \\
0 & 0 & z & 1 & 0 \\
0 & 0 & u & 0 & 1
\end{array}\right]$ with $(x,y,z,u)\in\{(5,6,0,0),(3,0,0,0),(0,6,6,3),(0,0,0,6),(0,1,0,0),\\(3,6,6,0),(0,0,3,6),(0,6,0,3),(6,5,6,0),(0,0,0,0),(3,6,0,0),(2,0,0,3),(0,2,3,6)\};$
\\$\bullet$ $\left[\begin{array}{ccc|cc}
1 & 0 & 0 & 0 & x \\
0 & 1 & 0 & 0 & y \\
0 & 0 & 1 & 0 & z \\
0 & 0 & 0 & 1 & u
\end{array}\right]$ with $(x,y,z,u)\in\{(0,0,0,0),(0,0,0,2),(1,0,0,2),(2,2,1,0),(0,2,1,2),\\(1,0,2,0),(0,2,0,0)\};$
\item[$\blacklozenge$] 101 Euclidean $\mathbb{Z}_9\mathbb{Z}_3$-LCD codes of Lee distance 2 whose generator matrices are listed below:
\\$\bullet$ $[0~0~0~|~1~2],$ $[1~0~0~|~0~2]$ and $[1~1~0~|~0~0];$ 
\\$\bullet$ $\left[\begin{array}{ccc|cc}
3&0&6&1&0\\
6&0&3&0&1
\end{array}\right],$ $\left[\begin{array}{ccc|cc}
0&0&3&1&0\\
0&0&3&0&1
\end{array}\right]$ and $\left[\begin{array}{ccc|cc}
3&6&6&1&0\\
6&3&3&0&1
\end{array}\right];$ \\
$\bullet$ $\left[\begin{array}{ccc|cc}
1&x&y&0&z\\
0&u&v&1&w
\end{array}\right]$ with $(x,y,z,u,v,w)\in\{(0,7,1,0,0,2),(6,3,0,3,6,0),(0,6,0,0,3,0),\\(4,6,1,0,0,2),(0,0,1,0,6,0),(3,6,0,0,0,1),(0,0,1,6,6,2),(3,0,0,0,0,1),(0,0,1,0,0,2),\\(0,7,0,0,0,1),(0,0,1,6,6,0),(0,6,1,0,3,2),(3,6,1,0,0,2),(4,6,1,0,0,1),(3,0,1,0,0,2)\};$    \\$\bullet$ $\left[\begin{array}{ccc|cc}
1&0&x&y&z\\
0&1&u&v&w
\end{array}\right]$ with $(x,y,z,u,v,w)\in\{(0,2,1,0,2,1),(6,0,2,0,0,2),(3,0,2,3,0,2),\\(8,0,0,3,0,1),(6,1,1,0,1,0),(2,1,0,0,0,2),(6,1,0,0,0,2),(3,0,0,3,0,0),(8,0,0,3,0,0),\\(0,0,2,0,0,2),(2,0,0,0,0,1),(2,0,2,0,0,2),(6,0,0,0,0,1),(0,1,1,0,1,0),(0,1,0,0,0,2),\\(8,0,0,3,1,2),(3,2,1,3,2,1),(2,1,1,0,1,0)\};$   
\\$\bullet$  $\left[\begin{array}{ccc|cc}
1 & 0 & x & 0 & y \\
0 & 1 & z & 0 & u \\
0 & 0 & v & 1 & w
\end{array}\right]$ with $(x,y,z,u,v,w)\in\{(6,1,6,2,0,1),(0,2,3,0,3,0),(8,2,3,0,6,0),\\(6,0,3,0,3,0),(0,1,3,1,3,2),(8,0,3,1,6,1),(8,0,3,0,6,0),(6,1,3,2,3,2),(8,0,3,1,0,1),\\(4,0,0,1,3,1),(6,0,3,0,3,1),(6,1,6,1,0,1),(6,1,3,2,3,0),(8,0,3,0,6,1),(0,1,6,1,0,1),\\(0,1,3,1,3,1),(0,1,0,2,6,0),(0,1,3,2,3,0),(8,0,3,0,0,1),(0,1,0,1,0,1),(8,0,3,1,6,0),\\(8,1,3,1,0,1),(6,1,6,0,0,1),(4,1,0,1,3,1),(6,1,3,2,3,1),(0,1,2,1,0,1),(4,1,0,2,3,0),\\(4,0,0,1,3,0),(6,0,6,0,0,1),(0,1,6,0,0,1),(0,1,3,0,3,1),(6,0,3,1,3,0),(8,1,3,0,0,1),\\(0,1,2,0,0,1)\};$
\\$\bullet$ $\left[\begin{array}{ccc|cc}
1 & x & y & 0 & 0 \\
0 & z & u & 1 & 0 \\
0 & v & w & 0 & 1
\end{array}\right]$ with $(x,y,z,u,v,w)\in\{(6,6,3,6,3,6),(6,3,6,3,6,6),(0,3,6,0,6,0),\\(6,0,6,0,6,0),(8,0,6,3,3,0),(6,6,0,6,0,3),(3,0,3,6,6,3),(3,4,0,3,0,6),(3,8,6,0,6,0),\\(1,6,6,3,3,6),(0,6,6,6,0,6),(0,5,3,0,3,0),(6,6,6,6,3,3),(3,3,6,0,3,3),(6,0,6,0,0,3),\\(7,0,3,6,6,3),(1,0,3,0,6,0)\};$
\\$\bullet$ $\left[\begin{array}{ccc|cc}
1 & 0 & 0 & x & y \\
0 & 1 & 0 & z & u \\
0 & 0 & 1 & v & w
\end{array}\right]$ with $(x,y,z,u,v,w)\in\{(2,1,2,1,1,2),(1,0,0,2,0,2),(1,0,2,0,2,0),\\(2,2,2,2,1,2),(1,2,0,1,0,1),(1,1,1,1,0,2),(1,2,0,2,2,2),(0,2,2,2,1,0)\};$ 
\\$\bullet$  $\left[\begin{array}{ccc|cc}
1 & 0 & 3 & 0 & 0 \\
0 & 1 & 7 & 0 & 0 \\
0 & 0 & 6 & 1 & 0 \\
0 & 0 & 6 & 0 & 1
\end{array}\right],$ 
$\left[\begin{array}{ccc|cc}
1 & 0 & 3 & 0 & 0 \\
0 & 1 & 6 & 0 & 0 \\
0 & 0 & 6 & 1 & 0 \\
0 & 0 & 6 & 0 & 1
\end{array}\right]$ and $\left[\begin{array}{ccc|cc}
1 & 0 & 0 & 0 & 1 \\
0 & 1 & 0 & 0 & 2 \\
0 & 0 & 1 & 0 & 1 \\
0 & 0 & 0 & 1 & 1
\end{array}\right];$
\item[$\blacklozenge$] 77 Euclidean $\mathbb{Z}_9\mathbb{Z}_3$-LCD codes of Lee distance 3 whose generator matrices are listed below:
\\$\bullet$ $[1~x~y~|~z~u]$ with $(x,y,z,u)\in\{(3,6,1,2),(1,0,0,2),(3,6,0,0),(0,0,1,2),(3,0,0,2),(3,6,0,2),\\(3,0,1,2),(3,0,0,0)\};$
\\$\bullet$ $\left[\begin{array}{ccc|cc}
1&x&y&0&z\\
0&u&v&1&w
\end{array}\right]$ with $(x,y,z,u,v,w)\in\{(6,0,1,3,6,2),(6,6,1,3,6,0),(0,5,0,3,6,0),\\(3,3,0,6,0,0),(6,0,1,3,6,0),(0,6,0,6,0,1),(6,0,1,3,6,1),(0,5,0,3,6,1),(3,3,0,6,0,1),\\(6,6,0,3,6,0),(6,3,0,3,6,1),(0,6,0,0,3,1),(7,0,0,0,6,0),(4,6,1,0,0,0),(0,6,1,0,3,2),\\(3,3,1,6,0,0),(6,0,0,3,6,0),(0,6,1,6,0,2),(7,0,0,0,6,1),(6,0,0,3,6,1),(2,0,0,3,0,0),\\(0,6,1,6,0,0),(6,3,1,3,6,0),(0,6,1,0,3,0),(2,0,0,3,0,1),(3,3,2,6,0,2),(6,3,2,3,6,2),\\(6,6,0,3,6,1),(0,6,0,6,0,0)\};$
\\$\bullet$  $\left[\begin{array}{ccc|cc}
1&0&x&y&z\\
0&1&u&v&w
\end{array}\right]$ with $(x,y,z,u,v,w)\in\{(8,1,2,3,1,1),(8,1,0,3,0,2),(2,0,0,0,1,2),\\(3,1,2,3,0,0),(6,0,0,0,1,1),(2,2,1,0,2,1),(3,1,2,3,2,1),(3,0,2,3,0,1),(2,0,2,0,2,1),\\(6,0,2,0,2,1),(8,0,2,3,0,2),(2,1,2,0,1,1),(6,1,2,0,1,1),(8,0,1,3,0,0),(3,0,2,3,2,1),\\(3,1,2,3,1,1),(3,1,1,3,1,0),(8,1,1,3,1,0),(3,1,0,3,0,2),(8,1,2,3,0,0),(6,2,1,0,2,1),\\(0,1,2,0,1,1),(3,0,0,3,0,1),(8,2,1,3,2,1),(8,0,2,3,2,1)\};$
\\$\bullet$ $\left[\begin{array}{ccc|cc}
1 & 0 & x & 0 & y \\
0 & 1 & z & 0 & u \\
0 & 0 & v & 1 & w
\end{array}\right]$ with $(x,y,z,u,v,w)\in\{(6,1,3,1,3,0),(8,1,3,1,6,1),(8,1,3,2,6,0),\\(6,1,3,0,3,1),(8,1,3,0,6,1)\};$
\\$\bullet$ $\left[\begin{array}{ccc|cc}
1 & x & y & 0 & 0 \\
0 & z & u & 1 & 0 \\
0 & v & w & 0 & 1
\end{array}\right]$ with $(x,y,z,u,v,w)\in\{(6,5,3,6,3,3),(6,0,3,3,6,3),(0,7,3,3,3,6),\\(3,6,0,3,3,0),(5,0,0,3,3,0),(6,7,0,6,6,6),(3,6,0,3,3,3),(0,6,3,0,3,6),(4,6,6,3,0,6),\\(7,0,6,3,0,6)\};$
\item[$\blacklozenge$] 19 Euclidean $\mathbb{Z}_9\mathbb{Z}_3$-LCD codes of Lee distance 4 whose generator matrices  are listed below:
\\$\bullet$ $[x~y~z~|~1~u]$ with $(x,y,z,u)\in\{(0,3,0,2),(0,3,0,0),(1,1,0,2)\};$
\\$\bullet$ $\left[\begin{array}{ccc|cc}
x&y&z&1&0\\
u&v&w&0&1
\end{array}\right]$ with $(x,y,z,u,v,w)\in\{(6,0,6,0,3,0),(0,0,6,3,3,6),(0,0,3,6,0,6),\\(0,0,3,3,0,0)\};$
\\$\bullet$ $\left[\begin{array}{ccc|cc}
1&x&y&0&z\\
0&u&v&1&w
\end{array}\right]$ with $(x,y,z,u,v,w)\in\{(7,0,1,0,6,0),(2,0,1,3,0,2),(2,3,1,0,6,0),\\(2,6,1,3,0,0),(2,0,1,3,0,0),(2,3,0,0,6,0),(2,6,0,3,0,0),(0,5,1,3,6,2),(6,8,0,3,3,0),\\(0,5,1,3,6,0),(7,0,1,0,6,2)\};$
\\$\bullet$ $\left[\begin{array}{ccc|cc}
1 & 2 & 3 & 0 & 0 \\
0 & 3 & 0 & 1 & 0 \\
0 & 0 & 3 & 0 & 1
\end{array}\right];$ 

\item[$\blacklozenge$] 10 Euclidean $\mathbb{Z}_9\mathbb{Z}_3$-LCD codes of Lee distance 5 whose generator matrices are listed below:
\\$\bullet$ $\left[\begin{array}{ccc|cc}
x&y&z&1&0\\
u&v&w&0&1
\end{array}\right]$ with $(x,y,z,u,v,w)\in\{(3,3,3,3,6,3),(3,0,3,3,0,6),(3,0,6,3,6,6)\};$
\\$\bullet$ $\left[\begin{array}{ccc|cc}
1&x&y&0&z\\
0&u&v&1&w
\end{array}\right]$ with $(x,y,z,u,v,w)\in\{(2,6,1,3,0,2),(2,3,1,0,6,1),(6,8,1,3,3,0),\\(6,8,1,3,3,2),(2,3,0,0,6,1),(2,6,0,3,0,1),(6,8,0,3,3,1)\};$

\item[$\blacklozenge$] 3 Euclidean $\mathbb{Z}_9\mathbb{Z}_3$-LCD codes of Lee distance 6 whose generator matrices are listed below:
\\$\bullet$ $[1~6~7~|~x~y]$ with $(x,y)\in\{(0,0),(0,2),(1,2)\};$

\item[$\blacklozenge$] 3 Euclidean $\mathbb{Z}_9\mathbb{Z}_3$-LCD codes of Lee distance 7 with generator matrices $[6~6~0~|~1~0],$ $\left[\begin{array}{ccc|cc}
3&0&6&1&0\\
6&6&6&0&1
\end{array}\right]$ and $\left[\begin{array}{ccc|cc}
3&6&0&1&0\\
3&0&3&0&1
\end{array}\right];$

\item[$\blacklozenge$] 1 Euclidean $\mathbb{Z}_9\mathbb{Z}_3$-LCD code of Lee distance 8 with a generator matrices $[6~6~0~|~1~2];$

\item[$\blacklozenge$] 1 Euclidean $\mathbb{Z}_9\mathbb{Z}_3$-LCD code of Lee distance 10 with a generator matrix $[6~6~3~|~1~0];$ and

\item[$\blacklozenge$] 1 Euclidean $\mathbb{Z}_9\mathbb{Z}_3$-LCD code of Lee distance 11 with a generator matrix $[6~6~3~|~1~2].$
\end{enumerate}\end{itemize}
\noindent \textbf{\large C.} \textbf{\large Classification of  Hermitian $\frac{\mathbb{F}_{4}[u]}{\langle u^2\rangle}\;\mathbb{F}_{4}$-LCD codes of block-length $(\mathpzc{a},\mathpzc{b}),$ where $\mathpzc{a} ,\mathpzc{b} \in \{1,2\}$}\\

Finally, we will classify all Hermitian $\frac{\mathbb{F}_{4}[u]}{\langle u^2\rangle}\;\mathbb{F}_{4}$-LCD codes of block-length $(\mathpzc{a},\mathpzc{b})$ up to monomial equivalence, where $\mathpzc{a},\mathpzc{b}\in\{1,2\}.$  To do this, we first recall the  Lee weight $w_L$ on $\mathbb{F}_{4}[u]/\langle u^2\rangle$ defined by Alahmadi \etal~\cite{Leeweight}. The Lee weight $w_L$ on $\mathbb{F}_{4}[u]/\langle u^2\rangle$ is defined as $w_L(a_0+a_1u)=w_H(a_0+a_1)+w_H(a_1)$
for all $a_0+a_1u\in \mathbb{F}_{4}[u]/\langle u^2\rangle.$   It can be extended to the Lee weight on $\frac{\mathbb{F}_{4}[u]}{\langle u^2\rangle}^{\mathpzc{a}} \oplus \mathbb{F}_{4}^{\mathpzc{b}}$ as $w_L(m)=\sum\limits_{i=1}^{\mathpzc{a}}w_{L}(c_i)+\sum\limits_{j=1}^{\mathpzc{b}}w_H(d_j)$ for all $m=(c_1,c_2,\ldots,c_{\mathpzc{a}}|d_1,d_2,\ldots,d_{\mathpzc{b}}) \in \frac{\mathbb{F}_{4}[u]}{\langle u^2\rangle}^{\mathpzc{a}} \oplus \mathbb{F}_{4}^{\mathpzc{b}}.$ The Lee distance of a $\frac{\mathbb{F}_{4}[u]}{\langle u^2\rangle}\;\mathbb{F}_{4}$-linear code $\mathtt{C}$ of block-length $(\mathpzc{a},\mathpzc{b})$ is defined as $d_L(\mathtt{C})=\min\{w_L(c): c (\neq 0)\in \mathtt{C}\}.$ 

\begin{itemize}\item[I.] There are precisely 5 monomially inequivalent non-zero Hermitian $\frac{\mathbb{F}_4[u]}{\langle u^2\rangle}\;\mathbb{F}_4$-LCD codes of block-length $(1,1).$ Among these codes, there are
\vspace{-2mm}\begin{enumerate}
\item[$\blacklozenge$] 3 Hermitian $\frac{\mathbb{F}_4[u]}{\langle u^2\rangle}\;\mathbb{F}_4$-LCD codes of Lee distance 1 with generator matrices $[1~|~0],$ $[0~|~1]$ and $\left[\begin{array}{c|c}
1&0 \\
0&1
\end{array}\right];$
\item[$\blacklozenge$] 1 Hermitian $\frac{\mathbb{F}_4[u]}{\langle u^2\rangle}\;\mathbb{F}_4$-LCD code of Lee distance 2 with a generator matrix $[1~|~1+\xi];$ and
\item[$\blacklozenge$] 1 Hermitian $\frac{\mathbb{F}_4[u]}{\langle u^2\rangle}\;\mathbb{F}_4$-LCD code of Lee distance 3 with a generator matrix $[(1+\xi)u~|~1].$ 
\end{enumerate}
\item[II.] There are precisely 11 monomially inequivalent non-zero Hermitian $\frac{\mathbb{F}_4[u]}{\langle u^2\rangle}\;\mathbb{F}_4$-LCD codes of block-length $(1,2).$ Among these codes, there are
\vspace{-2mm}\begin{enumerate}
\item[$\blacklozenge$] 7 Hermitian $\frac{\mathbb{F}_4[u]}{\langle u^2\rangle}\;\mathbb{F}_4$-LCD codes of Lee distance 1 with generator matrices  
$[ 1~|~0~0],$ 
 $[ 0~|~1~0],$  
$\left[\begin{array}{c|cc}
1&0&1+\xi \\
0&1&0
\end{array}\right],$ 
$\left[\begin{array}{c|cc}
1&0&0 \\
0&1&0
\end{array}\right],$ $\left[\begin{array}{c|cc}
0&1&0 \\
0&0&1
\end{array}\right],$ $\left[\begin{array}{c|cc}
\xi u&1&0 \\
0&0&1
\end{array}\right]$ and $\left[\begin{array}{c|c}
     1&0  \\
     0&I_2
\end{array}\right];$
\item[$\blacklozenge$] 3 Hermitian $\frac{\mathbb{F}_4[u]}{\langle u^2\rangle}\;\mathbb{F}_4$-LCD codes of Lee distance 2 with  generator matrices $[ 1~|~0~\xi],$ $[ 1~|~\xi~\xi]$ and $\left[\begin{array}{c|cc}
u&1&0 \\
u&0&1
\end{array}\right];$ and 
\item[$\blacklozenge$] 1 Hermitian $\frac{\mathbb{F}_4[u]}{\langle u^2\rangle}\;\mathbb{F}_4$-LCD code of Lee distance 3 with a generator matrix $[u
~|~1~0].$
\end{enumerate}
\item[III.] There are precisely 15 monomially inequivalent non-zero Hermitian $\frac{\mathbb{F}_4[u]}{\langle u^2\rangle}\;\mathbb{F}_4$-LCD codes of block-length $(2,1).$ Among these codes, there are
\vspace{-2mm}\begin{enumerate}
\item[$\blacklozenge$] 8 Hermitian $\frac{\mathbb{F}_4[u]}{\langle u^2\rangle}\;\mathbb{F}_4$-LCD codes of Lee distance 1 with generator matrices   
$[1~0~|~0],$  $[0~0~|~1],$  
$\left[\begin{array}{cc|c}
1&0&0 \\
0&1&0
\end{array}\right],$ $\left[\begin{array}{cc|c}
1&0&0 \\
0&1&1
\end{array}\right],$ $\left[\begin{array}{cc|c}
1&(1+\xi)u&0 \\
0&0&1
\end{array}\right],$ $\left[\begin{array}{cc|c}
1&0&0 \\
0&u&1
\end{array}\right],$ $\left[\begin{array}{cc|c}
1&0&0 \\
0&0&1
\end{array}\right]$  and $\left[\begin{array}{c|cc}
     I_2&0  \\
     0&1
\end{array}\right];$  
\item[$\blacklozenge$] 5 Hermitian $\frac{\mathbb{F}_4[u]}{\langle u^2\rangle}\;\mathbb{F}_4$-LCD code of Lee distance 2 with generator matrices $[ 1~u~|~0],$ $[1~u~|~1+\xi],$ $[1~0~|~1+\xi],$ $\left[\begin{array}{cc|c}
1&u&0 \\
0&u&1
\end{array}\right]$ and $\left[\begin{array}{cc|c}
1&0&1 \\
0&1&\xi
\end{array}\right];$
\item[$\blacklozenge$] 1 Hermitian $\frac{\mathbb{F}_4[u]}{\langle u^2\rangle}\;\mathbb{F}_4$-LCD code of Lee distance 3 with a generator matrix $[\xi u~0~|~1];$ and
\item[$\blacklozenge$] 1 Hermitian $\frac{\mathbb{F}_4[u]}{\langle u^2\rangle}\;\mathbb{F}_4$-LCD code of Lee distance 5 with a generator matrix $[u~u~|~1].$
\end{enumerate}
\item[IV.] There are precisely 43 monomially inequivalent non-zero Hermitian $\frac{\mathbb{F}_4[u]}{\langle u^2\rangle}\;\mathbb{F}_4$-LCD codes of block-length $(2,2).$ Among these codes, there are
\vspace{-2mm}\begin{enumerate}
\item[$\blacklozenge$] 22 Hermitian $\frac{\mathbb{F}_4[u]}{\langle u^2\rangle}\;\mathbb{F}_4$-LCD codes of Lee distance 1 whose generator matrices are listed below:
\\$\bullet$ $[1~0~|~0~0],$ $[0~0~|~1~0]$ and $\left[\begin{array}{c|c}
     I_2&0  \\
     0&I_2
\end{array}\right];$
\\$\bullet$  $\left[\begin{array}{cc|cc}
x&y&1&0 \\
z&v&0&1 
\end{array}\right]$ with $(x,y,z,v)\in\{(0,0,0,0),((1+\xi)u,u,0,0),(0,u,0,0)\};$
\\$\bullet$ $\left[\begin{array}{cc|cc}
1&x&y&z \\
0&u'&v&w 
\end{array}\right]$ with $(x,y,z,u',v,w)\in\{(0,0,0,u,1,0),(0,0,1+\xi,0,1,0),(0,0,0,0,1,0),\\((1+\xi)u,0,0,0,1,0),((1+\xi)u,0,1+\xi,0,1,0),(0,0,0,1,0,0),(0,0,\xi,1,0,0),(0,0,0,1,1+\xi,1)\};$
\\$\bullet$ $\left[\begin{array}{cc|cc}
1&0&0&x \\
0&1&0&y \\
0&0&1&z
\end{array}\right]$  with $(x,y,z)\in\{(0,0,0),(1,1+\xi,0),(0,1+\xi,0)\};$
\\$\bullet$ $\left[\begin{array}{cc|cc}
1&x&0&0 \\
0&y&1&0 \\
0&z&0&1
\end{array}\right]$ with $(x,y,z)\in\{(0,(1+\xi)u,0),((1+\xi)u,0,0),(0,u,u),(0,0,0),((1+\xi)u,\xi u,0)\};$ 
\item[$\blacklozenge$] 16 Hermitian $\frac{\mathbb{F}_4[u]}{\langle u^2\rangle}\;\mathbb{F}_4$-LCD codes of Lee distance 2 whose generator matrices are listed below:
\\$\bullet$ $[1~x~|~y~z]$ with $(x,y,z)\in\{(u,0,\xi),(0,0,\xi),(u,0,0),(u,\xi,\xi),(0,\xi,\xi)\};$
\\$\bullet$ $\left[\begin{array}{cc|cc}
x&y&1&0 \\
z&v&0&1 
\end{array}\right]$ with $(x,y,z,v)\in\{(\xi u,\xi u,u,u),(\xi u,0,u,0)\};$
\\$\bullet$ $\left[\begin{array}{cc|cc}
1&x&y&z \\
0&u'&v&w 
\end{array}\right]$ with $(x,y,z,u',v,w)\in\{(0,0,1+\xi,u,1,0),(0,0,\xi,1,\xi,0),(0,\xi,0,1,1,\xi),\\(0,\xi,\xi,1,1,1+\xi),(0,1+\xi,0,1,\xi,0),(u,0,0,u,1,0),(0,1,\xi,1,1+\xi,1),(u,0,1+\xi,u,1,0)\};$
\\$\bullet$ $\left[\begin{array}{cc|cc}
1&\xi u&0&0 \\
0&u&1&0 \\
0&u&0&1
\end{array}\right];$
\item[$\blacklozenge$] 3 Hermitian $\frac{\mathbb{F}_4[u]}{\langle u^2\rangle}\;\mathbb{F}_4$-LCD codes of Lee distance 3 with generator matrices  $[\xi u~0~|~1~0],$  $\left[\begin{array}{cc|cc}
\xi u&0&1&0 \\
0&u&0&1 
\end{array}\right]$ and $\left[\begin{array}{cc|cc}
u&(1+\xi)u&1&0 \\
(1+\xi)u&0&0&1 
\end{array}\right];$
\item[$\blacklozenge$] 1 Hermitian $\frac{\mathbb{F}_4[u]}{\langle u^2\rangle}\;\mathbb{F}_4$-LCD code of Lee distance 4 with a generator matrix $\left[\begin{array}{cc|cc}
\xi u&\xi u&1&0 \\
\xi u&(1+\xi)u&0&1 
\end{array}\right];$ and
\item[$\blacklozenge$] 1 Hermitian $\frac{\mathbb{F}_4[u]}{\langle u^2\rangle}\;\mathbb{F}_4$-LCD code of Lee distance 5 with a generator matrix $[u~u~|~1~0].$

\end{enumerate}

\end{itemize}

\end{document}